\pdfoutput=1
\documentclass[acmsmall,screen]{acmart}
\setcopyright{none}
\usepackage[utf8]{inputenc} 
\usepackage{amsmath, amsthm, amssymb}
\usepackage{thmtools}
\usepackage[english]{babel}
\usepackage{pdfpages}
\usepackage{paralist}
\usepackage{graphicx}
\usepackage{subcaption}
\usepackage[labelfont=bf]{caption}
\usepackage{longtable}
\usepackage{booktabs}
\usepackage{tabu}
\usepackage[referable]{threeparttablex}
\usepackage{varwidth}
\usepackage{multirow}
\usepackage{wrapfig}

\usepackage{pifont}

\usepackage[ruled, linesnumbered]{algorithm2e}
\usepackage{parskip}

\usepackage{hyperref}
\usepackage{thm-restate}
\usepackage[capitalise]{cleveref}

\crefname{figure}{Figure}{Figure}
\crefformat{footnote}{#2\footnotemark[#1]#3}
\usepackage{tikz}
\usepackage{comment}

\usepackage{todonotes}
\usepackage{tikz}
\usetikzlibrary{arrows,automata,shapes,decorations,decorations.markings,calc, matrix,decorations.pathmorphing, patterns,backgrounds}

\usepackage{bm}
\usepackage{pgfplots}
\usepackage{enumerate,paralist}
\usepackage{pbox}

\usepackage{appendix}
\usepackage{calc}
\usepackage{fp}
\usepackage{multirow}
\usepackage{pbox}

\usepackage{etoolbox}

\theoremstyle{plain}
\newtheorem{remark}{Remark}

\DeclareMathAlphabet{\mathpzc}{OT1}{pzc}{m}{it}

\usepackage{graphicx} 

\mathchardef\mhyphen="2D % Define a "math hyphen"

\newcommand{\Ints}{\mathbb{Z}}
\newcommand{\ov}{\overline}

\newcommand{\Width}{\mathsf{width}}
\newcommand{\Confl}[2]{#1 \Join #2}
\newcommand{\Obs}{\mathsf{O}}

\newcommand{\Trace}{t}

\newcommand{\SysAcquires}{\mathcal{L}^A}
\newcommand{\SysReleases}{\mathcal{L}^R}
\newcommand{\SysReads}{\mathcal{R}}

\newcommand{\SysWrites}{\mathcal{W}}
\newcommand{\Vars}{\mathcal{V}}
\newcommand{\Read}{r}
\newcommand{\Write}{w}
\newcommand{\AtomicWrite}{w}
\newcommand{\AtomicRead}{r}
\newcommand{\RespectPO}[2]{R_{#1}(#2)}
\newcommand{\Event}{e}

\newcommand{\SHB}{\mathsf{SHB}}

\newcommand{\DC}{\mathsf{DC}}
\newcommand{\HB}{\mathsf{HB}}
\newcommand{\CP}{\mathsf{CP}}
\newcommand{\WCP}{\mathsf{WCP}}

\newcommand{\RPast}{\mathsf{RCone}}

\newcommand{\RaceDecision}{\mathsf{RaceDecision}}
\newcommand{\RaceFunction}{\mathsf{M2}}

\newcommand{\ForkT}{\mathsf{fork}}
\newcommand{\JoinT}{\mathsf{join}}
\newcommand{\Events}[1]{\SysEvents(#1)}
\newcommand{\Reads}[1]{\SysReads(#1)}
\newcommand{\Writes}[1]{\SysWrites(#1)}

\newcommand{\Acquires}[1]{\SysAcquires(#1)}
\newcommand{\Releases}[1]{\SysReleases(#1)}
\newcommand{\Project}{|}

\newcommand{\SeqTrace}{\tau}
\newcommand{\MayRaces}{\mathcal{C}}

\newcommand{\FalseNegativesBound}{\mathsf{FN}}

\newcommand{\True}{\mathsf{True}}
\newcommand{\False}{\mathsf{False}}

\newcommand{\System}{\mathcal{P}}
\newcommand{\Process}{p}
\newcommand{\Proc}[1]{\mathsf{p}(#1)}
\newcommand{\MinMaxAlgo}{\mathsf{MaxMin}}

\newcommand{\Closure}{\mathsf{Closure}}
\newcommand{\From}{\mathsf{After}}
\newcommand{\To}{\mathsf{Before}}

\newcommand{\Push}{\mathsf{push}}

\newcommand{\Successor}{\mathsf{successor}}
\newcommand{\Predecessor}{\mathsf{predecessor}}

\newcommand{\ObsClosure}{\mathsf{ObsClosure}}
\newcommand{\LockClosure}{\mathsf{LockClosure}}

\newcommand{\Pop}{\mathsf{pop}}

\newcommand{\Update}{\mathsf{update}}
\newcommand{\Min}{\min}
\newcommand{\ArgMin}{\arg\mathsf{leq}}
\newcommand{\ArgQuery}{\arg\min}
\newcommand{\Query}{\mathsf{query}}
\newcommand{\Insert}{\mathsf{insert}}

\newcommand*{\Match}[2][]{
\ifx\\#1\\%
  \mathsf{match}(#2)
\else
  \mathsf{match}_{#1}(#2)
\fi
}

\newcommand{\Implies}{\Rightarrow}
\newcommand{\Unordered}[3]{#1\parallel_{#2} #3}
\newcommand{\Ordered}[3]{#1 \not \parallel_{#2} #3}

\newcommand{\OpenAcquires}{\mathsf{OpenAcqs}}
\newcommand{\Refines}{\sqsubseteq}

\newcommand{\Globals}{\mathcal{G}}
\newcommand{\Locks}{\mathcal{L}}

\newcommand{\CNode}[1]{\langle #1\rangle}

\newcommand{\Outgoing}{\mathsf{Out}}
\newcommand{\DS}{\mathsf{DS}}
\newcommand{\Fenwick}{\mathsf{FenwickTree}}

\newcommand{\TO}{\mathsf{PO}}

\newcommand{\Acquire}{\mathsf{acq}}
\newcommand{\Release}{\mathsf{rel}}

\newcommand{\Location}[1]{\mathsf{loc}(#1)}

\newcommand{\SysEvents}{\mathcal{E}}

\newcommand{\Path}{\rightsquigarrow}
\newcommand{\Queue}{\mathcal{Q}}

\newcommand{\Observation}{\mathcal{O}}

\newcommand{\LastFlow}{\mathcal{F}}

\newcommand{\InsertAndClose}{\mathsf{InsertAndClose}}
\newcommand{\Races}{\mathcal{Z}}
\newcommand{\Time}{\mathsf{Time}}
\newcommand{\PR}{\mathsf{Races}}

\usetikzlibrary{arrows.meta,calc,decorations.markings,math,arrows.meta}

\preto\tabular{\setcounter{magicrownumbers}{0}}
\newcounter{magicrownumbers}

\pgfdeclarelayer{bg}    % declare background layer
\pgfsetlayers{bg,main}  % set the order of the layers (main is the standard layer)

\def \darkred {black!20!red}
\def \darkgreen {black!20!green}
\SetKw{Continue}{continue}

\SetCommentSty{mycommfont}

\makeatletter
\newcommand*{\centerfloat}{%
  \parindent \z@
  \leftskip \z@ \@plus 1fil \@minus \textwidth
  \rightskip\leftskip
  \parfillskip \z@skip}
\makeatother

\usepackage{booktabs}   %% For formal tables:
\usepackage{subcaption} %% For complex figures with subfigures/subcaptions
\begin{document}

\newskip\smallskipamount \smallskipamount=1pt plus 0.5pt minus 0.5pt
\title{Fast, Sound and Effectively Complete Dynamic Race Prediction}

\author{Andreas Pavlogiannis}
%\orcid{nnnn-nnnn-nnnn-nnnn}             %% \orcid is optional
\affiliation{
%\position{Position1}
%\department{Department1}              %% \department is recommended
\institution{Aarhus University}            %% \institution is required
\streetaddress{IT-parken, Aabogade 34}
\city{Aarhus N}
\postcode{DK-8200}
\country{Denmark}                    %% \country is recommended
}
\email{pavlogiannis@cs.au.dk}          %% \email is recommended

 \begin{abstract}
Writing concurrent programs is highly error-prone due to the nondeterminism in interprocess communication.
The most reliable indicators of errors in concurrency are \emph{data races}, which are accesses to a shared resource that can be executed concurrently.
We study the problem of predicting data races in lock-based concurrent programs.
The input consists of a concurrent trace $\Trace$, and the task is to determine all pairs of events of $\Trace$ that constitute a data race.
The problem lies at the heart of concurrent verification and has been extensively studied for over three decades.
However, existing polynomial-time sound techniques are highly incomplete and can miss simple races.

In this work we develop $\RaceFunction$: a new polynomial-time algorithm for this problem, which has no false positives.
In addition, our algorithm is \emph{complete} for input traces that consist of two processes,
i.e., it provably detects \emph{all} races in the trace.
We also develop sufficient criteria for detecting completeness \emph{dynamically} in cases of more than two processes.
We make an experimental evaluation of our algorithm on a challenging set of benchmarks taken from recent literature on the topic.
Our algorithm soundly reports \emph{hundreds} of real races, many of which are missed by existing methods.
In addition, using our dynamic completeness criteria, $\RaceFunction$ concludes that it has detected \emph{all} races in the benchmark set, hence the reports are both sound and complete.
%misses \emph{at most one} race in the whole benchmark set.
%In addition, our technique detects \emph{all} racy memory locations of the benchmark set.
Finally, its running times are comparable, and often smaller than the theoretically fastest, yet highly incomplete, existing methods.
To our knowledge, $\RaceFunction$ is the first sound algorithm that achieves such a level of performance on both running time and completeness of the reported races.
\end{abstract}



  %% \keywords are mandatory in final camera-ready submission

\maketitle

\section{Introduction}\label{SEC:INTRO}

\smallskip\noindent{\bf Verification of concurrent programs.}
Writing concurrent software is notoriously hard due to the inherent nondeterminism in the way that accesses to shared resources are scheduled.
Accounting for all possible nondeterministic choices is hard, even to experienced developers.
This makes the development of concurrent software prone to concurrency bugs~\cite{Lu08,Shi10}, i.e., bugs that are present only in a few among the (possibly exponentially) many executions of the program.
Since developers have no control over the scheduler, concurrency bugs are also hard to reproduce by testing
(often categorized as Heisenbugs~\cite{Gray85,Musuvathi08}).
Consequently, testing alone is considered an ineffective approach for detecting bugs in concurrent programs.
To circumvent this difficulty, testing techniques are often combined with model checking.
First, a testing phase produces a set of concrete program executions.
Then, a verification phase makes a formal treatment of these executions and identifies whether there exist other ``neighboring'' executions that are not present in the test set but
(i)~constitute valid executions of the program and
(ii)~manifest a bug.
Hence, even though the scheduler might ``hide'' a bug in the test set, this bug can be effectively caught by formal techniques applied on the test set.

\smallskip\noindent{\bf Data races.}
Two events $(\Event_1, \Event_2)$ of a concurrent program are called \emph{conflicting} if 
they access the same shared resource (e.g., the same global variable $x$) and
at least one of them modifies the resource (e.g., writes to $x$).
A \emph{data race} is typically defined as a conflicting pair $(\Event_1, \Event_2)$ that can be executed concurrently~\cite{Helmbold91,OCallahan03,Flanagan09,Bond10}.
Data races are the prime suspects of erroneous behavior, and there have been significant efforts spanning across several decades towards detecting data races efficiently, starting with seminal papers found in~\cite{Schonberg89,Helmbold91,Dinning91,Savage97}.

\smallskip\noindent{\bf Dynamic race detection.}
Dynamic algorithms for race detection operate on a single execution (i.e., a \emph{trace}) of the concurrent program, 
and their task is to identify pairs of events of the trace that constitute a race, even though the race might not be manifested in the input trace.
Dynamic race detection is a popular technique that combines testing with formal reasoning.
Existing dynamic algorithms typically fall into one of the following three categories.
%\begin{compactenum}

\noindent{\em Lockset-based techniques}~\cite{Dinning91,Savage97,Elmas07} report races by comparing the sets of locks which guard conflicting data accesses.
This approach typically reports spurious races, as data accesses protected by different locks can nevertheless be separated by other control-flow and data dependencies, and thus not constitute a race.

\noindent{\em Exhaustive predictive-runtime techniques}~\cite{Savage97,Koushik05,Chen07,Said11,Huang14} report races by exploring all possible valid reorderings of the input trace.
These techniques typically rely on SAT/SMT solvers and are sound and complete in theory; however,
as there are exponentially many valid reorderings, they have \emph{exponential complexity}.
In practice, completeness is traded for runtime, by using windowing techniques which slice the input trace into small fragments and analyze each fragment separately.

\noindent{\em Partial-order-based techniques}~are probably the most well-known and widely-used.
The underlying principle is to construct a partial order $P$ on the events of the input trace.
Afterwards, a race is reported between a pair of events if the two events are unordered by $P$.
These techniques are usually efficient, as constructing the partial order typically requires polynomial time.
However, in order for $P$ to admit a linearization to a valid witness trace that exposes the race,
$P$ enforces many arbitrary orderings between events.
These arbitrary orderings often result in an ordering between the events of an actual race,
and thus $P$ misses the race.

Most of the above techniques are based on Lamport's \emph{happens-before} ($\HB$) partial order~\cite{Lamport78}
which is implemented in various tools~\cite{Schonberg89,Christiaens01,Pozniansky03,Yu05,Flanagan09,Bond10}.
As $\HB$ is highly incomplete, there have been several efforts for constructing weaker partial orders that are efficiently computable,
such as the \emph{causally-precedes} partial order $\CP$~\cite{Smaragdakis12}.
%On the downside, $\CP$ is only weakly sound, meaning that it creates false positives when the input trace contains a deadlock.
Partial-order techniques recently led to important advances in predictive race detection, based on the \emph{weakly-causally-precedes} $\WCP$~\cite{Kini17}, \emph{schedulably-happens-before} $\SHB$~\cite{Mathur18} and \emph{doesn't-commute} $\DC$~\cite{Roemer18} partial orders.
%To our knowledge, the most recent algorithmic development in this direction was made in~\cite{Kini17}, where the $\CP$ partial order was relaxed to
%\emph{weakly-causally-precedes} ($\WCP$).
%The $\WCP$ partial order is computable in linear time and detects races missed by $\CP$ and $\HB$, while remaining weakly sound.
%In the next section 
We next discuss these approaches in more detail and outline the motivation behind our work.
%In the next section we discuss the completeness and soundness issues of such partial-order based methods in more detail, and outline the motivation behind our work.
%\end{compactenum}

\subsection{Motivating Examples}\label{subsec:motivating_examples}

We illustrate the motivation behind our work with a few simple examples (\cref{fig:motivating}) which highlight some completeness issues that the existing approaches based on $\HB$, $\WCP$ and $\DC$ partial orders suffer from.
We focus on single races here, in which case $\SHB$ is subsumed by $\HB$.
We remark that we focus on polynomial-time, sound methods here, and hence we do not consider unsound techniques (e.g., lockset-based~\cite{Savage97})
or techniques that rely on SAT/SMT solvers and are thus not polynomial time (e.g.,~\cite{Huang14}).
In each example, we use the notation $\SeqTrace_i$ to refer to the local trace of the $i$-th process,
and $\Event_j$ to refer to the $j$-th event in the concurrent trace.
We note that the underlying memory model is sequentially consistent, i.e.,
in every trace, a read event observes the value of the last write event that writes to the location read by the read event.

To develop some context, we briefly outline how each of these techniques works by ordering events of the input trace.
We refer to \cref{sec:hb_wcp_dc} for the formal definitions.
In all cases, events that belong to the same process are always totally ordered according to their order in the input trace.
\begin{compactenum}
\item The $\HB$ and $\WCP$ techniques operate in a similar manner. They perform a single pass of $\Trace$ and construct a partial order $\leq_{\HB}$ (resp., $\leq_{\WCP}$).
A race $(\Event, \Event')$ is reported if $\Event, \Event'$ are conflicting and $\Event\not \leq_{\HB}\Event'$ (resp., $\Event\not \leq_{\WCP}\Event'$),
i.e., the two events are unordered by the respective partial order.
\item $\DC$ operates in three phases, which all have to succeed for $(\Event, \Event')$ to be reported as a race.
\begin{compactenum}
\item In Phase~1, a $\DC$ partial order is constructed, similarly to $\HB$ and $\WCP$.
If $\Event\leq_{\DC} \Event'$ then $(\Event, \Event')$ is reported as a non-race.
\item In Phase~2, a constraint graph $G$ is constructed which contains the $\DC$ orderings.
Then, more ordering constraints are inserted in $G$.
If $G$ becomes cyclic during this process, $(\Event, \Event')$ is reported as a non-race.
%\item In Phase~3, a witness trace $\Trace^*$ is attempted to be constructed from $G$.
If $\Trace^*$ fails to respect lock semantics, $(\Event, \Event')$ is reported as a non-race.
\end{compactenum}
\end{compactenum}

\begin{figure}[!h]
\begin{subfigure}[b]{0.4\textwidth}
\centering
\footnotesize
\def\rownumber{}
\begin{tabular}[b]{@{\makebox[1.2em][r]{\rownumber\space}} | l | l |}
\normalsize{$\mathbf{\SeqTrace_1}$} & \normalsize{$\mathbf{\SeqTrace_2}$}
\gdef\rownumber{\stepcounter{magicrownumbers}\arabic{magicrownumbers}} \\
\hline
$\Acquire(\ell)$ & \\
$\mathbf{\Write(x)}$ & \\
$\Release(\ell)$ & \\
& $\Acquire(\ell)$  \\
& $\Write(x)$  \\
& $\Release(\ell)$  \\
& $\mathbf{\Read(x)}$\\
\hline
\end{tabular}
\caption{Is $(\Event_2, \Event_7)$ a race?}
\label{subfig:motivating_incomplete1}
\end{subfigure}
%\begin{subfigure}[b]{0.25\textwidth}
%\centering
%\footnotesize
%\def\rownumber{}
%\begin{tabular}[b]{@{\makebox[1.2em][r]{\rownumber\space}} | l | l | l |}
%\normalsize{$\mathbf{\SeqTrace_1}$} & \normalsize{$\mathbf{\SeqTrace_2}$} & \normalsize{$\mathbf{\SeqTrace_3}$}
%\gdef\rownumber{\stepcounter{magicrownumbers}\arabic{magicrownumbers}} \\
%\hline
%$\mathbf{\Write(x)}$ & & \\
%$\Acquire(\ell_1)$ & & \\
%$\Write(y)$ & & \\
%$\Release(\ell_1)$ & &\\
%& $\Acquire(\ell_1)$ & \\
%& $\Read(y)$ & \\
%& $\Release(\ell_1)$ & \\
%& $\Acquire(\ell_2)$ & \\
%& $\Release(\ell_2)$ & \\
%& & $\Acquire(\ell_2)$ \\
%& & $\Release(\ell_2)$ \\
%& & $\mathbf{\Read(x)}$\\
%%$\Acquire(\ell_1)$ &  & \\
%%$\Write(x)$ &  & \\
%%$\Release(\ell_1)$ & &\\
%%& $\Acquire(\ell_1)$ & \\
%%& $\Read(x)$ & \\
%%& $\Release(\ell_1)$ &  \\
%%& $\Acquire(\ell_2)$ & \\
%%& $\Release(\ell_2)$ & \\
%%& & $\Acquire(\ell_2)$ \\
%%& & $\Read(x)$ \\
%%& & $\Release(\ell_2)$ \\
%\hline
%\end{tabular}
%\caption{}
%\label{subfig:motivating_incomplete2}
%\end{subfigure}
\begin{subfigure}[b]{0.4\textwidth}
\centering
\footnotesize
\def\rownumber{}
\begin{tabular}[b]{@{\makebox[1.2em][r]{\rownumber\space}} | l | l | l |}
\normalsize{$\mathbf{\SeqTrace_1}$} & \normalsize{$\mathbf{\SeqTrace_2}$} & \normalsize{$\mathbf{\SeqTrace_3}$}
\gdef\rownumber{\stepcounter{magicrownumbers}\arabic{magicrownumbers}} \\
\hline
$\Acquire(\ell_1)$ & & \\
$\mathbf{\Write(x)}$ &  & \\
$\Write(y)$ & & \\
$\Release(\ell_1)$ & & \\
& $\Acquire(\ell_1)$ & \\
& $\Acquire(\ell_2)$ & \\
& $\Write(z)$ & \\
& $\Release(\ell_2)$ & \\
& $\Write(y)$ & \\
& $\Release(\ell_1)$ & \\
& & $\Acquire(\ell_2)$ \\
& & $\Read(z)$ \\
& & $\Release(\ell_2)$\\
& & $\mathbf{\Write(x)}$\\
\hline
\end{tabular}
\caption{Is $(\Event_2, \Event_{14})$ a race?}
\label{subfig:motivating_incomplete2}
\end{subfigure}
\caption{
Examples in which $\HB$, $\WCP$ and $\DC$ are incomplete.
(\protect\subref{subfig:motivating_incomplete1})~A race $(\Event_2, \Event_7)$ missed by $\HB$, $\WCP$ and $\DC$.
(\protect\subref{subfig:motivating_incomplete2})~A race $(\Event_2, \Event_{14})$ missed by $\HB$, $\WCP$ and $\DC$ (in Phase 2).
}
\label{fig:motivating}
\end{figure}

\smallskip\noindent{\bf Incompleteness.}
Each of $\HB$, $\WCP$ and $\DC$ methods are incomplete
i.e., the input trace $\Trace$ can have arbitrarily many predictable races, however each of these methods falsely reports that there is no race in $\Trace$.
We present a couple of examples where $\HB$, $\WCP$ and $\DC$ fail to detect simple races.

\noindent{\em \cref{subfig:motivating_incomplete1}.}
There is a predictable race $(\Event_2, \Event_7)$.
$\HB$ defines $\Event_3\leq_{\HB} \Event_4$,
and thus $\Event_2\leq_{\HB} \Event_7$, hence missing the race.
Similarly, $\WCP$ (resp., $\DC$ ) defines $\Event_3\leq_{\WCP} \Event_5$ (resp., $\Event_3\leq_{\DC} \Event_5$) and thus $\Event_2\leq_{\WCP} \Event_7$ (resp., $\Event_2\leq_{\DC} \Event_7$), hence missing the race.
Intuitively, $\WCP$ and $\DC$ fail to swap the two critical sections because they contain the conflicting events $\Write(x)$.
Note that here $\DC$ fails in Phase~1.
However, $(\Event_2, \Event_7)$ is a true race that is detected by the techniques developed in this work, exposed by the witness trace
$\Trace^*= \Event_4, \Event_5, \Event_6, \Event_1, \Event_2, \Event_7$.

\noindent{\em \cref{subfig:motivating_incomplete2}.}
There is a predictable race $(\Event_2, \Event_{14})$.
$\HB$ defines  $\Event_4 \leq_{\HB} \Event_{5}$ and $\Event_8 \leq_{\HB} \Event_{11}$,
and thus $\Event_2\leq_{\HB} \Event_{14}$, hence missing the race.
Similarly, $\WCP$ defines $\Event_4\leq_{\WCP} \Event_{5}$ and $\Event_8 \leq_{\WCP}\Event_{12}$ and thus $\Event_2\leq_{\WCP} \Event_{14}$, hence missing the race.
Intuitively, $\WCP$ fails to swap the critical sections of $\SeqTrace_1$ and $\SeqTrace_2$ on $\ell_1$ because $\WCP$ is closed under composition with $\HB$, and in turn $\HB$ totally orders critical sections as in the input trace.
On the other hand, $\DC$ does not compose with $\HB$, and the only enforced orderings are $\Event_4\leq_{\DC}\Event_9$ and $\Event_8\leq_{\DC}\Event_{12}$.
Hence $\DC$ proceeds with Phase~2, where it constructs a constraint graph $G$. 
Since $\Event_{4}\leq_{\DC}\Event_9$ and $\Event_9$ belongs in a critical section on lock $\ell_1$ which is released by $\Event_4$, 
in order to not violate lock semantics, $G$ forces the ordering $\Event_4\Path\Event_5$.
In addition, $G$ forces the ordering $\Event_5\Path \Event_2$, since $\Event_2$ is the racy event and must appear last in the witness trace.
Note that this creates a cycle and hence $\DC$ fails in Phase~2.
However, $(\Event_2, \Event_{14})$ is a true race that is detected by the techniques developed in this work, exposed by the witness  trace
$\Trace^*=\Event_5,\Event_6, \Event_7, \Event_8, \Event_9, \Event_{10}, \Event_{11}, \Event_{12}, \Event_{13}, \Event_{1}, \Event_{2}, \Event_{14}$.

%\noindent{\em More examples.}
%We have seen in \cref{fig:motivating} examples where $\HB$, $\WCP$ and $\DC$ fail to detect simple races.
%In particular, in \cref{subfig:motivating_incomplete1} $\DC$ fails in its Phase~1,
%whereas in \cref{subfig:motivating_incomplete2} $\DC$ succeeds in its Phase~1 but fails in Phase~2.
%For an interesting case where $\DC$ succeeds both in its Phase~1 and Phase~2 but fails in Phase~3, we refer to \cref{sec:dc_example}.
%This is a difficult race presented by the authors of $\DC$ in~\cite{ROEMER18TR} to challenge their method in Phase~3.
%On the other hand, this race is also detected by the techniques we develop in our work.
%As the example is fairly complicated (it consists of $7$ processes and $32$ events), it is presented only in the appendix at the interest of the motivated reader.

\smallskip\noindent{\bf Algorithmic challenge.}
We have seen that state-of-the-art approaches fail to catch simple races.
Intuitively, the algorithmic challenge that underlies race detection is that of constructing a partial order $P$ with the following
properties.
\begin{compactenum}
\item $P$ is as weak as possible, so that a race $(\Event_i, \Event_j)$ remains unordered in $P$.
\item $P$ is efficiently linearizable to a valid trace that exposes the race.
\end{compactenum}
These two features are opposing each other, as the weaker the partial order,
the more linearizations it admits, and finding a valid one becomes harder.
Intuitively, existing techniques solve the efficiency problem by ordering conflicting accesses in $P$ in the same way as in $\Trace$.
As we have seen, this results in strong partial orders that miss simple races.

\smallskip\noindent{\bf Our approach.}
In this work we develop a new predictive technique for race detection.
At its core, our algorithm constructs partial orders that are much weaker than existing approaches (hence detecting more races), while these partial orders are efficiently (polynomial-time) linearizable to valid traces (hence the reported races are exposed efficiently).
To give a complete illustration of our insights, we use the more involved example in \cref{fig:motivating2}.

The task is to decide whether $(\Event_{10}, \Event_{19})$ is a predictable race of the input trace $\Trace$ (\cref{subfig:motivating_complete}).
To keep the presentation simple, we ignore the other data races that occur, which can be trivially avoided by inserting additional lock events. 
%To keep the presentation of this example simple, every write and read event to a variable $x_i$ is made atomic, by protecting it with a dedicated lock $\ell_{x_i}$,
%and hence no race occurs on $x_i$ (the corresponding event is denoted as $\AtomicWrite(x_i)/\AtomicRead(x_i)$).
Note that $\HB$, $\WCP$ and $\DC$ report no race in $\Trace$, as they all order $\Event_{11}\leq  \Event_{14}$.
In order to detect this race, we need to make some non-trivial reasoning about reordering certain events in $\Trace$.
Our reasoning can be summarized in the following steps.
\begin{compactenum}
\item If $(\Event_{10}, \Event_{19})$ is a race of $\Trace$, a witness trace $\Trace^*$ can be constructed in which both $\Event_{10}$ and $\Event_{19}$ are the last events.
Observe that $\Trace^*$ will not contain the $\Release(\ell)$ event $\Event_{11}$.
\item Since we ignore event $\Event_{11}$, that critical section of $\SeqTrace_1$ remains open in $\Trace^*$.
Hence the $\Release(\ell)$ event $\Event_{15}$ must be ordered before the $\Acquire(\ell)$ event $\Event_8$.
In addition, the $\AtomicWrite(x_2)$ event $\Event_2$ is observed by the $\AtomicRead(x_2)$ event $\Event_{17}$, hence $\Event_2$ must be ordered before $\Event_{17}$.
These constraints, together with the program order which requires events of each process to occur in the same order as in the input trace, are captured by the partial order shown in solid edges in \cref{subfig:motivating2_closed}.
Note that  several conflicting accesses to $x_1$, $x_3$ and $x_4$ are still unordered. 
How can we obtain a valid linearization?
First, we can infer a few more orderings.
\item The $\AtomicRead(x_4)$ event $\Event_9$ must observe the same write event as in $\Trace$.
Due to the previous step, the $\AtomicWrite(x_4)$ event $\Event_{14}$ now is ordered before $\Event_9$.
To avoid $\Event_{9}$ observing $\Event_{14}$, we perform an \emph{observation-closure} step, by ordering $\Event_{14}$ before the observation $\Event_5$ of $\Event_{9}$ (see dashed edge in \cref{subfig:motivating2_closed}).
\item Due to the previous step, the $\Acquire(\ell)$ event $\Event_{13}$ is now ordered before the $\Release(\ell)$ event $\Event_7$.  
In order to not violate lock semantics, the critical section of the second process must be ordered before the first critical section of the first process.
Hence we perform a \emph{lock-closure} step, by ordering the $\Release(\ell)$ event $\Event_{15}$ before the lock-acquire event $\Event_{4}$
(see dashed edge in \cref{subfig:motivating2_closed}).
\item At this point, no other closure step is performed, and the partial order is called \emph{trace-closed}.
Note that there still exist conflicting accesses to variables $x_1$ and $x_3$ which are pairwise unordered and quite distant, hence \emph{not every linearization} produces a valid trace, and a correct linearization is not obvious.
We observe that we can obtain a valid trace by starting from the beginning of $\SeqTrace_1$ and $\SeqTrace_2$, and execute the former \emph{maximally} and the latter \emph{minimally}, according to the partial order.
That is, we repeatedly execute $\SeqTrace_1$ until we reach an event that is preceded by an event of $\SeqTrace_2$,
and then execute $\SeqTrace_2$ only until an event of $\SeqTrace_1$ becomes enabled again.
This \emph{max-min linearization} produces a valid witness trace (see \cref{subfig:motivating_complete_witness}).
\end{compactenum}
%As a final remark, we note that our closure orderings are \emph{necessarily} present in any witness trace, and hence completeness was not sacrificed.

In this work we make the above insights formal.
We define the notion of trace-closed partial orders, which captures observation and lock-closure steps, and develop an efficient (polynomial-time) algorithm for computing the closure.
For two processes, we show that max-min linearizations \emph{always} produce valid traces, \emph{as long as} the partial order is trace-closed.
Hence, in this case, we have a sound and complete algorithm.
The case of three or more processes is more complicated, and our algorithm might eventually order some (but crucially, not all) conflicting events arbitrarily.
Although these choices might sacrifice completeness, the resulting partial orders are much weaker than before, so that complex races can still be exposed soundly by a max-min linearization.

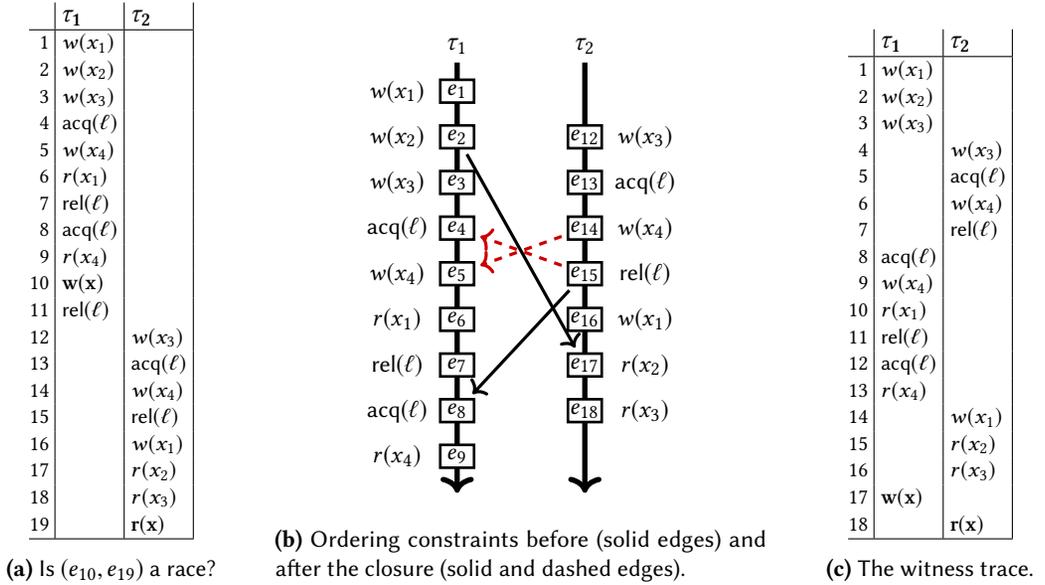
\begin{figure*}[!t]
\begin{subfigure}[b]{0.2\textwidth}
\centering
\footnotesize
\def\rownumber{}
\begin{tabular}[b]{@{\makebox[1.2em][r]{\rownumber\space}} | l | l |}
\normalsize{$\mathbf{\SeqTrace_1}$} & \normalsize{$\mathbf{\SeqTrace_2}$}
\gdef\rownumber{\stepcounter{magicrownumbers}\arabic{magicrownumbers}} \\
\hline
$\AtomicWrite(x_1)$ & \\
$\AtomicWrite(x_2)$ & \\
$\AtomicWrite(x_3)$ & \\
$\Acquire(\ell)$ & \\
$\AtomicWrite(x_4)$ & \\
$\AtomicRead(x_1)$ & \\
$\Release(\ell)$ & \\
$\Acquire(\ell)$ & \\
$\AtomicRead(x_4)$ & \\
$\mathbf{\Write(x)}$ & \\
$\Release(\ell)$ & \\
& $\AtomicWrite(x_3)$ \\
& $\Acquire(\ell)$ \\
& $\AtomicWrite(x_4)$ \\
& $\Release(\ell)$ \\
& $\AtomicWrite(x_1)$ \\
& $\AtomicRead(x_2)$ \\
& $\AtomicRead(x_3)$ \\
& $\mathbf{\Read(x)}$ \\
\hline
\end{tabular}
\caption{Is $(\Event_{10}, \Event_{19})$ a race?}
\label{subfig:motivating_complete}
\end{subfigure}
\qquad
\begin{subfigure}[b]{0.47\textwidth}
\centering
\small
\begin{tikzpicture}[thick,
pre/.style={<-,shorten >= 2pt, shorten <=2pt, very thick},
post/.style={->,shorten >= 2pt, shorten <=2pt,  very thick},
seqtrace/.style={->, line width=2},
und/.style={very thick, draw=gray},
event/.style={rectangle, minimum height=3mm, draw=black, fill=white, minimum width=4.5mm,   line width=1pt, inner sep=0, font={\small}},
virt/.style={circle,draw=black!50,fill=black!20, opacity=0}]

\newcommand{\xdisposition}{0}
\newcommand{\ydisposition}{0}
\newcommand{\xstep}{1.7}
\newcommand{\ystep}{0.6}
\newcommand{\ybias}{0.4}
\newcommand{\xbias}{0.8}
\newcommand{\xbiassmall}{0.6}

\node	[]		(t1a)	at	(\xdisposition + 0*\xstep, \ydisposition + 0*\ystep)	{\normalsize$\SeqTrace_1$};
\node	[]		(t1b)	at	(\xdisposition + 0*\xstep, \ydisposition + -10*\ystep)	{};
\node	[]		(t2a)	at	(\xdisposition + 1*\xstep, \ydisposition + 0*\ystep)	{\normalsize$\SeqTrace_2$};
\node	[]		(t2b)	at	(\xdisposition + 1*\xstep, \ydisposition + -10*\ystep)	{};

\draw[seqtrace] (t1a) to (t1b);
\draw[seqtrace] (t2a) to (t2b);

\node[event] (e11) at (\xdisposition + 0*\xstep, \ydisposition + -1*\ystep) {$\Event_{1}$};
\node[] (et11) at (\xdisposition + 0*\xstep - \xbias, \ydisposition + -1*\ystep) {$\AtomicWrite(x_1)$};
\node[event] (e12) at (\xdisposition + 0*\xstep, \ydisposition + -2*\ystep) {$\Event_{2}$};
\node[] (et12) at (\xdisposition + 0*\xstep - \xbias, \ydisposition + -2*\ystep) {$\AtomicWrite(x_2)$};
\node[event] (e13) at (\xdisposition + 0*\xstep, \ydisposition + -3*\ystep) {$\Event_{3}$};
\node[] (et13) at (\xdisposition + 0*\xstep - \xbias, \ydisposition + -3*\ystep) {$\AtomicWrite(x_3)$};
\node[event] (e14) at (\xdisposition + 0*\xstep, \ydisposition + -4*\ystep) {$\Event_{4}$};
\node[] (et14) at (\xdisposition + 0*\xstep - \xbias, \ydisposition + -4*\ystep) {$\Acquire(\ell)$};
\node[event] (e15) at (\xdisposition + 0*\xstep, \ydisposition + -5*\ystep) {$\Event_{5}$};
\node[] (et15) at (\xdisposition + 0*\xstep - \xbias, \ydisposition + -5*\ystep) {$\AtomicWrite(x_4)$};
\node[event] (e16) at (\xdisposition + 0*\xstep, \ydisposition + -6*\ystep) {$\Event_{6}$};
\node[] (et16) at (\xdisposition + 0*\xstep - \xbias, \ydisposition + -6*\ystep) {$\AtomicRead(x_1)$};
\node[event] (e17) at (\xdisposition + 0*\xstep, \ydisposition + -7*\ystep) {$\Event_{7}$};
\node[] (et17) at (\xdisposition + 0*\xstep - \xbias, \ydisposition + -7*\ystep) {$\Release(\ell)$};
\node[event] (e18) at (\xdisposition + 0*\xstep, \ydisposition + -8*\ystep) {$\Event_{8}$};
\node[] (et18) at (\xdisposition + 0*\xstep - \xbias, \ydisposition + -8*\ystep) {$\Acquire(\ell)$};
\node[event] (e19) at (\xdisposition + 0*\xstep, \ydisposition + -9*\ystep) {$\Event_{9}$};
\node[] (et19) at (\xdisposition + 0*\xstep - \xbias, \ydisposition + -9*\ystep) {$\AtomicRead(x_4)$};

\renewcommand{\ydisposition}{-\ystep}

\node[event] (e21) at (\xdisposition + 1*\xstep, \ydisposition + -1*\ystep) {$\Event_{12}$};
\node[] (et21) at (\xdisposition + 1*\xstep + \xbias, \ydisposition + -1*\ystep) {$\AtomicWrite(x_3)$};
\node[event] (e22) at (\xdisposition + 1*\xstep, \ydisposition + -2*\ystep) {$\Event_{13}$};
\node[] (et22) at (\xdisposition + 1*\xstep + \xbias, \ydisposition + -2*\ystep) {$\Acquire(\ell)$};
\node[event] (e23) at (\xdisposition + 1*\xstep, \ydisposition + -3*\ystep) {$\Event_{14}$};
\node[] (et23) at (\xdisposition + 1*\xstep + \xbias, \ydisposition + -3*\ystep) {$\AtomicWrite(x_4)$};
\node[event] (e24) at (\xdisposition + 1*\xstep, \ydisposition + -4*\ystep) {$\Event_{15}$};
\node[] (et24) at (\xdisposition + 1*\xstep + \xbias, \ydisposition + -4*\ystep) {$\Release(\ell)$};
\node[event] (e25) at (\xdisposition + 1*\xstep, \ydisposition + -5*\ystep) {$\Event_{16}$};
\node[] (et25) at (\xdisposition + 1*\xstep + \xbias, \ydisposition + -5*\ystep) {$\AtomicWrite(x_1)$};
\node[event] (e26) at (\xdisposition + 1*\xstep, \ydisposition + -6*\ystep) {$\Event_{17}$};
\node[] (et26) at (\xdisposition + 1*\xstep + \xbias, \ydisposition + -6*\ystep) {$\AtomicRead(x_2)$};
\node[event] (e27) at (\xdisposition + 1*\xstep, \ydisposition + -7*\ystep) {$\Event_{18}$};
\node[] (et27) at (\xdisposition + 1*\xstep + \xbias, \ydisposition + -7*\ystep) {$\AtomicRead(x_3)$};

\draw[post] (e12) to (e26);
\draw[post] (e24) to (e18);
\draw[post, dashed, draw=\darkred] (e23) to (e15);
\draw[post, dashed, draw=\darkred] (e24) to (e14);

\end{tikzpicture}
\caption{Ordering constraints before (solid edges) and after the closure (solid and dashed edges).
%An $\AtomicWrite(x_i)/\AtomicRead(x_i)$ event denotes an atomic write/read to variable $x_i$.
%E.e., $\AtomicWrite(x_i)$ denotes the sequence $\Acquire(\ell_{x_i}), \Write(x_i), \Release(\ell_{x_i})$.
}
\label{subfig:motivating2_closed}
\end{subfigure}
\qquad
\begin{subfigure}[b]{0.2\textwidth}
\centering
\footnotesize
\def\rownumber{}
\begin{tabular}[b]{@{\makebox[1.2em][r]{\rownumber\space}} | l | l |}
\normalsize{$\mathbf{\SeqTrace_1}$} & \normalsize{$\mathbf{\SeqTrace_2}$}
\gdef\rownumber{\stepcounter{magicrownumbers}\arabic{magicrownumbers}} \\
\hline
$\AtomicWrite(x_1)$ & \\
$\AtomicWrite(x_2)$ & \\
$\AtomicWrite(x_3)$ & \\
& $\AtomicWrite(x_3)$ \\
& $\Acquire(\ell)$ \\
& $\AtomicWrite(x_4)$ \\
& $\Release(\ell)$ \\
$\Acquire(\ell)$ & \\
$\AtomicWrite(x_4)$ & \\
$\AtomicRead(x_1)$ & \\
$\Release(\ell)$ & \\
$\Acquire(\ell)$ & \\
$\AtomicRead(x_4)$ & \\
& $\AtomicWrite(x_1)$ \\
& $\AtomicRead(x_2)$ \\
& $\AtomicRead(x_3)$ \\
$\mathbf{\Write(x)}$ & \\
& $\mathbf{\Read(x)}$ \\
\hline
\end{tabular}
\caption{The witness trace.}
\label{subfig:motivating_complete_witness}
\end{subfigure}
\caption{
Example of a race that requires non-trivial reasoning about reorderings of the input trace.
}
\label{fig:motivating2}
\end{figure*}

\subsection{Our Contributions}\label{subsec:contributions}

In summary, the contributions of this work are as follows.

\smallskip\noindent{\bf A new algorithm for dynamic race detection.}
Our main contribution is a \emph{polynomial-time} and \emph{sound} algorithm for detecting predictable races present in the input trace.
In addition, our algorithm is \emph{complete} for input traces that consist of events of two processes.
First we study the \emph{decision problem}, that is, given an input trace $\Trace$ and a pair of events $(\Event_1, \Event_2)$ of $\Trace$,
decide whether the pair constitutes a data race of $\Trace$.
We present a sound algorithm for the problem that operates in $O(n^2\cdot \log n)$ time, where $n$ is the length of $\Trace$.
Since all data races can be computed by solving the decision problem for each of the $\binom{n}{2}$ event pairs, we obtain a sound algorithm for reporting \emph{all races} that requires $O(n^4\cdot \log n)$ time.
In all cases, if the input trace consists of events of two processes, our race reports are also complete.

Our techniques rely on a new notion of \emph{trace-closed partial orders}, which might be of independent interest.
Informally, a closed partial order wrt a trace $\Trace$ is a partial order over a subset of events of $\Trace$ that respects
(i)~the observation $\Write(x)$ of each read event $\Read(x)$ in $\Trace$, and
(ii)~the lock semantics.
We define \emph{max-min} linearizations of closed partial orders, and prove sufficient conditions under which a max-min linearization produces a valid trace.
Finally, we show that given a partial order of small width, its closure can be computed in $O(n^2\cdot \log n)$ time.
To this end, we develop a data structure $\DS$ for maintaining the incremental transitive closure of directed acyclic graphs of small width.
$\DS$ requires $O(n)$ initialization time, after which it supports edge insertions and reachability queries in $O(\log n)$ time.
Here, the width of partial orders is bounded by the number of processes, which is a small constant compared to the length of the trace, and hence our data structure is relevant.

\smallskip\noindent{\bf A practical algorithm and implementation.}
We develop an algorithm for the \emph{function problem} of race detection that is more practical than simply solving the decision problem for all possible pairs.
The efficiency of the algorithm comes while retaining the soundness and completeness guarantees.
We also develop sufficient conditions for detecting dynamically that our algorithm is complete for a given input, 
even in cases where completeness is not guaranteed theoretically.

We make a prototype implementation of our practical algorithm and evaluate it on a standard set of benchmark traces that contain hundreds of millions of events.
We compare the performance of our tool against state-of-the-art, polynomial-time, partial-order-based methods, namely the $\HB$~\cite{Lamport78}, $\WCP$~\cite{Kini17}, $\DC$~\cite{Roemer18} and $\SHB$~\cite{Mathur18} methods.
Our approach detects significantly more races than each of these methods,
while it has comparable running time, and typically being faster.
In fact, our algorithm does not simply detect more races; it detects \emph{all} races in the benchmark traces, and soundly reports that no more races (other than the detected ones) exist.
To our knowledge, this is the first sound algorithm that achieves such a level of performance on both running time and completeness of the reported races.

\section{Preliminaries}\label{SEC:PRELIMINARIES}

In this section we introduce useful notation and define the problem of dynamic race detection for lock-based concurrent programs.
The model follows similar recent works (e.g.,~\cite{Kini17}). %\cite{Smaragdakis12,Kini17,Mathur18,Roemer18}.

\smallskip\noindent{\bf Concurrent program.}
Given a natural number $k$, let $[k]$ denote the set $\{1,\dots, k\}$.
We consider a shared-memory concurrent program $\System$ that consists of $k$ processes $\{\Process_i\}_{i\in [k]}$, under sequential consistency semantics.
For simplicity of presentation we assume that $k$ is fixed a-priori, and no process is created dynamically.
All results presented here can be extended to a setting with dynamic process creation.
Communication between processes occurs over a set of global variables $\Globals$, and synchronization over a set of locks $\Locks$.
We let $\Vars=\Globals\cup \Locks$ be the set of all variables of $\System$.
Each process is deterministic, and performs a sequence of operations on execution.
We are only interested in the operations that access a global variable or a lock, which are called \emph{events}. 
%In particular, the allowed events are the following.
\begin{compactenum}
\item Given a global variable $x\in \Globals$, a process can \emph{write/read} to $x$ via an event $\Write(x)$/$\Read(x)$.
\item Given a lock $l\in \Locks$, a process can \emph{acquire} $\ell$ via an event $\Acquire(l)$ and \emph{release} $l$ via an event $\Release(l)$.

\end{compactenum}
Each such event is atomic. Given an event $\Event$, we let $\Location{\Event}$ denote the global variable (or lock) that $\Event$ accesses.
We denote by $\SysWrites_{\Process}$ (resp. $\SysReads_{\Process}$, $\SysAcquires_{\Process}$, $\SysReleases_{\Process}$) the set of all write (resp. read, acquire, release) events that can be performed by process $\Process$.
We let $\SysEvents_{\Process}=\SysWrites_{\Process}\cup \SysReads_{\Process}\cup \SysAcquires_{\Process}\cup \SysReleases_{\Process}$, and assume that $\SysEvents_{\Process}\cap \SysEvents_{\Process'}=\emptyset$ for every $\Process\neq \Process'$.
We denote by $\SysEvents=\bigcup_{\Process} \SysEvents_{\Process}$, $\SysWrites=\bigcup_{\Process} \SysWrites_{\Process}$, $\SysReads=\bigcup_{\Process} \SysReads_{\Process}$, $\SysAcquires=\bigcup_{\Process} \SysAcquires_{\Process}$, $\SysReleases=\bigcup_{\Process} \SysReleases_{\Process}$ the events, write, read, acquire and release events of the program $\System$, respectively.
Given an event $\Event\in \SysEvents$, we denote by $\Proc{\Event}$ the process that $\Event$ belongs to.
Finally, given a set of events $X\subseteq \SysEvents$, we denote by $\Reads{X}$ (resp., $\Writes{X}$, $\Acquires{X}$, $\Releases{X}$) the set of read (resp., write, lock-acquire, lock-release) events of $X$.

\smallskip\noindent{\bf Conflicting events.}
Given two distinct events $\Event_1, \Event_2\in \SysWrites\cup \SysReads$, we say that $\Event_1$ and $\Event_2$ are \emph{conflicting}, denoted by $\Confl{\Event_1}{\Event_2}$, if 
(i)~$\Location{\Event_1}=\Location{\Event_2}$ (i.e., they access the same global variable) and
(ii)~$\{\Event_1, \Event_2 \}\cap \SysWrites\neq \emptyset$ (i.e., at least one is a write event).
We extend the notion of conflict to locks, and say that two events $\Event_1, \Event_2\in \SysAcquires \cup \SysReleases$ are conflicting if $\Location{\Event_1}=\Location{\Event_2}$ (i.e., they use the same lock).

\smallskip\noindent{\bf Event sequences.}
Let $\Trace$ be a sequence of events.
We denote by $\Events{\Trace}$ the set of events, by $\Locks(\Trace)$ the set of locks, and by $\Globals(\Trace)$ the set of global variables in $\Trace$.
We let $\Writes{\Trace}$ (resp., $\Reads{\Trace}$, $\Acquires{\Trace}$, $\Releases{\Trace}$) denote the set $\Writes{\Events{\Trace}}$ (resp., $\Reads{\Events{\Trace}}$, $\Acquires{\Events{\Trace}}$, $\Releases{\Events{\Trace}}$).
Given two distinct events $\Event_1, \Event_2\in \Events{\Trace}$, 
we say that $\Event_1$ \emph{is earlier than}  $\Event_2$ in $\Trace$,
denoted by $\Event_1 <_{\Trace} \Event_2$ iff $\Event_1$ appears before $\Event_2$ in $\Trace$.
We say that $\Event_1$ is \emph{program-ordered earlier than} $\Event_2$,
denoted by $\Event_1<_{\TO(\Trace)}\Event_2$, to mean that $\Event_1 <_{\Trace} \Event_2$ and $\Proc{\Event_1}=\Proc{\Event_2}$.
When $\Trace$ is clear from the context, we simply write $\TO$ to denote $\TO(\Trace)$.
We let $=^{\Trace}$ be the identity relation on $\Events{\Trace}$, and denote by $\leq_{\Trace}$, $\leq_{\TO}$ the relations
$<_{\Trace} \cup =_{\Trace}$ and $<_{\TO} \cup =_{\Trace}$ respectively.
Given a set of events $X\subseteq \SysEvents$, we denote by $\Trace\Project X$ the \emph{projection} of $\Trace$ onto $X$,
i.e., it is the sub-sequence of events of $\Trace$ that belong to $X$.
Given two event sequences $\Trace_1$, $\Trace_2$, we denote by $\Trace_1\circ \Trace_2$ the concatenation of $\Trace_1$ with $\Trace_2$.
Finally, given a process $\Process_i$, we let $\Trace\Project\Process_i=\Trace\Project\SysEvents_{\Process_i}$.

\smallskip\noindent{\bf Lock events.}
Given a sequence of events $\Trace$ and a lock-acquire event $\Acquire\in \Acquires{\Trace}$, we denote by $\Match[\Trace]{\Acquire}$ the earliest lock-release event in $\Release\in \Releases{\Trace}$ such that $\Confl{\Release}{\Acquire}$ and $\Acquire<_{\Trace}\Release$, and let $\Match[\Trace]{\Acquire}=\bot$ if no such lock-release event exists.
If $\Match[\Trace]{\Acquire}\neq \bot$, we require that $\Proc{\Acquire}=\Proc{\Match[\Trace]{\Acquire}}$, i.e., the two lock events belong to the same process.
Similarly, given a lock-release event $\Release\in \Releases{\Trace}$, we denote by $\Match[\Trace]{\Release}$ the acquire event $\Acquire\in \Acquires{\Trace}$ such that $\Match[\Trace]{\Acquire}=\Release$ and require that such a lock-acquire event always exists.
%Given a lock-acquire event $\Acquire$, the \emph{critical section} $\CS[\Trace]{\Acquire}$ is the subsequence of events of $\Trace$ between $\Acquire$ and $\Match[\Trace]{\Acquire}$ such that these events belong to the same process as $\Acquire$.
%If $\Match[\Trace]{\Acquire}=\bot$ then the critical section of $\Acquire$ is the sub-sequence of events of the same process as $\Acquire$ that happen after $\Acquire$ in $\Trace$.

\smallskip\noindent{\bf Traces and observation functions.}
A sequence $\Trace$ is called a \emph{trace} if it satisfies the following.
\begin{compactenum}
\item For every read event $\Read\in \Reads{\Trace}$, there exists a write event $\Write\in \Writes{\Trace}$
such that $\Location{\Read}=\Location{\Write}$ and $\Write<_{\Trace} \Read$.
\item For any two lock-acquire events $\Acquire_1,\Acquire_2\in \Acquires{\Trace}$, if $\Location{\Acquire_1}=\Location{\Acquire_2}$ and $\Acquire_1<_{\Trace}\Acquire_2$,
then $\Release_1 =\Match[\Trace]{\Acquire_1}\in \Releases{\Trace}$ and $\Release_1<_{\Trace}\Acquire_2$.
\end{compactenum}

Given a trace $\Trace$, we define its \emph{observation function} $\Observation_{\Trace}:\Reads{\Trace}\to \Writes{\Trace}$ as follows:
$\Observation_{\Trace}(\Read)=\Write$ iff
\[
\Location{\Read}=\Location{\Write}
\quad\text{and}\quad
\Write<_{\Trace} \Read
\quad\text{and}\quad
\forall \Write'\in \Writes{\Trace}\setminus\{ \Write \} \text{ with } \Confl{\Write}{\Write'}: \Write'<_{\Trace} \Read \Implies \Write'<_{\Trace} \Write
\]
In words, $\Observation_{\Trace}$ maps every read event $\Read$ to the write event $\Write$ that $\Read$ observes in $\Trace$.
For simplicity, we assume that $\Trace$ starts with a write event to every location, hence $\Observation_{\Trace}$ is well-defined.

\smallskip\noindent{\bf Enabled events and races.}
An event $\Event\in \SysEvents$ is said to be \emph{enabled} in a trace $\Trace$ if $\Trace^*=\Trace\circ \Event$ is a trace of $\System$.
A trace $\Trace$ is said to exhibit a \emph{race} if there exist two consecutive conflicting events in $\Trace$ that belong to different processes.
Formally, there exist two events $\Event_1, \Event_2\in \SysReads\cup \SysWrites$ such that 
(i)~$\Proc{\Event_1}\neq \Proc{\Event_2}$,
(ii)~$\Confl{\Event_1}{\Event_2}$, 
(iii)~$\Event_1<_{\Trace}\Event_2$, and 
(iv)~for every $\Event\in \Events{\Trace}\setminus\{\Event_1, \Event_2\}$, we have that $\Event<_{\Trace}\Event_2\Implies \Event<_{\Trace}\Event_1$.

\smallskip\noindent{\bf Predictable races.}
A trace $\Trace'$ is a (prefix) \emph{correct reordering} of another trace $\Trace$ if (i)~for every process $\Process_i$,
we have that $\Trace'\Project \Process_i$ is a prefix of $\Trace\Project \Process_i$ and (ii)~$\Observation_{\Trace'}\subseteq \Observation_{\Trace}$,
i.e., the observation functions of $\Trace'$ and $\Trace$ agree on their common read events.
We say that $\Trace$ has a \emph{predictable race} on a pair of events $\Event_1, \Event_2\in \Events{\Trace}$ 
if there exists a correct reordering $\Trace'$ of $\Trace$ such that $\Trace^*=\Trace' \circ \Event_1 \circ \Event_2$ is a trace that exhibits the race $(\Event_1, \Event_2)$.

%\subsection{Problems Considered}\label{subsec:problems}

\smallskip\noindent{\bf Computational problems.}
The aim of this work is to present sound and fast algorithms for race detection, that also have certain completeness guarantees.
As usual in algorithmic parlance, we are concerned with two versions of the problem, namely the following.
Given an input trace $\Trace$,
\begin{compactenum}
\item the \emph{decision problem} is stated on two events $\Event_1,\Event_2\in\Events{\Trace}$, 
and asks whether $(\Event_1, \Event_2)$ is a predictable race of $\Trace$, and
\item the \emph{function problem} asks to compute the set of all pairs $\{(\Event_1^i, \Event_2^i)\}_{i}$
such that each $(\Event_1^i, \Event_2^i)$ is a predictable race of $\Trace$.
\end{compactenum}

\smallskip\noindent{\bf Soundness, completeness and complexity.}
A predictive race-detection algorithm is called \emph{sound} if on every input trace $\Trace$, every reported race is a predictable race of $\Trace$.
The algorithm is called \emph{complete} if it reports all predictable races of $\Trace$.
We note that these notions are often used in reverse in program verification.
However, here we align with the terminology used in predictive techniques, hence soundness (resp., completeness) means the absence of false positives (resp., false negatives).
We measure complexity in terms of the length $n$ of $\Trace$.
Other important parameters are the number of processes $k$ and the number of global variables $\Globals$.
Typically $k$ is much smaller than $n$, and is treated as a constant.
For simplicity, we also ignore $\Globals$ in our complexity statements.
In all cases, our algorithms have a dependency of factor $k^2\cdot |\Globals|$ (and hence polynomial) on these parameters.

\smallskip\noindent{\bf Dynamic process creation and other synchronization primitives.}
To keep the presentation simple, in the theoretical part of this work we neglect dynamic process creation (i.e., fork/join events).
We note that such events can be handled naturally in our framework.
In our experiments (\cref{SEC:EXPERIMENTS}) we explain how we handle dynamic process creation, which is present in our benchmark set.
Similarly, our focus on locks is for simplicity of presentation and not restrictive to our model.
For dynamic race detection, other synchronization primitives, such as compare-and-swap, intrinsic locks and synchronized methods
can be simulated with locks and extra orderings in the partial orders.
Indeed, this modeling approach has been taken in many other works, as e.g. in~\cite{Smaragdakis12,Kini17,Mathur18,Roemer18}.

Due to limited space, all proofs are relegated to \cref{sec:missing_proofs}.
\section{Trace-closed Partial Orders}\label{sec:po}

In this section we present relevant notation on partial orders, and introduce the concept of closed partial orders.
We also present max-min linearizations which linearize closed partial orders to valid traces.
Since this our most technical section, we provide here an overview to assist the reader.

\begin{compactenum}
\item In \cref{subsec:po} we define general notation on partial orders.
Since these are partial orders over sets of events $X$ of an input trace $\Trace$,
we introduce a \emph{feasibility} criterion for these sets, 
which requires that certain events are present in the partial order.
For example, for every two conflicting lock-acquire events in $X$, at least one corresponding lock-release event must also be in $X$.
\item In \cref{subsec:ol_closure} we define trace-closed partial orders.
Intuitively, this notion requires certain orderings between conflicting events to be present in the partial order.
Note that \emph{not} every linearization of a partial order leads to a valid trace
(e.g., some linearizations might not respect the lock semantics).
Nevertheless, we show that for a specific class of trace-closed partial orders, 
a specific type of \emph{max-min} linearization is guaranteed to \emph{always} produce a valid trace.
\item In \cref{subsec:closure_comp} we develop an algorithm that computes the trace-closure of a partial order efficiently. 
To this end, we develop a data structure $\DS$ for the efficient representation of partial orders.
For ease of presentation, we relegate the technical description of $\DS$ in \cref{sec:dag_reach}.
\end{compactenum}

%The technical tools we develop here are directly used for the race detection problem in later sections.

\subsection{Partial Orders}\label{subsec:po}

\smallskip\noindent{\bf Feasible sets.}
Given a set of events $X\subseteq \Events{\Trace}$, we say that $X$ is \emph{prefix-closed} for $\Trace$ if
for every pair of events $\Event_1,\Event_2\in \Events{\Trace}$ if $\Event_1\leq_{\TO} \Event_2$ and $\Event_2\in X$,
then $\Event_1\in X$ (i.e., $X$ is an ideal of $\leq_{\TO}$).
We define the \emph{open acquires} of $X$ under $\Trace$ as
$
\OpenAcquires_{\Trace}(X)=\{\Acquire\in \Acquires{X}:~\Match[\Trace]{\Acquire}\not\in X \}
$.

%We define the \emph{open locks} of $X$ under $\Trace$ as
%\[
%\OpenLocks_{\Trace}(X)=\bigcup_{\Acquire\in \OpenAcquires_{\Trace}{X}:
%\]
%i.e., $\OpenLocks_{\Trace}(X)$ contains the set of locks that are acquired in $X$ but not released in $X$.
We call $X$ \emph{observation-feasible} for $\Trace$ if for every read event $\Read\in \Reads{X}$, we have  $\Observation_{\Trace}(\Read)\in X$.
We call $X$ \emph{lock-feasible} for $\Trace$ if
(i)~for every lock-release event $\Release\in \Releases{X}$, we have  $\Match[\Trace]{\Release}\in X$, and
(ii)~for every distinct pair of lock-acquire events $\Acquire_1,\Acquire_2\in \OpenAcquires_{\Trace}(X)$, we have  $\Location{\Acquire_1}\neq \Location{\Acquire_2}$.
In words, $X$ is lock-feasible if every release event of $X$ has its matching acquire event also in $X$, and every open lock of $X$ remains open by exactly one acquire event of $X$.
Finally, we call $X$ \emph{feasible} for $\Trace$ if $X$ is 
prefix-closed, 
observation-feasible, and 
lock-feasible for $\Trace$.

\smallskip\noindent{\bf Partial orders.}
Given a trace $\Trace$ and a set $X\subseteq \Events{\Trace}$, a \emph{partial order} $P(X)$ over $X$ is a reflexive, antisymmetric and transitive relation over $X$
(i.e., $\leq_{P(X)}\subseteq X\times X$).
When $X$ is clear from the context, we will simply write $P$ instead of $P(X)$.
Given two events $\Event_1, \Event_2$ we write $\Event_1<_{P}\Event_2$ to denote that $\Event_1\leq_{P}\Event_2$ and $\Event_1\neq \Event_2$.
Given two distinct events $\Event_1,\Event_2\in X$, we say that $\Event_1$ and $\Event_2$ are \emph{unordered} by $P$, denoted by $\Unordered{\Event_1}{P}{\Event_2}$, if neither $\Event_1<_{P}\Event_2$ nor $\Event_2<_{P} \Event_1$.
Given a set $Y\subseteq X$, we denote by $P\Project Y$ the \emph{projection} of $P$ on $Y$,
i.e., we have $\leq_{P\Project Y}\subseteq Y\times Y$, and for all $\Event_1,\Event_2\in Y$, $\Event_1\leq_{P\Project Y} \Event_2$ iff $\Event_1\leq_{P} \Event_2$.
Given two partial orders $P$ and $Q$ over a common set $X$, we say that $Q$ \emph{refines} $P$, denoted by $Q\Refines P$, if
for every pair of events $\Event_1, \Event_2\in X$, if $\Event_1\leq_{P}\Event_2$ then $\Event_1\leq_{Q}\Event_2$.
If $Q$ refines $P$, we say that $P$ is \emph{weaker} than $Q$.
A \emph{linearization} of $P$ is a total order that refines $P$.
We make the following remark.

\smallskip
\begin{remark}\label{rem:linearization}
Not every linearization of a partial order $P$ is a valid trace, and generally, $P$ is not guaranteed to have such a linearization.
Our algorithm for dynamic race detection relies on developing sufficient conditions
under which $P$ indeed has a linearization to a valid trace.
% as well as constructing the witness trace.
\end{remark}

\smallskip\noindent{\bf Width and Mazurkiewicz traces.}
Let $P$ be a partial order over a set $X\subseteq \Events{\Trace}$.
The \emph{width} $\Width(P)$ of $P$ is the length of its longest antichain.
i.e., it is the largest size of a set $Y\subseteq X$ such that for every pair of distinct elements $\Event_1,\Event_2\in Y$ we have $\Unordered{\Event_1}{P}{\Event_2}$.
The partial order $P$ is called a \emph{Mazurkiewicz trace} (or \emph{M-trace} for short) if for every two conflicting events $\Event_1, \Event_2\in X$, we have $\Ordered{\Event_1}{P}{\Event_2}$~\cite{Mazurkiewicz87}.
Note that if $\Width(P)=1$ then $P$ is trivially an M-trace.

%\smallskip\noindent{\em Example on width and Mazurkiewicz width.}
%Consider the partial order $P$ shown in \cref{fig:maz_width}.
%We have $\Width(P)=4$ since 
%(i)~the events in $Y=\{\Write_1(x), \Write_2(x), \Read_1(y), \Read_2(y),\}$ are pairwise unordered and 
%(ii)~every set of $5$ events contains an ordered pair.
%We have $\MWidth(P)=2$ since 
%(i)~the events in $Y=\{\Write_1(x), \Write_2(x) \}$ are conflicting and unordered, and
%(ii)~every set of $3$ pairwise conflicting events contains an ordered pair.
%Note that there exists only one such set, on variable $x$, namely $\{\Write_1(x), \Write_2(x), \Read_1(x) \}$,
%as on variable $y$, all sets of pairwise conflicting events have size $2$, namely $\{\Write_1(y), \Read_1(y)\}$ and $\{\Write_1(y), \Read_2(y)\}$.
%
%\input{figures/maz_width}

\subsection{Trace-closed Partial Orders}\label{subsec:ol_closure}

In this section we define the notion of trace-closed partial orders.
This is a central concept in this work, as our race-detection algorithm is based on computing trace-closed partial orders efficiently.

\smallskip\noindent{\bf Trace-respecting partial orders.}
Let $\Trace$ be a trace, and $P$ a partial order over a feasible set $X\subseteq \Events{\Trace}$.
We say that $P$ \emph{respects} $\Trace$ if the following conditions hold.
\begin{compactenum}
\item\label{item:resp1} $P\Refines \TO\Project X$, i.e., $P$ refines the program order when restricted to the set $X$.
\item\label{item:resp2} For every read event $\Read\in \Reads{X}$ we have  $\Observation_{\Trace}(\Read)<_{P} \Read$.
\item\label{item:resp3} For every lock-acquire event $\Acquire\in \Acquires{X}$, if $\Match[\Trace]{\Acquire}\not \in X$,
then for every lock-release event $\Release\in \Releases{X}$ such that $\Confl{\Release}{\Acquire}$,
we have  $\Release<_{P}\Acquire$.
\end{compactenum}
We denote by $\RespectPO{\Trace}{X}$ the weakest partial order over $X$ that respects $\Trace$.
%$\RespectPO{\Trace}{X}$ is easily constructed by a single pass of $\Trace$.

\smallskip\noindent{\bf Trace-closed partial orders.}
Let $\Trace$ be a trace, and $P$ a partial order over a feasible set $X\subseteq \Events{\Trace}$ such that $P$ respects $\Trace$.
We call $P$ \emph{observation-closed} if the following condition holds.
For every read event $\Read\in \Reads{X}$, let $\Write=\Observation_{\Trace}(\Read)$.
For every write event $\Write'\in \Writes{X}\setminus\{ \Write \}$ such that
$\Confl{\Write'}{\Read}$, we have
\[
\text{if }  \Write' <_{P} \Read  \text{ then } \Write'<_{P}\Write \quad\text{and}\quad
\text{if }  \Write <_{P} \Write' \text{ then } \Read<_{P} \Write'
\]
For a pair of lock-release events $\Release_1, \Release_2\in \Releases{X}$, let $\Acquire_i=\Match[\Trace]{\Release_i}$.
We call $P$ \emph{lock-closed} if for every $\Acquire_1,\Acquire_2\in\SysAcquires$ and $\Release_1,\Release_2\in\SysReleases$,
if $\Confl{\Release_2}{\Acquire_1}$ and $\Acquire_1\leq_{P} \Release_2$, then $\Release_1\leq_{P} \Acquire_2$.
Finally, we call $P$ \emph{trace-closed} (or simply \emph{closed}) if it is both observation-closed and lock-closed.
See \cref{fig:closure_defs} for an illustration.
Note that a closed partial order can still contain conflicting events that are unordered.
In addition, it does not necessarily admit a linearization to a valid trace.
In the next paragraph we develop sufficient conditions for when such a linearization exists.

%\smallskip
%\begin{remark}\label{rem:closed_maz}
%If $\MWidth(P)=1$ then any linearization of $P$ is closed.
%\end{remark}

\begin{figure}[!h]
\begin{subfigure}[t]{0.25\textwidth}
\centering
\begin{tikzpicture}[thick,
pre/.style={<-,shorten >= 1pt, shorten <=1pt, thick},
post/.style={->,shorten >= 2pt, shorten <=2pt,  very thick},
und/.style={very thick, draw=gray},
node1/.style={rectangle, minimum size=5mm, fill=white!100,  line width=1pt, inner sep=0},
virt/.style={circle,draw=black!50,fill=black!20, opacity=0}]

\newcommand{\xdisposition}{0}
\newcommand{\ydisposition}{0}
\newcommand{\xstep}{1.2}
\newcommand{\ystep}{0.9}
\newcommand{\ybias}{0.6}

\node	[node1]		(1)	at	(\xdisposition + 0*\xstep, \ydisposition + 0*\ystep)	{$\Write$};
\node	[node1]		(2)	at	(\xdisposition + 0*\xstep, \ydisposition -1.5*\ystep)	{$\Read$};
\node	[node1]		(3)	at	(\xdisposition + 1*\xstep, \ydisposition -0.75*\ystep)	{$\Write'$};

\draw[post ] (1) to (2);
\draw[post ] (3) to (2);
\draw[post, \darkred, dashed] (3) to (1);

\end{tikzpicture}
\caption{}
\label{subfig:oc1}
\end{subfigure}
\begin{subfigure}[t]{0.25\textwidth}
\centering
\begin{tikzpicture}[thick,
pre/.style={<-,shorten >= 1pt, shorten <=1pt, thick},
post/.style={->,shorten >= 2pt, shorten <=2pt,  very thick},
und/.style={very thick, draw=gray},
node1/.style={rectangle, minimum size=5mm, fill=white!100,  line width=1pt, inner sep=0},
virt/.style={circle,draw=black!50,fill=black!20, opacity=0}]

\newcommand{\xdisposition}{0}
\newcommand{\ydisposition}{0}
\newcommand{\xstep}{1.2}
\newcommand{\ystep}{0.9}
\newcommand{\ybias}{0.6}

\node	[node1]		(10)	at	(\xdisposition + 0*\xstep, \ydisposition + 0*\ystep)	{$\Write$};
\node	[node1]		(20)	at	(\xdisposition + 0*\xstep, \ydisposition -1.5*\ystep)	{$\Read$};
\node	[node1]		(30)	at	(\xdisposition + 1*\xstep, \ydisposition -0.75*\ystep)	{$\Write'$};

\draw[post] (10) to (20);
\draw[post, \darkred, dashed ] (20) to (30);
\draw[post ] (10) to (30);

\end{tikzpicture}
\caption{}
\label{subfig:oc2}
\end{subfigure}
\qquad
\begin{subfigure}[t]{0.4\textwidth}
\centering
\begin{tikzpicture}[thick,
pre/.style={<-,shorten >= 1pt, shorten <=1pt, thick},
post/.style={->,shorten >= 2pt, shorten <=2pt, very thick},
und/.style={very thick, draw=gray},
node1/.style={rectangle, minimum size=5mm, fill=white!100,  line width=1pt, inner sep=0},
virt/.style={circle,draw=black!50,fill=black!20, opacity=0}]

\newcommand{\xdisposition}{0}
\newcommand{\ydisposition}{0}
\newcommand{\xstep}{1.2}
\newcommand{\ystep}{1.4}
\newcommand{\ybias}{0.6}

\node	[node1]		(1)	at	(\xdisposition + 0*\xstep, \ydisposition + 0*\ystep)	{$\Acquire_1$};
\node	[node1]		(3)	at	(\xdisposition + 0*\xstep, \ydisposition -1*\ystep)	{$\Release_1$};

\node	[node1]		(4)	at	(\xdisposition + 1*\xstep, \ydisposition + 0*\ystep)	{$\Acquire_2$};
\node	[node1]		(6)	at	(\xdisposition + 1*\xstep, \ydisposition -1*\ystep)	{$\Release_2$};

\draw[post ] (1) to (3);
\draw[post ] (4) to (6);
\draw[post ] (1) to (6);
\draw[post, \darkred, dashed] (3) to (4);

\end{tikzpicture}
\caption{}
\label{subfig:lc}
\end{subfigure}
\caption{The conditions of observation closure (\protect\subref{subfig:oc1},\protect\subref{subfig:oc2}) and lock closure (\protect\subref{subfig:lc}). 
Solid edges and dashed edges represent existing and inferred orderings, respectively.}
\label{fig:closure_defs}
\end{figure}
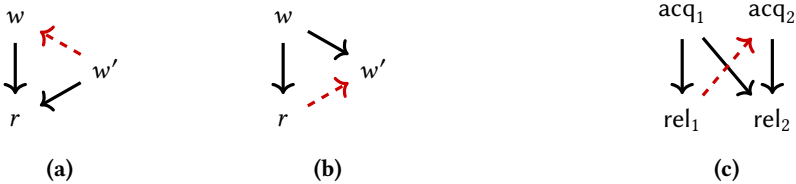

\smallskip\noindent{\bf Max-min linearizations.}
The key technical challenge in race prediction is, given a trace $\Trace$, to construct a partial order $P$ over $\Events{\Trace}$ such that $P$ is efficiently linearizable to a correct reordering of $\Trace$ that manifests the race.
Here we use trace-closed partial orders to provide a sufficient condition for efficient linearization, which we call the \emph{max-min linearization}.
In later sections, our race-detection algorithm constructs trace-closed partial orders.
The max-min linearization of such partial orders will guarantee that the races exposed by these partial orders are indeed valid races,
which are exhibited by a trace constructed using the max-min linearization.

Let $\Trace$ be a trace, and consider a partial order $P$ over a feasible set $X\subseteq \Events{\Trace}$ such that
$P$ is trace-closed for $\Trace$ and $X$ can be partitioned into two sets $X_1, X_2\subseteq X$ such that
(i)~$\Width(P\Project X_1)=1$ and 
(ii)~$P\Project X_2$ is an M-trace.
The \emph{max-min} linearization $\Trace^{*}$ is a linearization of $P$ given by \cref{algo:maxmin}.
In words, first every event of $X_1$ is ordered before every event of $X_2$, as long as this is allowed by $P$, and then the resulting partial order is linearized arbitrarily.
Intuitively, we obtain the sequence $\Trace^*$ by linearizing $X_1$ maximally, and $X_2$ minimally.
See \cref{fig:maxmin} for an illustration.

\noindent{\em Intuition.}
First, observe that $P$ can contain pairs of conflicting events that are unordered, i,.e., between the sets $X_1$ and $X_2$.
Conceptually, $\MinMaxAlgo$ shows that as we attempt to linearize $P$, we do not have to make an exhaustive search over all the possible (exponentially many) orderings of such pairs.
Instead, the specific orderings made by $\MinMaxAlgo$ are guaranteed to produce a correct linearization.
The intuition behind the correctness of $\MinMaxAlgo$ can be summarized as follows.
\begin{compactenum}
\item Since $\Width(P\Project X_1)=1$, ordering every two events $\Event_1\in X_1$, $\Event_2\in X_2$ as $\Event_1\to \Event_2$ (provided that $\Event_2\not<_{P}\Event_1)$ creates a partial order (i.e., no cycle is formed).
\item Since $P$ is closed and $P\Project X_2$ is an M-trace, this ordering respects the observation $\Write$ of every read event $\Read$.
Indeed, if the ordering was forcing some other conflicting write event $\Write'$ between $\Write$ and $\Read$, then $\Write'$ must be ordered with at least one of $\Write$ and $\Read$, and then the corresponding closure rule (\cref{fig:closure_defs}) would  have resolved this conflict entirely.
\end{compactenum}

\smallskip
\begin{restatable}{theorem}{maxmin}\label{them:maxmin}
Let $\Trace$ be a trace and $P$ a partial order over a feasible set $X\subseteq\Events{\Trace}$ such that
$P$ is trace-closed for $\Trace$ and $X$ can be partitioned into two sets $X_1, X_2$ so that
(i)~$\Width(P\Project X_1)=1$ and
(ii)~$P\Project X_2$ is a Mazurkiewicz trace.
The max-min linearization of $P$ produces a correct reordering of $\Trace$.
\end{restatable}

\begin{figure}[!h]
\centering
\begin{tikzpicture}[thick,
pre/.style={<-,shorten >= 1pt, shorten <=1pt, thick},
post/.style={->,shorten >= 1pt, shorten <=1pt,  thick},
und/.style={very thick, draw=gray},
node1/.style={rectangle, minimum size=5mm, fill=white!100,  line width=1pt, inner sep=0},
virt/.style={circle,draw=black!50,fill=black!20, opacity=0}]

\newcommand{\xdisposition}{0}
\newcommand{\ydisposition}{0}
\newcommand{\xstep}{-1.2}
\newcommand{\ystep}{1.2}
\def\yone{0.7}
\def\ytwo{1.2}
\def\ythree{1.7}
\def\yfour{3}
\def\yfive{1.4}
\def\ysix{2.5}
\def\yseven{3}

%\draw[-, line width=2.5, ] (\xdisposition + 0*\xstep, \ydisposition + 0*\ystep) to (\xdisposition + 0*\xstep, \ydisposition - \yone*\ystep);
%\draw [-, decorate,decoration={brace,amplitude=5pt,mirror,raise=2pt},  shorten >= 1pt, shorten <=1pt] (\xdisposition + 0*\xstep, \ydisposition + 0*\ystep) to node [left, xshift=-5] {$1$}  (\xdisposition + 0*\xstep, \ydisposition - \yone*\ystep);

\draw[-, line width=2.5, ] (\xdisposition - 1*\xstep, \ydisposition -\yone*\ystep) to (\xdisposition - 1*\xstep, \ydisposition - \ytwo*\ystep);
\draw [-, decorate,decoration={brace,amplitude=5pt,raise=2pt},  shorten >= 1pt, shorten <=1pt] (\xdisposition - 1*\xstep, \ydisposition -\yone*\ystep) to node [left, xshift=20] {$2$}  (\xdisposition - 1*\xstep, \ydisposition - \ytwo*\ystep);
\draw[-, line width=2.5, ] (\xdisposition - 1*\xstep, \ydisposition -\ytwo*\ystep) to (\xdisposition - 1*\xstep, \ydisposition - \ythree*\ystep);
\draw [-, decorate,decoration={brace,amplitude=5pt,raise=2pt},  shorten >= 1pt, shorten <=1pt] (\xdisposition - 1*\xstep, \ydisposition -\ytwo*\ystep) to node [left, xshift=20] {$4$}  (\xdisposition - 1*\xstep, \ydisposition - \ythree*\ystep);
\draw[->, line width=2.5, ] (\xdisposition - 1*\xstep, \ydisposition -\ythree*\ystep) to (\xdisposition - 1*\xstep, \ydisposition - \yfour*\ystep - 0.3);
\draw [-, decorate,decoration={brace,amplitude=5pt,raise=2pt},  shorten >= 1pt, shorten <=1pt] (\xdisposition - 1*\xstep, \ydisposition -\ythree*\ystep) to node [left, xshift=20] {$6$}  (\xdisposition - 1*\xstep, \ydisposition - \yfour*\ystep);

\draw[-, line width=2.5, ] (\xdisposition + 1*\xstep, \ydisposition -\yone*\ystep) to (\xdisposition + 1*\xstep, \ydisposition - \yfive*\ystep);
\draw [-, decorate,decoration={brace,amplitude=5pt,mirror,raise=2pt},  shorten >= 1pt, shorten <=1pt] (\xdisposition + 1*\xstep, \ydisposition -\yone*\ystep) to node [right, xshift=-20] {$1$}  (\xdisposition + 1*\xstep, \ydisposition - \yfive*\ystep);
\draw[-, line width=2.5, ] (\xdisposition + 1*\xstep, \ydisposition -\yfive*\ystep) to (\xdisposition + 1*\xstep, \ydisposition - \ysix*\ystep);
\draw [-, decorate,decoration={brace,amplitude=5pt,mirror,raise=2pt},  shorten >= 1pt, shorten <=1pt] (\xdisposition + 1*\xstep, \ydisposition -\yfive*\ystep) to node [right, xshift=-20] {$3$}  (\xdisposition + 1*\xstep, \ydisposition - \ysix*\ystep);
\draw[->, line width=2.5, ] (\xdisposition + 1*\xstep, \ydisposition -\ysix*\ystep) to (\xdisposition + 1*\xstep, \ydisposition - \yseven*\ystep - 0.3);
\draw [-, decorate,decoration={brace,amplitude=5pt,mirror,raise=2pt},  shorten >= 1pt, shorten <=1pt] (\xdisposition + 1*\xstep, \ydisposition -\ysix*\ystep) to node [right, xshift=-20] {$5$}  (\xdisposition + 1*\xstep, \ydisposition - \yseven*\ystep);

\draw[->, very thick, shorten >= 2pt, shorten <=1pt, thick]  (\xdisposition - 1*\xstep, \ydisposition -\ytwo*\ystep) to (\xdisposition + 1*\xstep, \ydisposition - \yfive*\ystep);
\draw[->, very thick, shorten >= 2pt, shorten <=1pt, thick]  (\xdisposition - 1*\xstep, \ydisposition -\ythree*\ystep) to (\xdisposition + 1*\xstep, \ydisposition - \ysix*\ystep);
\draw[->, very thick, shorten >= 2pt, shorten <=1pt, thick]  (\xdisposition + 1*\xstep, \ydisposition -1.8*\ystep) to (\xdisposition - 1*\xstep, \ydisposition - 2.4*\ystep);

\node[] at (\xdisposition - 2*\xstep, \ydisposition - 1.85*\ystep) {\Large $P\Project X_2$};
\node[] at (\xdisposition + 2*\xstep, \ydisposition - 1.85*\ystep) {\Large $P\Project X_1$};

\end{tikzpicture}
\caption{Illustration of the max-min linearization. 
Here $\Width(P\Project X_1)=\Width(P\Project X_2)=1$.
The numbers show the order in which various segments of the sets $X_1$ and $X_2$ are executed,
i.e., the linearization $\Trace^*$ consists of the segments 
$\Trace^*=\langle 1\rangle \circ \langle 2\rangle \circ \langle 3\rangle \circ \langle 4\rangle \circ \langle 5\rangle \circ \langle 6\rangle$.
\cref{them:maxmin} guarantees that $\Trace^*$ is a correct reordering.
}
\label{fig:maxmin}
\end{figure}
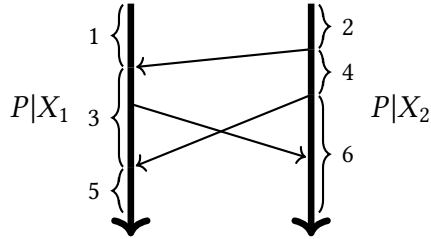

\smallskip
\begin{algorithm}%[H]
\small
\DontPrintSemicolon
\SetInd{0.4em}{0.4em}
\caption{$\MinMaxAlgo$}\label{algo:maxmin}
\KwIn{ A trace $\Trace$, a closed partial order $P$ over a feasible set $X\subseteq \Events{\Trace}$, a partitioning of $X$ to $X_1, X_2$ s.t. $\Width(P\Project X_1)=1$ and $P\Project X_2$ is an M-trace.}
\KwOut{A linearization of $P$ that is a correct reordering of $\Trace$.}
\BlankLine
Let $Q\gets P$\label{line:Q}\\
\ForEach{$\Event_1\in X_1, \Event_2\in X_2$ such that $\Unordered{\Event_1}{P}{\Event_2}$}{
Insert $(\Event_1\to \Event_2)$ in $Q$\label{line:e2toe1}\\
}
$\Return{}$ any linearization $\Trace^*$ of $Q$
\end{algorithm}

\subsection{Computing the Closure of a Partial Order}\label{subsec:closure_comp}

In this section we define the trace-closure of partial orders, and develop an efficient algorithm that, given a partial order $P$, either computes the closure of $P$ or concludes that the closure does not exist.
In the next section we will solve the decision problem of race detection by constructing specific partial orders and computing their closure.

\smallskip\noindent{\bf Feasible partial orders.}
Let $\Trace$ be a trace and $P$ a partial order over a feasible set $X\subseteq \Events{\Trace}$ such that $P$ respects $\Trace$.
If there exists a partial order $Q$ over $X$ such that
(i)~$Q\Refines P$ and (ii)~$Q$ is closed,
we define the \emph{closure} of $P$ as the smallest such partial order $Q$.
If no such partial order $Q$ exists, then the closure of $P$ is undefined (i.e., $P$ does not have a closure).
We call $P$ \emph{feasible} iff it has a closure.
The following lemma states that $P$ has a unique closure.

\smallskip
\begin{restatable}{lemma}{closureunique}\label{lem:closure_unique}
There exists at most one smallest partial order $Q$ such that (i)~$Q\Refines P$ and (ii)~$Q$ is closed.
\end{restatable}

\smallskip\noindent{\bf Computing the closure of a partial order.}
It is straightforward to verify that, given a partial order $P$, the closure of $P$ (or deducing that $P$ is not feasible)
can be computed in polynomial time.
This is simply achieved by iteratively detecting whether one of the cases shown in \cref{fig:closure_defs} is violated,
and strengthening $P$ with the appropriate orderings.
However, since our goal is to handle large traces with hundreds of millions of events,
polynomial-time guarantees are not enough, and the goal is to develop an algorithm with low polynomial complexity.
Here we develop such an algorithm called, $\Closure$, that computes the closure of a partial order in $O(n^2\cdot \log n)$ time.

\smallskip\noindent{\em The data structure $\DS$.}
To make the closure computation efficient, we develop a data structure $\DS$ for manipulating partial-orders efficiently.
Given a partial order $P$ over $n$ events such that $P$ has width $k=O(1)$,
$\DS$ represents $P$ in $O(n)$ space and supports the following operations:
(i)~initialization in $O(n)$ time,
(ii)~querying whether $\Event_1\leq_{P}\Event_2$, for any two events $\Event_1, \Event_2$ in $O(\log n)$ time, and 
(iii)~inserting an ordering $\Event_1\leq_{P}\Event_2$, for any two events $\Event_1, \Event_2$ in $O(\log n)$ time.
For ease of presentation, we relegate the formal description of $\DS$ to \cref{sec:dag_reach}.

\smallskip\noindent{\em The event maps $\From$, $\To$ and $\LastFlow$.}
Consider a trace $\Trace$.
For every lock $l\in \Locks(\Trace)$, we define the maps $\From_l^{\SysAcquires}, \From_l^{\SysReleases}, \To_l^{\SysAcquires}, \To_l^{\SysReleases}: \Events{\Trace} \to \Events{\Trace}\cup \{\bot\}$, as follows.
Given an event $\Event\in \Events{\Trace}$, the maps $\From_l^{\SysAcquires}(\Event)$ and $\To_l^{\SysAcquires}(\Event)$ point to the first lock-acquire  event $\Acquire$ after $\Event$ in $\Trace$,
and last lock-acquire event $\Acquire$ before $\Event$ in $\Trace$, respectively,
such that $\Proc{\Event}=\Proc{\Acquire}$ and $\Location{\Acquire}=l$.
The maps $\From_l^{\SysReleases}(\Event)$ and $\To_l^{\SysReleases}(\Event)$ are defined analogously, pointing to lock-release instead of lock-acquire events.
Similarly, for every global variable $x\in \Globals(\Trace)$, we define the maps $\From_x^{\SysWrites}, \From_x^{\SysReads}, \To_x^{\SysWrites}, \To_x^{\SysReads}: \Events{\Trace} \to \Events{\Trace}\cup \{\bot\}$, as follows.
Given an event $\Event\in \Events{\Trace}$, the map $\From_x^{\SysWrites}(\Event)$ (resp. $\To_x^{\SysWrites}(\Event)$) points to the first write event $\Write$ after (resp., before) $\Event$ in $\Trace$ such that $\Proc{\Event}=\Proc{\Write}$ and $\Location{\Write}=x$.
The maps $\From_l^{\SysReads}(\Event)$ and $\To_l^{\SysReads}(\Event)$ are defined analogously, pointing to read instead of write events.
Finally, the \emph{flow map} $\LastFlow_{\Process}: \Writes{\Trace} \to \Reads{\Trace}\cap \SysReads_{\Process}$ of $\Trace$ is a partial function that maps each write event $\Write$ to the last read event of $\Process$ that observes $\Write$.
In all the above cases, if no corresponding event exists, the respective map points to $\bot$.
Observe that each of these maps has size $O(|G|\cdot n)$, where $|G|$ is the number of memory locations of $\Trace$.
The maps can be constructed in $O(|G|\cdot n)$ time, simply by traversing $\Trace$ and maintaining on-the-fly each map.

\begin{figure}[!h]
\begin{subfigure}{0.25\textwidth}
\small
\centering
\begin{tikzpicture}[thick,
pre/.style={<-,shorten >= 1pt, shorten <=1pt, thick},
post/.style={->,shorten >= 2pt, shorten <=2pt,  very thick},
seqtrace/.style={->, line width=2},
und/.style={very thick, draw=gray},
event/.style={rectangle, minimum height=0.8mm, minimum width=3mm, fill=black!100,  line width=1pt, inner sep=0},
virt/.style={circle,draw=black!50,fill=black!20, opacity=0}]

\newcommand{\xdisposition}{0}
\newcommand{\ydisposition}{0}
\newcommand{\xstep}{1.5}
\newcommand{\ystep}{0.55}
\newcommand{\ybias}{0.4}
\newcommand{\xbias}{0.4}

\node	[]		(t1a)	at	(\xdisposition + 0*\xstep, \ydisposition + 0*\ystep)	{\normalsize$\SeqTrace_1$};
\node	[]		(t1b)	at	(\xdisposition + 0*\xstep, \ydisposition + -4*\ystep)	{};
\node	[]		(t2a)	at	(\xdisposition + 1*\xstep, \ydisposition + 0*\ystep)	{\normalsize$\SeqTrace_2$};
\node	[]		(t2b)	at	(\xdisposition + 1*\xstep, \ydisposition + -4*\ystep)	{};

\node[event] (e1) at (\xdisposition + 0*\xstep, \ydisposition + -2*\ystep) {};
\node[] (et1) at (\xdisposition + 0*\xstep - \xbias, \ydisposition + -2*\ystep) {$\Event_1$};
\node[event] (e2) at (\xdisposition + 1*\xstep, \ydisposition + -2*\ystep) {};
\node[] (et2) at (\xdisposition + 1*\xstep + \xbias, \ydisposition + -2*\ystep) {$\Event_2$};

\node[event] (w) at (\xdisposition + 0*\xstep, \ydisposition + -1*\ystep) {};
\node[] (wt) at (\xdisposition + 0*\xstep - \xbias, \ydisposition + -1*\ystep) {$\Write$};
\node[event] (r) at (\xdisposition + 1*\xstep, \ydisposition + -3*\ystep) {};
\node[] (rt) at (\xdisposition + 1*\xstep + \xbias, \ydisposition + -3*\ystep) {$\Read$};

\node[event] (o) at (\xdisposition + 0.65*\xstep, \ydisposition + -1.25*\ystep) {};
\node[] (ot) at (\xdisposition + 0.65*\xstep , \ydisposition + -1.25*\ystep + \ybias) {$\Observation_{\Trace}(\Read)$};

\draw[seqtrace] (t1a) to (t1b);
\draw[seqtrace] (t2a) to (t2b);

\draw[post] (e1) to (e2);
\draw[post, \darkred, dashed] (w) to (o);

\end{tikzpicture}
\caption{}
\label{subfig:algo_oc1}
\end{subfigure}
\begin{subfigure}{0.25\textwidth}
\small
\centering
\begin{tikzpicture}[thick,
pre/.style={<-,shorten >= 1pt, shorten <=1pt, thick},
post/.style={->,shorten >= 2pt, shorten <=2pt,  very thick},
seqtrace/.style={->, line width=2},
und/.style={very thick, draw=gray},
event/.style={rectangle, minimum height=0.8mm, minimum width=3mm, fill=black!100,  line width=1pt, inner sep=0},
virt/.style={circle,draw=black!50,fill=black!20, opacity=0}]

\newcommand{\xdisposition}{0}
\newcommand{\ydisposition}{0}
\newcommand{\xstep}{1.5}
\newcommand{\ystep}{0.55}
\newcommand{\ybias}{0.4}
\newcommand{\xbias}{0.4}

\node	[]		(t1a)	at	(\xdisposition + 0*\xstep, \ydisposition + 0*\ystep)	{\normalsize$\SeqTrace_1$};
\node	[]		(t1b)	at	(\xdisposition + 0*\xstep, \ydisposition + -4*\ystep)	{};
\node	[]		(t2a)	at	(\xdisposition + 1*\xstep, \ydisposition + 0*\ystep)	{\normalsize$\SeqTrace_2$};
\node	[]		(t2b)	at	(\xdisposition + 1*\xstep, \ydisposition + -4*\ystep)	{};

\node[event] (e1) at (\xdisposition + 0*\xstep, \ydisposition + -2*\ystep) {};
\node[] (et1) at (\xdisposition + 0*\xstep - \xbias, \ydisposition + -2*\ystep) {$\Event_1$};
\node[event] (e2) at (\xdisposition + 1*\xstep, \ydisposition + -2*\ystep) {};
\node[] (et2) at (\xdisposition + 1*\xstep + \xbias, \ydisposition + -2*\ystep) {$\Event_2$};

\node[event] (w2) at (\xdisposition + 0*\xstep, \ydisposition + -1*\ystep) {};
\node[] (w2t) at (\xdisposition + 0*\xstep - \xbias, \ydisposition + -1*\ystep) {$\Write$};
\node[event] (ow) at (\xdisposition + 1*\xstep, \ydisposition + -3*\ystep) {};
\node[] (owt) at (\xdisposition + 1*\xstep + \xbias, \ydisposition + -3*\ystep) {$\ov{\Write}$};

\node[event] (r2) at (\xdisposition + 0.4*\xstep, \ydisposition + -2.75*\ystep) {};
\node[] (r2t) at (\xdisposition + 0.4*\xstep , \ydisposition + -2.75*\ystep - \ybias) {$\LastFlow_{\Process}(\Write)$};

\draw[seqtrace] (t1a) to (t1b);
\draw[seqtrace] (t2a) to (t2b);

\draw[post] (e1) to (e2);
\draw[post, \darkred, dashed] (r2) to (ow);

\end{tikzpicture}
\caption{}
\label{subfig:algo_oc2}
\end{subfigure}
%\qquad
\begin{subfigure}{0.4\textwidth}
\small
\centering
\begin{tikzpicture}[thick,
pre/.style={<-,shorten >= 1pt, shorten <=1pt, thick},
post/.style={->,shorten >= 2pt, shorten <=2pt,  very thick},
seqtrace/.style={->, line width=2},
und/.style={very thick, draw=gray},
event/.style={rectangle, minimum height=0.8mm, minimum width=3mm, fill=black!100,  line width=1pt, inner sep=0},
virt/.style={circle,draw=black!50,fill=black!20, opacity=0}]

\newcommand{\xdisposition}{0}
\newcommand{\ydisposition}{0}
\newcommand{\xstep}{1.2}
\newcommand{\ystep}{0.55}
\newcommand{\ybias}{0.4}
\newcommand{\xbias}{0.4}

\node	[]		(t1a)	at	(\xdisposition + 0*\xstep, \ydisposition + 0*\ystep)	{\normalsize$\SeqTrace_1$};
\node	[]		(t1b)	at	(\xdisposition + 0*\xstep, \ydisposition + -4*\ystep)	{};
\node	[]		(t2a)	at	(\xdisposition + 1*\xstep, \ydisposition + 0*\ystep)	{\normalsize$\SeqTrace_2$};
\node	[]		(t2b)	at	(\xdisposition + 1*\xstep, \ydisposition + -4*\ystep)	{};

\node[event] (e1) at (\xdisposition + 0*\xstep, \ydisposition + -2*\ystep) {};
\node[] (et1) at (\xdisposition + 0*\xstep - \xbias, \ydisposition + -2*\ystep) {$\Event_1$};
\node[event] (e2) at (\xdisposition + 1*\xstep, \ydisposition + -2*\ystep) {};
\node[] (et2) at (\xdisposition + 1*\xstep + \xbias, \ydisposition + -2*\ystep) {$\Event_2$};

\node[event] (a1) at (\xdisposition + 0*\xstep, \ydisposition + -1*\ystep) {};
\node[] (a1t) at (\xdisposition + 0*\xstep - 1.4*\xbias, \ydisposition + -1*\ystep) {$\Acquire_1$};
\node[event] (r1) at (\xdisposition + 0*\xstep, \ydisposition + -3*\ystep) {};
\node[] (r1t) at (\xdisposition + 0*\xstep - 1.35*\xbias, \ydisposition + -3*\ystep) {$\Release_1$};

\node[event] (a2) at (\xdisposition + 1*\xstep, \ydisposition + -1*\ystep) {};
\node[] (a2t) at (\xdisposition + 1*\xstep + 1.4*\xbias, \ydisposition + -1*\ystep) {$\Acquire_2$};
\node[event] (r2) at (\xdisposition + 1*\xstep, \ydisposition + -3*\ystep) {};
\node[] (r2t) at (\xdisposition + 1*\xstep + 1.35*\xbias, \ydisposition + -3*\ystep) {$\Release_2$};

\draw[seqtrace] (t1a) to (t1b);
\draw[seqtrace] (t2a) to (t2b);

\draw[post] (e1) to (e2);
\draw[post, \darkred, dashed] (r1) to (a2);

\end{tikzpicture}
\caption{}
\label{subfig:algo_lc}
\end{subfigure}
\caption{
Illustration of $\ObsClosure(\Event_1, \Event_2)$ (\subref{subfig:algo_oc1}, \subref{subfig:algo_oc2}) and $\LockClosure(\Event_1, \Event_2)$ (\protect\subref{subfig:algo_lc}) for an edge $(\Event_1, \Event_2)$ added in $P$.
In each case, the dashed edge corresponds to the new ordering inserted in $P$.
Recall that $\Observation_{\Trace}(\Read)$ denotes the observation of $\Read$ in $\Trace$,
and $\LastFlow_{\Process}(\Write)$ denotes the last read event of process $\Process$ that observes $\Write$ in $\Trace$.
}
\label{fig:algo_closure}
\end{figure}
%\smallskip
\begin{algorithm}%[H]
\small
\SetInd{0.4em}{0.4em}
\DontPrintSemicolon
%\setstretch{1.05}
\caption{$\Closure$}\label{algo:closure}
\KwIn{
A trace $\Trace$, a partial order $P$ over a set $X$ s.t. $P$ respects $\Trace$ and is represented as a DAG $G=(V,E)$.
}
\KwOut{The closure of $P$, if it exists, otherwise $\bot$.}
\BlankLine
\tcp{Initialization - $P$ is represented as $k$ total orders $\{\SeqTrace_i\}_i$ with extra orderings between $\SeqTrace_i$}
Initialize the data structure $\DS$ for $G$\\
$\Queue\gets$ an empty worklist\\
\ForEach(\tcp*[f]{Push partial-order edges}){$\Event_1\in V$}{\label{line:po_edges}
\ForEach{$i\in [k]$}{
Let $\Event_2\gets \DS.\Successor(\Event_1, i)$\tcp*[f]{The first successor of $\Event_1$ in the total order $\SeqTrace_i$}\label{line:ds_successor}\\
$\DS.\Insert(\Event_1, \Event_2)$\tcp*[f]{Insert the edge in $\DS$}\\
$\ObsClosure(\Event_1,\Event_2)$\label{line:obs_closure_pre}\tcp*[f]{Resolve observations}\\
$\LockClosure(\Event_1,\Event_2)$\label{line:lock_closure_pre}\tcp*[f]{Resolve locks}
}
}
\BlankLine
\tcp{Main computation}
\While{$\Queue$ is not empty}{\label{line:while}
$(\ov{\Event}_1,\ov{\Event}_2)\gets \Queue.\Pop()$\label{line:q_pop}\\
\lIf(\tcp*[f]{Cycle formed, abort}){$\DS.\Query(\ov{\Event}_2, \ov{\Event}_1)=\True$}{\label{line:cycle_formed}
\Return{$\bot$}
}
\uIf(\tcp*[f]{Edge not present}){$\DS.\Query(\ov{\Event}_1,\ov{\Event}_2)=\False$}{\label{line:edge_exists}
$\DS.\Insert(\ov{\Event}_1, \ov{\Event}_2)$\label{line:ds_insert_edge}\tcp*[f]{Besides $\ov{\Event}_1\leq_{Q} \ov{\Event}_2$, inserts $O(k^2)$ transitive orderings}\\
\ForEach{$(\Event_1, \Event_2)$ inserted}{\label{line:inner_loop}
$\ObsClosure(\Event_1,\Event_2)$\label{line:obs_closure_comp}\tcp*[f]{Resolve observations}\\
$\LockClosure(\Event_1,\Event_2)$\label{line:lock_closure_comp}\tcp*[f]{Resolve locks}
}
}
}
\Return{$\DS$}\label{line:return}\tcp*[f]{At this point $\DS$ represents the closure of $P$}
\end{algorithm}
\begin{minipage}[t]{.5\textwidth}
\vspace{-0.4cm}
%\smallskip
\begin{algorithm}[H]
\small
\SetInd{0.4em}{0.4em}
\DontPrintSemicolon
%\setstretch{1.05}
\caption{$\ObsClosure(\Event_1,\Event_2)$}\label{algo:obs_closure}
\BlankLine
\ForEach{$x\in \Globals(\Trace)$}{
Let $\Read\gets \From_x^{\SysReads}(\Event_2)$\label{line:obs_get_read}\\
Let $\Write\gets \To_x^{\SysWrites}(\Event_1)$\\
\lIf{$\Observation_{\Trace}(\Read)\neq \Write$}{
$\Queue.\Push(\Write, \Observation_{\Trace}(\Read))$\label{line:obs_before}
}
Let $\ov{\Write}\gets \From_x^{\SysWrites}(\Event_2)$\\
\lForEach{$\Process\in \{\Process_i \}_i$}{
$\Queue.\Push(\LastFlow_{\Process}(\Write), \ov{\Write})$\label{line:obs_after}
}
}
\end{algorithm}
\end{minipage}
\qquad
\begin{minipage}[t]{.4\textwidth}
\vspace{-0.4cm}
%\smallskip
\begin{algorithm}[H]
\small
\SetInd{0.4em}{0.4em}
\DontPrintSemicolon
%\setstretch{1.05}
\caption{$\LockClosure(\Event_1,\Event_2)$}\label{algo:lock_closure}
\BlankLine
\ForEach{$\ell\in \Locks(\Trace)$}{
Let $\Acquire_1\gets\To_{\ell}^{\SysAcquires}(\Event_1)$\\
Let $\Release_1\gets\Match[\Trace]{\Acquire_1}$\\
Let $\Release_2\gets \From_l^{\SysReleases}(\Event_2)$\\
Let $\Acquire_2\gets \Match[\Trace]{\Release_2}$\\
$\Queue.\Push(\Release_1, \Acquire_2)$\label{line:lock}
}
\end{algorithm}
\end{minipage}
\smallskip\noindent{\em The algorithm $\Closure$.}
We now present $\Closure$ for computing the closure of a partial order $P$ over a set $X$, or concluding that $P$ is not feasible.
The algorithm maintains a partial order as a DAG represented by the data structure $\DS$.
Conceptually, $\DS$ consists of $k$ total orders, $\SeqTrace_1,\dots,\SeqTrace_k$, where $k=\Width(P)$, with some extra orderings that go across the $\SeqTrace_i$.
Each total order $\SeqTrace_i$ contains the events of process $\Process_i$ in $X$.
Initially $\DS$ represents $P$.
The main computation iterates over a worklist $\Queue$ which holds edges to be inserted in $\DS$.
Upon extracting such an edge $(\ov{\Event}_1, \ov{\Event}_2)$ from $\Queue$, the algorithm inserts the edge in $\DS$ using the operation $\DS.\Insert$.
This operation results in various edges $(\Event_1, \Event_2)$ inserted in the graph, transitively through $(\ov{\Event}_1, \ov{\Event}_2)$.
For every $(\Event_1, \Event_2)$, the algorithm calls methods $\ObsClosure$ and $\LockClosure$ to resolve any violation of observation and lock constraints created by the insertion of $(\Event_1, \Event_2)$.
\cref{fig:algo_closure} illustrates $\ObsClosure$ and $\LockClosure$.
\cref{algo:closure},  \cref{algo:obs_closure} and \cref{algo:lock_closure} give the description of $\Closure$, $\ObsClosure$ and $\LockClosure$, respectively.

\smallskip\noindent{\bf Correctness and complexity.}
%Now we give some intuition behind the correctness and complexity of $\Closure$.
It is rather straightforward that if $P$ has a closure $Q$, then 
for each $\Queue.\Push(\Event_1, \Event_2)$ operation performed by $\Closure$, we have $\Event_1<_{Q} \Event_2$.
It follows that if $\Closure$ returns $\bot$, then $P$ is unfeasible.
On the other hand, if $\Closure$ does not return $\bot$, then the partial order $Q$ stored in the data structure $\DS$ is the closure of $P$.
Indeed, each of the closure rules can only be violated by an ordering $\Event_1<_{Q} \Event_2$.
The algorithm guarantees that every such edge is processed by the methods $\ObsClosure$ and $\LockClosure$, and new edges will be inserted in $\DS$ according to the rules of \cref{fig:algo_closure}. 
After such edges have been inserted, the ordering $\Event_1<_{Q} \Event_2$ can no longer violate any of the conditions of closure.
Regarding the time complexity, the algorithm inserts at most $n^2$ edges in the partial order represented by $\DS$.
Using the algorithms for $\DS$ (see \cref{lem:data_structure} in \cref{sec:dag_reach}), for every edge inserted by the algorithm, identifying which other edges are imposed by the closure rules requires only $O(\log n)$ time.
We have the following theorem.

\smallskip
\begin{theorem}\label{them:closure}
Let $\Trace$ be a trace and $P$ a partial order over a feasible set $X\subseteq \Events{\Trace}$ such that $P$ respects $\Trace$.
$\Closure$ correctly computes the closure of $P$ and requires $O(n^2\cdot \log n)$ time.
\end{theorem}

\smallskip\noindent{\bf Incremental closure.}
In our race detection algorithm, we also make use of the following operation on partial orders.
Let $\Trace$ be a trace and $P$ a partial order over a feasible set $X\subseteq \Events{\Trace}$ and such that
$P$ is closed wrt $\Trace$ and is represented as a DAG using the data structure $\DS$.
Given a pair of events $\Event_1, \Event_2\in X$, we define the operation
$\InsertAndClose(P, \Event_1\to \Event_2)$ as follows.
We execute the algorithm $\Closure$ starting from \cref{line:ds_insert_edge},
performing a $\DS.\Insert(\Event_1, \Event_2)$.
Hence we perform the $\ObsClosure$ and $\LockClosure$ only for the new orderings added due to $(\Event_1, \Event_2)$.

\smallskip
\begin{restatable}{lemma}{closuredynamiccomplexity}\label{lem:dynamic_closure}
Let $\Sigma$ be a sequence of $\InsertAndClose$ operations.
Performing all operations of $\Sigma$ requires $O(n^2\cdot \log n + |\Sigma|\cdot \log n)$ time in total,
and produces a closed partial order.
\end{restatable}

\section{The Decision Problem of Race Detection}\label{sec:decision}

Here we present a polynomial-time algorithm for the decision problem of dynamic race detection,
i.e., given an input trace $\Trace$ and two events $\Event_1, \Event_2\in \Events{\Trace}$, 
decide whether $(\Event_1, \Event_2)$ is a predictable race of $\Trace$.
Our algorithm is sound but incomplete in general, 
and it becomes complete if the input trace contains events of only two processes.
To assist the reader, we provide an outline of this section.
\begin{compactenum}
\item First, we introduce the notion of relative causal cone of an event $\Event$, which is a subset of events of $\Trace$.
Intuitively, it can be thought of as a trace slice of $\Trace$ up to $\Event$.
When deciding whether $(\Event_1, \Event_2)$ is a race of $\Trace$, our algorithm tries to find a witness trace for the race, such that the witness consists of events of the causal cones of $\Event_1$ and $\Event_2$.
\item Second, we present our main algorithm $\RaceDecision$ for deciding whether $(\Event_1, \Event_2)$ is a predictable race of $\Trace$.
In high level, the algorithm computes the causal cones of $\Event_1$ and $\Event_2$, and constructs a partial order $P$ over events of the causal cones.
Afterwards, it computes the closure $Q$ of $P$ using algorithm $\Closure$ from the previous section,
and resolves certain orderings of conflicting events in $Q$ according to the trace order in $\Trace$.
Finally, it uses algorithm $\MinMaxAlgo$ to linearize $Q$ to a witness trace that exhibits the race.
We also prove that this process is \emph{sound}, i.e., if $(\Event_1, \Event_2)$ is reported as a race, then it is a true predictable race of $\Trace$.
\item Third, we prove that the above process is also \emph{complete} for traces that consist of events of two processes, i.e., it detects all predictable races.
\item Finally, we illustrate $\RaceDecision$ on a few examples.
\end{compactenum}

%\subsection{A Strongly Sound Algorithm for the Decision Problem}\label{subsec:decision}

\smallskip\noindent{\bf Relative causal cones.}
Given a trace $\Trace$, an event $\Event\in \Events{\Trace}$ and a process $\Process$, the \emph{causal past cone} $\RPast_{\Trace}(\Event, \Process)$ of $\Event$ \emph{relative to} $\Process$ in $\Trace$ is the smallest set that contains the following events:
\begin{compactenum}
\item\label{item:cp1} For every event $\Event'\in \Events{\Trace}$ with $\Event'<_{\TO}\Event$, we have that $\Event'\in \RPast_{\Trace}(\Event, \Process)$.
\item\label{item:cp2} For every pair of events $\Event_1\in \RPast_{\Trace}(\Event, \Process)$ and $\Event_2\in \Events{\Trace}$, 
if $\Event_2\leq_{\TO} \Event_1$ then $\Event_2\in \RPast_{\Trace}(\Event, \Process)$.
\item\label{item:cp3} For every read event $\Read\in \Reads{\RPast_{\Trace}(\Event, \Process)}$, we have that $\Obs_{\Trace}(\Read)\in \RPast_{\Trace}(\Event, \Process)$.
\item\label{item:cp4} For every lock-acquire event $\Acquire\in \RPast_{\Trace}(\Event, \Process)$, if $\Proc{\Acquire}\neq \Proc{\Event}$ and $\Proc{\Acquire}\neq \Process$, then $\Match[\Trace]{\Acquire}\in \RPast_{\Trace}(\Event, \Process)$.
\end{compactenum}
It is easy to verify that $\RPast_{\Trace}(\Event, \Process)$ is always observation-feasible but not necessarily lock-feasible.

\smallskip\noindent{\em Intuition and example on relative causal cones.}
In order to decide whether an event pair $(\Event_1, \Event_2)$ is a predictable race of an input trace $\Trace$,
we first need to decide the events that will constitute a witness trace $\Trace^*$ that exposes the race.
In our race-detection algorithm, we take this set to be $\RPast_{\Trace}(\Event_1, \Proc{\Event_2})\cup \RPast_{\Trace}(\Event_2, \Proc{\Event_1})$,
i.e., it is the causal past cone of each focal event relative to the process of the other focal event.
Conditions \ref{item:cp1}-\ref{item:cp3} ensure that the cones are closed wrt the program order, and the observation of every read event is present, which is required for $\Trace^*$ to be a correct reordering of $\Trace$.
The intuition behind condition~\ref{item:cp4} is a bit more subtle.
To avoid having two critical sections on the same lock open, we include the matching release event of every lock-acquire event.
However, this rule does not apply for the processes of the focal events, since for these processes the events we have to include in $\Trace^*$ are precisely the predecessors of the corresponding focal event.

Consider the input trace in \cref{fig:cones}, where our task is to detect the race $(\Event_2, \Event_{10})$.
We outline here the computation of the relative causal cones $\RPast_{\Trace}(\Event_1, \Proc{\Event_2})$ and $\RPast_{\Trace}(\Event_2, \Proc{\Event_1})$.
\cref{item:cp1} of relative causal cones leads to $\RPast_{\Trace}(\Event_2, \Process_3)=\{ \Event_1 \}$.
For $\RPast_{\Trace}(\Event_{10}, \Process_1)$, \cref{item:cp1} makes $\Event_9\in \RPast_{\Trace}(\Event_{10}, \Process_1)$.
Since $\Event_9$ is a read event, \cref{item:cp3} makes $\Observation_{\Trace}(\Event_9)=\Event_5\in \RPast_{\Trace}(\Event_{10}, \Process_1)$,
and then \cref{item:cp2} makes $\Event_4\in \RPast_{\Trace}(\Event_{10}, \Process_1)$.
Since $\Event_4$ is a lock-acquire event and $\Proc{\Event_4}\neq \Process_1, \Process_3$
(i.e., the process of $\Event_4$ is neither the process of $\Event_{10}$, nor the process relative to which we are computing the causal cone of $\Event_{10}$),
\cref{item:cp4} makes $\Match[\Trace]{\Event_4}=\Event_6\in \RPast_{\Trace}(\Event_{10}, \Process_1)$.
Hence, in the end, $\RPast_{\Trace}(\Event_{10}, \Process_1)=\{ \Event_9,\Event_6, \Event_5, \Event_4\}$.
\begin{figure}[!h]
\begin{subfigure}[b]{0.3\textwidth}
\centering
\footnotesize
\def\rownumber{}
\begin{tabular}[b]{@{\makebox[1.2em][r]{\rownumber\space}} | l | l | l |}
\normalsize{$\mathbf{\SeqTrace_1}$} & \normalsize{$\mathbf{\SeqTrace_2}$} & \normalsize{$\mathbf{\SeqTrace_3}$}
\gdef\rownumber{\stepcounter{magicrownumbers}\arabic{magicrownumbers}} \\
\hline
$\Acquire(\ell)$ & & \\
$\mathbf{\Write(x)}$ & & \\
$\Release(\ell)$ & & \\
& $\Acquire(\ell)$ & \\
& $\Write(y)$ & \\
& $\Release(\ell)$ & \\
$\Write(z)$ & & \\
& $\Read(z)$ & \\
& & $\Read(y)$\\
& & $\mathbf{\Write(x)}$\\
\hline
\end{tabular}
\caption{An input trace $\Trace$.}
\label{subfig:cone_trace}
\end{subfigure}
\qquad
\begin{subfigure}[b]{0.6\textwidth}
\centering
\small
\begin{tikzpicture}[thick,
pre/.style={<-,shorten >= 2pt, shorten <=2pt, very thick},
post/.style={->,shorten >= 2pt, shorten <=2pt,  very thick},
seqtrace/.style={->, line width=2},
und/.style={very thick, draw=gray},
event/.style={rectangle, minimum height=0.8mm, minimum width=3mm, fill=black!100,  line width=1pt, inner sep=0},
virt/.style={circle,draw=black!50,fill=black!20, opacity=0}]

\newcommand{\xdisposition}{0}
\newcommand{\ydisposition}{0}
\newcommand{\xstep}{1.6}
\newcommand{\ystep}{0.8}
\newcommand{\ybias}{0.4}
\newcommand{\xbias}{0.7}
\newcommand{\xbiassmall}{0.55}

\node[] at (0,0.5) {$\RPast_{\Trace}(\Event_2, \Process_3)=\{ \Event_1 \}$};

\node[] at (0,0) {$\RPast_{\Trace}(\Event_{10}, \Process_1)=\{ \Event_9,\Event_6, \Event_5, \Event_4\}$};
\node[] at (0,-1) {};

\end{tikzpicture}
\caption{The relative causal cones $\RPast_{\Trace}(\Event_2, \Process_3)$ and $\RPast_{\Trace}(\Event_{10}, \Process_1)$}
\label{subfig:cones_cones}
\end{subfigure}
\caption{
The relative causal cones when testing for a race $(\Event_2, \Event_{10})$.
}
\label{fig:cones}
\end{figure}

\smallskip\noindent{\bf The algorithm $\RaceDecision$.}
We now describe our algorithm for reporting whether $\Trace$ has a predictable race on a given pair $(\Event_1, \Event_2)$.
In words, the algorithm constructs a set $X$ that is the union of the causal cones of each $\Event_i$ relative to the process of $\Event_{3-i}$.
Afterwards, the algorithm  constructs a partial order $P$ that respects $X$, and computes the closure $Q$ of $P$.
Finally, the algorithm non-deterministically chooses some $i\in[2]$, and examines all events that belong to processes other than $\Process_i$.
For every two such events $\ov{\Event}_1, \ov{\Event}_2$, if they conflict and are unordered by $Q$, the algorithm orders them according to their order in $\Trace$.
If a cycle is created in $Q$ during this process, the algorithm returns $\False$.
Otherwise, at the end of this process, the set $X$ can be naturally partitioned into two sets $X_1, X_2$ such that $\Width(Q\Project X_1)=1$ and $Q\Project X_2$ is an M-trace.
The first set is $X_1\Project \Proc{\Event_i}$, i.e., it contains the events of the process in which $\Event_i$ belongs to, and thus is a totally ordered set under $Q$.
The second set is $X_2=X\setminus X_1$, and note that all pairs of conflicting events of $X_2$ are now ordered under $Q$.
Hence, according to \cref{them:maxmin}, the partial order $Q$ is linearizable to a valid trace,
and the algorithm returns $\True$.
See \cref{algo:racedecision} for a formal description.

The complexity of the algorithm is $O(n^2\cdot \log n)$, which is the time required for computing the closure of the partial order $Q$ in \cref{line:close} and \cref{line:insert_and_close} (due to \cref{them:closure} and \cref{lem:dynamic_closure}, respectively).

Note that the algorithm is sound, i.e., if it returns $\True$ then $(\Event_1, \Event_2)$ is a true predictable race of $\Trace$. On the other hand, the algorithm is incomplete in general, i.e., it might return $\False$ even though $(\Event_1, \Event_2)$ is a true predictable race of $\Trace$.
For example, in \cref{line:false2}, the algorithm orders some pairs of conflicting events $(\ov{\Event}_1, \ov{\Event}_2)$ in the same order as in the input trace.
Although these orderings are expected to work most of the time, this choice might not always be correct.
In such cases, a cycle will be created in $Q$, and the algorithm will return $\False$ (see~\cref{SEC:APPENDIX_INCOMPLETENESS} for an example).
As we discuss next, the algorithm becomes complete if the input trace consists of only two processes.

\smallskip
\begin{algorithm}%[H]
\small
\DontPrintSemicolon
\SetInd{0.4em}{0.4em}
\caption{$\RaceDecision$}\label{algo:racedecision}
\KwIn{
A trace $\Trace$ and two events $\Event_1, \Event_2\in \Events{\Trace}$ with $\Confl{\Event_1}{\Event_2}$.
}
\KwOut{$\True$ if $(\Event_1, \Event_2)$ is detected as a predictable race of $\Trace$.}
\BlankLine
%\tcp{If a cycle is formed in $Q$, return $\False$}
Let $X\gets \RPast_{\Trace}(\Event_1, \Proc{\Event_2}) \cup \RPast_{\Trace}(\Event_2, \Proc{\Event_1})$\label{line:construct_x}\\
\lIf(){$\{\Event_1, \Event_2\}\cap X \neq\emptyset$ or $X$ is not feasible}{\label{line:test_lock_feasible}
\Return{$\False$}\tcp*[f]{No race}\label{line:false1}
}
Let $P\gets \RespectPO{\Trace}{X}$\tcp*[f]{The weakest po that respects $\Trace$}\label{line:respect_po}\\
Let $Q\gets \Closure(\Trace, P, X)$\tcp*[f]{Trace-close $P$}\label{line:close}\\
\lIf(\tcp*[f]{Closure created a cycle, no race}){$Q=\bot$}{
\Return{$\False$}\label{line:false3}
}
Non-deterministically chose $i\in [2]$\tcp*[f]{In practice, try both $i=1$ and $i=2$}\label{line:nondetermistic_choice}\\
\While{$\exists \ov{\Event}_1, \ov{\Event}_2\in X\setminus \SysEvents_{\Proc{\Event_i}}$ s.t. $\Confl{\ov{\Event}_1}{\ov{\Event}_2}$ and $\Unordered{\ov{\Event}_1}{Q}{\ov{\Event}_2}$ and $\ov{\Event}_1<_{\Trace}\ov{\Event}_2$}{\label{line:mwidth1}
$Q\gets \InsertAndClose(Q, \ov{\Event}_1 \to \ov{\Event}_2)$\tcp*[f]{Order $\ov{\Event}_1\to \ov{\Event}_2$ in $Q$ and close $Q$}\label{line:insert_and_close}\\
\lIf(\tcp*[f]{Closure created a cycle, no race}){$Q=\bot$}{
\Return{$\False$}\label{line:false2}
}
}
\Return{$\True$}
\end{algorithm}

\smallskip\noindent{\bf Completeness for two processes.}
We now discuss the completeness properties of $\RaceDecision$ for reporting races on input traces of two processes.
Assume that $(\Event_1, \Event_2)$ is a race of the input trace.
The key insight is that \cref{line:mwidth1} of $\RaceDecision$ is not executed, as every pair of events $\ov{\Event}_1, \ov{\Event}_2$ in that line belong to the same process, and thus are already ordered. 
Up until that point, all orderings used in constructing the partial order $\RespectPO{\Trace}{X}$ and computing the closure of $\RespectPO{\Trace}{X}$ are necessarily present in every trace that witnesses the race $(\Event_1,\Event_2)$.
Hence the closure computation cannot return $\bot$, and $\RaceDecision$ returns $\True$.
The following theorem concludes the results of this section.

\smallskip
\begin{restatable}{theorem}{decision}\label{them:decision}
Let $\Trace$ be a trace of $k\geq 2$ processes, and $n=|\Events{\Trace}|$.
Let $\Event_1, \Event_2 \in \Events{\Trace}$ be two conflicting events of $\Trace$.
The algorithm $\RaceDecision$ requires $O(n^2\cdot \log n)$ time and soundly reports whether $(\Event_1, \Event_2)$ is a predictable race of $\Trace$.
If $k=2$, $\RaceDecision$ is also complete.
If $\RaceDecision$ reports a race, a witness trace can be constructed in $O(n\cdot \log n)$ time.
\end{restatable}

As there are $O(n^2)$ pairs of events in $\Trace$, \cref{them:decision} yields the following corollary.

\smallskip
\begin{corollary}\label{cor:function}
Let $\Trace$ be a trace of $k\geq 2$ processes.
There exists a sound algorithm that requires $O(n^4\cdot \log n)$ time and soundly reports predictable races of $\Trace$.
If $k=2$, the algorithm is also complete (i.e., it reports \emph{all} predictable races).
\end{corollary}

\smallskip
\begin{remark}\label{rem:complexity}
We note that the dependency of \cref{cor:function} on the number of variables $|\Globals|$ and number of threads $k$
is $O(|\Globals|\cdot k^2\cdot n^4\cdot \log n)$.
To keep the presentation simple, we have neglected the dependency on $|\Globals|$ and $k$ in the analysis of the algorithm.
\end{remark}

\subsection{Examples}\label{subsec:examples}

We now illustrate the algorithm $\RaceDecision$ on a few examples.

\begin{figure*}[!h]
\begin{subfigure}[b]{0.24\textwidth}
\centering
\footnotesize
\def\rownumber{}
\begin{tabular}[b]{@{\makebox[1.2em][r]{\rownumber\space}} | l | l | l|}
\normalsize{$\mathbf{\SeqTrace_1}$} & \normalsize{$\mathbf{\SeqTrace_2}$} & \normalsize{$\mathbf{\SeqTrace_3}$}
\gdef\rownumber{\stepcounter{magicrownumbers}\arabic{magicrownumbers}} \\
\hline
$\Write(x)$ & & \\
$\Acquire(\ell_1)$ & & \\
$\Acquire(\ell_2)$ & & \\
$\Write(x)$ & & \\
$\Release(\ell_2)$ & & \\
$\mathbf{\Write(y)}$ & & \\
$\Release(\ell_1)$ & & \\
& $\Acquire(\ell_1)$ & \\
& $\Write(x)$  & \\
& $\Release(\ell_1)$ & \\
& &$\Acquire(\ell_1)$ \\
& &$\Release(\ell_1)$ \\
& & $\Acquire(\ell_2)$ \\
& & $\Read(x)$ \\
& & $\Release(\ell_2)$\\
& & $\mathbf{\Read(y)}$ \\
\hline
\end{tabular}
\caption{Is $(\Event_6, \Event_{16})$ a race?}
\label{subfig:example1_trace}
\end{subfigure}
\quad
\begin{subfigure}[b]{0.44\textwidth}
\centering
\small
\begin{tikzpicture}[thick,
pre/.style={<-,shorten >= 2pt, shorten <=2pt, very thick},
post/.style={->,shorten >= 2pt, shorten <=2pt,  very thick},
seqtrace/.style={->, line width=2},
und/.style={very thick, draw=gray},
event/.style={rectangle, minimum height=3mm, draw=black, fill=white, minimum width=4.5mm,   line width=1pt, inner sep=0, font={\small}},
virt/.style={circle,draw=black!50,fill=black!20, opacity=0}]

\newcommand{\xdisposition}{0}
\newcommand{\ydisposition}{0}
\newcommand{\xstep}{1.7}
\newcommand{\ystep}{0.8}
\newcommand{\ybias}{0.4}
\newcommand{\xbias}{0.8}
\newcommand{\xbiassmall}{0.6}

\node	[]		(t1a)	at	(\xdisposition + 0*\xstep, \ydisposition + 0*\ystep)	{\normalsize$\SeqTrace_1$};
\node	[]		(t1b)	at	(\xdisposition + 0*\xstep, \ydisposition + -6*\ystep)	{};
\node	[]		(t2a)	at	(\xdisposition + 1*\xstep, \ydisposition + 0*\ystep)	{\normalsize$\SeqTrace_2$};
\node	[]		(t2b)	at	(\xdisposition + 1*\xstep, \ydisposition + -6*\ystep)	{};
\node	[]		(t3a)	at	(\xdisposition + 2*\xstep, \ydisposition + 0*\ystep)	{\normalsize$\SeqTrace_3$};
\node	[]		(t3b)	at	(\xdisposition + 2*\xstep, \ydisposition + -6*\ystep)	{};

\draw[seqtrace] (t1a) to (t1b);
\draw[seqtrace] (t2a) to (t2b);
\draw[seqtrace] (t3a) to (t3b);

\node[event] (e10) at (\xdisposition + 0*\xstep, \ydisposition + -1*\ystep) {$\Event_{1}$};
\node[] (et10) at (\xdisposition + 0*\xstep - \xbias, \ydisposition + -1*\ystep) {$\Write(x)$};
\node[event] (e11) at (\xdisposition + 0*\xstep, \ydisposition + -2*\ystep) {$\Event_{2}$};
\node[] (et11) at (\xdisposition + 0*\xstep - \xbias, \ydisposition + -2*\ystep) {$\Acquire(\ell_1)$};
\node[event] (e12) at (\xdisposition + 0*\xstep, \ydisposition + -3*\ystep) {$\Event_{3}$};
\node[] (et12) at (\xdisposition + 0*\xstep - \xbias, \ydisposition + -3*\ystep) {$\Acquire(\ell_2)$};
\node[event] (e13) at (\xdisposition + 0*\xstep, \ydisposition + -4*\ystep) {$\Event_{4}$};
\node[] (et13) at (\xdisposition + 0*\xstep - \xbiassmall, \ydisposition + -4*\ystep) {$\Write(x)$};
\node[event] (e14) at (\xdisposition + 0*\xstep, \ydisposition + -5*\ystep) {$\Event_{5}$};
\node[] (et14) at (\xdisposition + 0*\xstep - \xbias, \ydisposition + -5*\ystep) {$\Release(\ell_2)$};

\node[event] (e21) at (\xdisposition + 1*\xstep, \ydisposition + -2*\ystep) {$\Event_{8}$};
\node[] (et21) at (\xdisposition + 1*\xstep + \xbias, \ydisposition + -2*\ystep) {$\Acquire(\ell_1)$};
\node[event] (e22) at (\xdisposition + 1*\xstep, \ydisposition + -3*\ystep) {$\Event_{9}$};
\node[] (et22) at (\xdisposition + 1*\xstep + \xbiassmall, \ydisposition + -3*\ystep) {$\Write(x)$};
\node[event] (e23) at (\xdisposition + 1*\xstep, \ydisposition + -4*\ystep) {$\Event_{10}$};
\node[] (et23) at (\xdisposition + 1*\xstep + \xbias, \ydisposition + -4*\ystep) {$\Release(\ell_1)$};

\node[event] (e31) at (\xdisposition + 2*\xstep, \ydisposition + -1*\ystep) {$\Event_{11}$};
\node[] (et31) at (\xdisposition + 2*\xstep + \xbias, \ydisposition + -1*\ystep) {$\Acquire(\ell_1)$};
\node[event] (e32) at (\xdisposition + 2*\xstep, \ydisposition + -2*\ystep) {$\Event_{12}$};
\node[] (et32) at (\xdisposition + 2*\xstep + \xbias, \ydisposition + -2*\ystep) {$\Release(\ell_1)$};
\node[event] (e33) at (\xdisposition + 2*\xstep, \ydisposition + -3*\ystep) {$\Event_{13}$};
\node[] (et33) at (\xdisposition + 2*\xstep + \xbias, \ydisposition + -3*\ystep) {$\Acquire(\ell_2)$};
\node[event] (e34) at (\xdisposition + 2*\xstep, \ydisposition + -4*\ystep) {$\Event_{14}$};
\node[] (et34) at (\xdisposition + 2*\xstep + \xbiassmall, \ydisposition + -4*\ystep) {$\Read(x)$};
\node[event] (e35) at (\xdisposition + 2*\xstep, \ydisposition + -5*\ystep) {$\Event_{15}$};
\node[] (et35) at (\xdisposition + 2*\xstep + \xbias, \ydisposition + -5*\ystep) {$\Release(\ell_2)$};

\draw[post] (e22) to (e34);
\draw[post] (e23) to (e11);
\draw[post, bend right=25] (e32) to (e11);
\draw[post, bend left=25, dashed, draw=\darkred] (e34) to (e13);
\draw[post, bend left=20, dashed, draw=\darkred] (e35) to (e12);

\end{tikzpicture}
\caption{The partial order $P$ and its closure $Q$.}
\label{subfig:example1_po}
\end{subfigure}
\quad
\begin{subfigure}[b]{0.24\textwidth}
\centering
\footnotesize
\def\rownumber{}
\begin{tabular}[b]{@{\makebox[1.2em][r]{\rownumber\space}} | l | l | l|}
\normalsize{$\mathbf{\SeqTrace_1}$} & \normalsize{$\mathbf{\SeqTrace_2}$} & \normalsize{$\mathbf{\SeqTrace_3}$}
\gdef\rownumber{\stepcounter{magicrownumbers}\arabic{magicrownumbers}} \\
\hline
$\Write(x)$ & & \\
& $\Acquire(\ell_1)$ & \\
& $\Write(x)$  & \\
& $\Release(\ell_1)$ & \\
& & $\Acquire(\ell_1)$\\
& & $\Release(\ell_1)$\\
& & $\Acquire(\ell_2)$ \\
& & $\Read(x)$ \\
& & $\Release(\ell_2)$\\
$\Acquire(\ell_1)$ & & \\
$\Acquire(\ell_2)$ & & \\
$\Write(x)$ & & \\
$\Release(\ell_2)$ & & \\
$\mathbf{\Write(y)}$ & & \\
& & $\mathbf{\Read(y)}$ \\
\hline
\end{tabular}
\caption{The witness trace.}
\label{subfig:example1_witness}
\end{subfigure}
\caption{
(\protect\subref{subfig:example1_trace}) The input trace.
(\protect\subref{subfig:example1_po}) The partial order $P$ (solid edges) and its closure $Q$ (solid and dashed edges).
(\protect\subref{subfig:example1_witness}) The witness trace obtained by extending $\MinMaxAlgo(Q)$ with the racy pair $(\Event_6, \Event_{16})$.
}
\label{fig:example1}
\end{figure*}

\smallskip\noindent{\bf Example of a race (\cref{fig:example1}).}
Consider the trace $\Trace$ shown in \cref{subfig:example1_trace}, and the task is to decide whether $(\Event_6, \Event_{16})$ is a predictable race of $\Trace$.
The algorithm constructs the causal cones
\[
\RPast_{\Trace}(\Event_6, \Process_3) = \{\Event_i\}_{i=1}^5 \quad \text{and} \quad \RPast_{\Trace}(\Event_{16}, \Process_1) = \{\Event_i\}_{i=8}^{15}
\]
and the partial order $P$ that respects $\Trace$, shown in \cref{subfig:example1_po} in solid edges.
Afterwards, the algorithm computes the closure $Q$ of $P$ by inserting the dashed edges in \cref{subfig:example1_po}.
In particular, since for the write event $\Event_4$ we have $\Event_9<_{P} \Event_4$ and $\Event_9$ is observed by the read event $\Event_{14}$, we have $\Event_{14}<_{Q} \Event_4$
(i.e., this is an observation-closure edge).
After this ordering is inserted, for the lock-acquire event $\Event_{13}$ we have $\Event_{13} <_{Q} \Event_5$ and $\Event_5$ is a lock-release event on the same lock, we also have $\Event_{15}<_{Q}\Event_3$, where $\Event_{15}=\Match[\Trace]{\Event_{13}}$ and $\Event_3=\Match[\Trace]{\Event_5}$
(i.e., this is a lock-closure edge).
Now consider the nondeterministic choice made in \cref{line:nondetermistic_choice} of $\RaceDecision$ such that $i=1$.
The algorithm also orders $\Event_{10} <_{Q} \Event_{11}$ by performing $\InsertAndClose(Q, \Event_{10} \to \Event_{11})$.
Notice that, after this operation, $Q$ is a closed partial order.
In addition, the $\Write(x)$ event $\Event_1$ is unordered with the conflicting events $\Event_9$ and $\Event_{10}$, i.e., $Q$ is not an M-trace.
However, by taking $X_1=X\Project \Process_1$ and $X_2=X\setminus X_2$, we have that $\Width(Q\Project X_1)=1$ and $Q\Project X_2$ is an M-trace.
Hence, by \cref{them:maxmin} $Q$ is linearizable to a correct reordering, constructed as the max-min linearization $\Trace^*=\MinMaxAlgo(Q)$.
This illustrates the advantage of our technique over existing methods, as here a correct reordering is exposed even though the initial partial order $P$ contains several pairs of conflicting events that are unordered.
Indeed, this race is missed by all $\HB$, $\WCP$, $\DC$ and $\SHB$.
Finally, the witness trace in is constructed by extending $\Trace^*$ with the racy events $\Event_6$, $\Event_{16}$,
shown in \cref{subfig:example1_witness}.

\begin{figure*}[!h]
\begin{subfigure}[b]{0.4\textwidth}
\centering
\footnotesize
\def\rownumber{}
\begin{tabular}[b]{@{\makebox[1.2em][r]{\rownumber\space}} | l | l | l|}
\normalsize{$\mathbf{\SeqTrace_1}$} & \normalsize{$\mathbf{\SeqTrace_2}$} & \normalsize{$\mathbf{\SeqTrace_3}$}
\gdef\rownumber{\stepcounter{magicrownumbers}\arabic{magicrownumbers}} \\
\hline
& $\Write(y)$ & \\
$\Acquire(\ell)$ & & \\
$\Write(x)$ & & \\
$\Read(y)$ & & \\
$\mathbf{\Write(z)}$ & & \\
$\Release(\ell)$ & & \\
& $\Acquire(\ell)$ & \\
& $\Write(x)$ & \\
& $\Release(\ell)$ & \\
& & $\Read(x)$ \\
& & $\Write(y)$ \\
& & $\Read(x)$ \\
& & $\mathbf{\Read(z)}$ \\
\hline
\end{tabular}
\caption{Is $(\Event_5, \Event_{13})$ a race?}
\label{subfig:example2_trace}
\end{subfigure}
\qquad
\begin{subfigure}[b]{0.5\textwidth}
\small
\centering
\begin{tikzpicture}[thick,
pre/.style={<-,shorten >= 2pt, shorten <=2pt, very thick},
post/.style={->,shorten >= 2pt, shorten <=2pt,  very thick},
seqtrace/.style={->, line width=2},
und/.style={very thick, draw=gray},
event/.style={rectangle, minimum height=3mm, draw=black, fill=white, minimum width=4.5mm,   line width=1pt, inner sep=0, font={\small}},
virt/.style={circle,draw=black!50,fill=black!20, opacity=0}]

\newcommand{\xdisposition}{0}
\newcommand{\ydisposition}{0}
\newcommand{\xstep}{1.7}
\newcommand{\ystep}{0.8}
\newcommand{\ybias}{0.4}
\newcommand{\xbias}{0.8}
\newcommand{\xbiassmall}{0.6}

\node	[]		(t1a)	at	(\xdisposition + 0*\xstep, \ydisposition + 0*\ystep)	{\normalsize$\SeqTrace_1$};
\node	[]		(t1b)	at	(\xdisposition + 0*\xstep, \ydisposition + -5*\ystep)	{};
\node	[]		(t2a)	at	(\xdisposition + 1*\xstep, \ydisposition + 0*\ystep)	{\normalsize$\SeqTrace_2$};
\node	[]		(t2b)	at	(\xdisposition + 1*\xstep, \ydisposition + -5*\ystep)	{};
\node	[]		(t3a)	at	(\xdisposition + 2*\xstep, \ydisposition + 0*\ystep)	{\normalsize$\SeqTrace_3$};
\node	[]		(t3b)	at	(\xdisposition + 2*\xstep, \ydisposition + -5*\ystep)	{};

\draw[seqtrace] (t1a) to (t1b);
\draw[seqtrace] (t2a) to (t2b);
\draw[seqtrace] (t3a) to (t3b);

\node[event] (e11) at (\xdisposition + 0*\xstep, \ydisposition + -2*\ystep) {$\Event_{2}$};
\node[] (et11) at (\xdisposition + 0*\xstep - \xbias, \ydisposition + -2*\ystep) {$\Acquire(\ell)$};
\node[event] (e12) at (\xdisposition + 0*\xstep, \ydisposition + -3*\ystep) {$\Event_{3}$};
\node[] (et12) at (\xdisposition + 0*\xstep - \xbiassmall, \ydisposition + -3*\ystep) {$\Write(x)$};
\node[event] (e13) at (\xdisposition + 0*\xstep, \ydisposition + -4*\ystep) {$\Event_{4}$};
\node[] (et13) at (\xdisposition + 0*\xstep - \xbiassmall, \ydisposition + -4*\ystep) {$\Read(y)$};

\node[event] (e21) at (\xdisposition + 1*\xstep, \ydisposition + -1*\ystep) {$\Event_{1}$};
\node[] (et21) at (\xdisposition + 1*\xstep + \xbiassmall, \ydisposition + -1*\ystep) {$\Write(y)$};
\node[event] (e22) at (\xdisposition + 1*\xstep, \ydisposition + -2*\ystep) {$\Event_{7}$};
\node[] (et22) at (\xdisposition + 1*\xstep + \xbias, \ydisposition + -2*\ystep) {$\Acquire(\ell)$};
\node[event] (e23) at (\xdisposition + 1*\xstep, \ydisposition + -3*\ystep) {$\Event_{8}$};
\node[] (et23) at (\xdisposition + 1*\xstep + \xbiassmall, \ydisposition + -3*\ystep) {$\Write(x)$};
\node[event] (e24) at (\xdisposition + 1*\xstep, \ydisposition + -4*\ystep) {$\Event_{9}$};
\node[] (et24) at (\xdisposition + 1*\xstep + \xbiassmall, \ydisposition + -4*\ystep) {$\Release(\ell)$};

\node[event] (e31) at (\xdisposition + 2*\xstep, \ydisposition + -2*\ystep) {$\Event_{10}$};
\node[] (et31) at (\xdisposition + 2*\xstep + \xbiassmall, \ydisposition + -2*\ystep) {$\Read(x)$};
\node[event] (e32) at (\xdisposition + 2*\xstep, \ydisposition + -3*\ystep) {$\Event_{11}$};
\node[] (et32) at (\xdisposition + 2*\xstep + \xbiassmall, \ydisposition + -3*\ystep) {$\Write(y)$};
\node[event] (e33) at (\xdisposition + 2*\xstep, \ydisposition + -4*\ystep) {$\Event_{12}$};
\node[] (et33) at (\xdisposition + 2*\xstep + \xbiassmall, \ydisposition + -4*\ystep) {$\Read(x)$};

\draw[post] (e23) to (e31);
\draw[post] (e24) to (e11);
\draw[post, bend right=0, dashed, draw=\darkred] (e13) to (e32);
\draw[post, bend left=0, dashed, draw=\darkred] (e12) to (e23);

\end{tikzpicture}
\caption{The partial order $P$  is not feasible.}
\label{subfig:example2_po}
\end{subfigure}
\caption{
(\protect\subref{subfig:example2_trace}) The input trace.
(\protect\subref{subfig:example2_po}) The partial order $P$ (solid edges) and its closure $Q$ (solid and dashed edges).
}
\label{fig:example2}
\end{figure*}
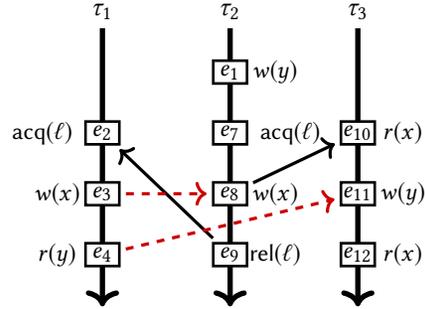
\smallskip\noindent{\bf Example of a non-race (\cref{fig:example2}).}
Consider the trace $\Trace$ shown in \cref{subfig:example2_trace}, and the task is to decide whether $(\Event_5, \Event_{13})$ is a predictable race of $\Trace$.
The algorithm constructs the causal cones
\[
\RPast_{\Trace}(\Event_5, \Process_3) = \{\Event_i\}_{i=1}^4 \quad  \text{and}\quad  \RPast_{\Trace}(\Event_{13}, \Process_1) = \{\Event_1\} \cup \{\Event_i\}_{i=7}^{12}
\]
and the partial order $P$ that respects $\Trace$, shown in \cref{subfig:example2_po} in solid edges.
Afterwards, the algorithm computes the closure $Q$ of $P$ by inserting the dashed edges in \cref{subfig:example2_po}.
Observe that $Q$ contains a cycle and thus $P$ is not feasible, hence the algorithm reports that $(\Event_4, \Event_{13})$ is not a race\footnote{While computing the closure different cycles might appear, depending on the order in which the closure rules are applied.}.
%Note that both $\HB$ and $\WCP$ falsely report that $(\Event_5, \Event_{13})$ is a race.

\smallskip\noindent{\bf Examples from \cref{SEC:INTRO}.}
Finally, we outline $\RaceDecision$ on the two races from \cref{fig:motivating}.

\smallskip\noindent{\em Example from \cref{subfig:motivating_incomplete1}.}
The algorithm constructs the causal cones
\[
\RPast_{\Trace}(\Event_2,\Process_2) = \{\Event_1 \} \quad \text{and} \quad \RPast_{\Trace}(\Event_7, \Process_1)=\{ \Event_i \}_{i=4}^6
\]
and the partial order $P$ that respects $\Trace$ by forcing the ordering $\Event_6<_{P} \Event_1$.
Note that $P$ is already closed, hence $Q=P$ by the algorithm $\Closure$.
Finally, the witness trace is constructed by obtaining the max-min linearization $\Trace^*=\MinMaxAlgo(Q)$ and extending $\Trace^*$ with the racy events $\Event_2$, $\Event_{7}$, thereby witnessing the race by the trace $\Trace^*=\Event_4, \Event_5, \Event_6, \Event_1, \Event_2, \Event_7$.

\smallskip\noindent{\em Example from \cref{subfig:motivating_incomplete2}.}
The algorithm constructs the causal cones
\[
\RPast_{\Trace}(\Event_2, \Process_3) = \{ \Event_1 \}
\quad\text{and}\quad
\RPast_{\Trace}(\Event_{14}, \Process_1)=\{\Event_i \}_{i=5}^{13}
\]
and the partial order $P$ that respects $\Trace$ by forcing the orderings $\Event_{10}<_{P} \Event_1$ and $\Event_{7}<_{P}\Event_{12}$.
Afterwards, the algorithm computes the closure of $Q$ by inserting $\Event_8<_{Q}\Event_{11}$.
Finally, the witness trace is constructed by obtaining the max-min linearization $\Trace^*=\MinMaxAlgo(Q)$ and extending $\Trace^*$ with the racy events $\Event_2$, $\Event_{14}$, thereby witnessing the race by the trace $\Trace^*=\Event_5, \Event_6, \Event_7, \Event_8, \Event_9, \Event_{10}, \Event_1, \Event_{11}, \Event_{12}, \Event_{13}, \Event_2, \Event_{14}$.

%\smallskip\noindent{\bf Further examples.}
%Finally, we refer to \cref{sec:dc_example} and \cref{SEC:APPENDIX_INCOMPLETENESS} for further examples, including one that illustrates the incompleteness of $\RaceDecision$ for $k\geq 3$ processes.

%\smallskip\noindent{\em Further examples.}
%Finally, we refer to the Appendix for two more interesting examples.
%\begin{compactenum}
%\item In \cref{sec:dc_example} we present an example that has been identified as difficult by~\cite{Roemer18} and missed by each of $\HB$, $\WCP$, $\DC$ and $\SHB$, but is nevertheless detected by our algorithm.
%The input trace consists of $7$ processes and $32$ events, and is thus rather large.
%Due to space restrictions, we present this example in the Appendix.
%
%\item Although \cref{them:decision} guarantees that $\RaceDecision$ is complete for input traces of two processes, completeness is not guaranteed for three or more processes.
%We refer to \cref{SEC:APPENDIX_INCOMPLETENESS} for an example.
%\end{compactenum}

%\input{comparison}
\section{The Function Problem in Practice}\label{sec:function}

\cref{cor:function} solves the function problem by solving the decision problem on every pair of events of the input trace.
Here we present an explicit algorithm for the function problem, called $\RaceFunction$, which is the main contribution of this work.
Although $\RaceFunction$ does not improve the worst-case complexity, it is faster in practice.
The algorithm relies on the following simple lemma.

\begin{restatable}{lemma}{functioncorrectness}\label{lem:function_correctness}
Consider two conflicting events $\Event_1, \Event_2$ and
let $X=\RPast_{\Trace}(\Event_1, \Proc{\Event_2}) \cup \RPast_{\Trace}(\Event_2, \Proc{\Event_1})$.
If $X\cap \{\Event_1, \Event_2 \}=\emptyset$ and $\OpenAcquires_{\Trace}(X)=\emptyset$ then $(\Event_1, \Event_2)$ is a predictable race of $\Trace$.
\end{restatable}

Intuitively, if the conditions of \cref{lem:function_correctness} are met, we can postpone the execution of $\Event_1$ in $\Trace$ until $\Event_2$, and the trace witnessing the race is simply $\Trace\Project X \circ \Event_1, \Event_2$.

\smallskip\noindent{\bf The algorithm $\RaceFunction_{\Event_1}^{\Process}$.}
We are now ready to describe an algorithm for partially solving the function problem on an input trace $\Trace$.
In particular, we present the algorithm $\RaceFunction_{\Event_1}^{\Process}$ given an event $\Event_1\in \Events{\Trace}$ and a process $\Process\neq \Proc{\Event_1}$.
The algorithm returns the set $\Races \subseteq \{\Event_1\}\times \Events{\Trace}\Project \Process$ of races detected between $\Event_1$ and events of process $\Process$.
The algorithm simply iterates over all events $\Event_2$ of $\Process$ in increasing order and computes the causal cone $\RPast_{\Trace}(\Event_2, \Proc{\Event_1})$.
Let $X=\RPast_{\Trace}(\Event_1, \Proc{\Event_2})\cup \RPast_{\Trace}(\Event_2, \Proc{\Event_1})$.
If there are open lock-acquire events in $X$, the algorithm invokes $\RaceDecision$ for solving the decision problem on $(\Event_1, \Event_2)$.
Otherwise, a race $(\Event_1, \Event_2)$ is directly inferred, due to \cref{lem:function_correctness}.
\cref{algo:racefunction} gives the formal description.
%\smallskip
\begin{algorithm}%[H]
\small
\DontPrintSemicolon
\SetInd{0.4em}{0.4em}
\caption{$\RaceFunction_{\Event_1}^{\Process}$}\label{algo:racefunction}
\KwIn{
A trace $\Trace$, an event $\Event_1\in \Events{\Trace}$, a process $\Process \neq \Proc{\Event_1}$.
}
\KwOut{A set $\Races \subseteq \{\Event_1\}\times \Events{\Trace}$ of predictable races of $\Trace$.
}
\BlankLine
Let $\Races\gets \emptyset$\\
Let $X\gets \RPast_{\Trace}(\Event_1, \Process)$\\
\ForEach{$\Event_2\in \SysEvents(\Process)$ in increasing order of $<_{\TO(\Trace)}$ s.t. $\Event_2 \not \in X$ and $\Confl{\Event_1}{\Event_2}$}{\label{line:functioninner_loop}
Insert $\RPast_{\Trace}(\Event, \Proc{\Event_1})\setminus X$ in $X$\label{line:update_cone}\tcp*[f]{At this point $X=\RPast_{\Trace}(\Event_1, \Process)\cup \RPast_{\Trace}(\Event, \Proc{\Event_1})$}\\
\lIf(\tcp*[f]{$\Event_1\in X$ for all remaining $\Event_2$, return early}){$\Event_1\in X$}{
\Return{$\Races$}
}
\eIf(\tcp*[f]{No open locks, race found}){$\OpenAcquires_{\Trace}(X) = \emptyset$}{\label{line:function_decision}
Insert $(\Event_1, \Event_2)$ in $\Races$\\
}{\tcp{Open locks, use the decision algorithm }
\lIf{$\RaceDecision(\Event_1, \Event_2)$}{
Insert $(\Event_1, \Event_2)$ in $\Races$\label{line:insert_race1}
}
}
}
\end{algorithm}

The efficiency of $\RaceFunction_{\Event_1}^{\Process}$ lies on two observations:
\begin{compactenum}
\item Since critical sections tend to be small, we expect the condition in \cref{line:function_decision} to be $\False$ only a few times,
thus \cref{lem:function_correctness} allows to soundly report a race without constructing a partial order. 
\item The causal cones are closed wrt the program order.
That is, given two events $\Event_2$ and $\Event'_2$ with $\Event_2<_{\TO}\Event'_2$, we have that  $\RPast_{\Trace}(\Event_2, \Proc{\Event_1})\subset \RPast_{\Trace}(\Event'_2, \Proc{\Event_1})$, and thus we only need to consider the difference $\RPast_{\Trace}(\Event'_2, \Proc{\Event_1})\setminus \RPast_{\Trace}(\Event_2, \Proc{\Event_1})$ in \cref{line:update_cone}.
This decreases the total time for constructing all causal cones from quadratic to linear.
\end{compactenum}

\smallskip\noindent{\bf The algorithm $\RaceFunction$.}
Finally, we outline the algorithm $\RaceFunction$ for solving the function problem.
Given an input trace $\Trace$, the algorithm simply invokes $\RaceFunction_{\Event_1}^{\Process}$ for every event $\Event_1\in \Events{\Trace}$ and process $\Process\neq \Proc{\Event_1}$ and obtains the returned race set $\Races_{\Event_1}^{\Process}$.
Since there are $O(n)$ such events, the algorithm makes $O(n)$ invocations.
The reported set of predictable races of $\Trace$ is then $\Races=\bigcup_{\Event_1, \Process}\Races_{\Event_1}^{\Process}$.

\smallskip\noindent{\bf Detecting completeness dynamically.}
Assume that we execute $\RaceDecision$ on input events $\Event_1, \Event_2$ and the algorithm returns $\False$.
It can be easily shown that if the following conditions hold, then $(\Event_1, \Event_2)$ is not a predictable race of $\Trace$ (and hence correctly rejected by $\RaceDecision$).
\begin{compactenum}
\item\label{item:complete1} When computing the relative causal past cones in \cref{line:construct_x} of $\RaceDecision$, no event is added to the cones due to \cref{item:cp4} of the definition of relative causal past cones.
\item\label{item:complete2} $\RaceDecision$ returns $\False$ before executing \cref{line:insert_and_close}.
\end{compactenum}
Let $\MayRaces$ be the set of races that are rejected by $\RaceFunction$ on input trace $\Trace$ such that at least one of the conditions above does not hold.
It follows that $\MayRaces$ over-approximates the set of false negatives of $\RaceFunction$.
The algorithm is \emph{dynamically complete} for $\Trace$ if $\MayRaces_{\Trace}=\emptyset$.

\section{Experiments}\label{SEC:EXPERIMENTS}

In this section we report on an implementation and experimental evaluation of our techniques.

\subsection{Implementation}\label{subsec:implementation}

We have implemented our algorithm $\RaceFunction$ in Java and evaluated its performance on a standard set of benchmarks.
We first discuss some details of the implementation.
%Here we discuss some details related to the implementation.

\smallskip\noindent{\bf Handling dynamic processes creation.}
In the theoretical part of this paper we have neglected dynamic process creation events.
In practice such events are common, and all our benchmark traces contain $\ForkT(i)$ and $\JoinT(j)$ events
(for forking process $\Process_i$ and joining with process $\Process_j$, respectively).
To handle such events, we include in the program order $\TO$ the following order relationships:
If $\Event_1$ is a $\ForkT(i)$ event and $\Event_2$ is a $\JoinT(j)$ event, then we include the order relationship
$(\Event_1 <_{\TO} \Event'_1)$ and $(\Event'_2<_{\TO} \Event_2)$, where $\Event'_1$ is the first event of $\Process_i$
and $\Event'_2$ is the last event of $\Process_j$ such that $\Event'_2<_{\Trace} \Event_2$.
%Note that this modified program order also affects the causal past cones of events.

\smallskip\noindent{\bf Optimizations.}
We make two straightforward optimizations, namely, ignoring non-racy locations and over-approximating the racy events.
Recall that $\RespectPO{\Trace}{X}$ is the weakest partial order over the events of $\Trace$ that respects $\Trace$.
$\RespectPO{\Trace}{X}$ can be constructed efficiently by a single pass of $\Trace$.
For the first optimization, we simply remove from $\Trace$ all events to a location $x$ if every pair of conflicting events on $x$ is ordered by $\RespectPO{\Trace}{X}$.
For the second optimization, we construct the set
\begin{align*}
A=\{ (\Event_1, \Event_2):~ &\Confl{\Event_1}{\Event_2} \quad \text{and} \quad \Unordered{\Event_1}{\RespectPO{\Trace}{X}}{\Event_2} \quad \text{and}
&\Event_1, \Event_2 \text{ are not protected by the same lock} \}.
\end{align*}
which over-approximates the races of $\Trace$,
and, we only consider the pairs $(\Event_1, \Event_2)\in A$ for races.

\subsection{Experimental Setup}\label{subsec:experimental_setup}

\smallskip\noindent{\bf Benchmarks.}
Our benchmark set is a standard one found in recent works on race detection~\cite{Huang14,Kini17,Yu18,Mathur18}, and parts of it also exist in other works~\cite{Flanagan09,Bond10,Zhai12,Roemer18}.
It contains concurrent traces of various sizes, which are concrete executions of concurrent programs taken from standard benchmark suits:
(i)~the IBM Contest benchmark suite~\cite{Farchi03},
(ii)~the Java Grande forum benchmark suite~\cite{Smith01},
(iii)~the DaCapo benchmark suite~\cite{Blackburn06}, and
(iv)~some standalone, real-world software.
We have also included the benchmark \texttt{cryptorsa} from the SPEC JVM08 benchmark suite~\cite{specjvm08}
which we have found to be racy.
We refer to \cref{tab:stats} for various interesting statistics on each benchmark trace.
The columns $k$ and $n$ denote the number of processes and number of events in each input trace.
In each case, $k$ is also used as the bound of the width of the partial orders constructed by our algorithm.

\begin{table}
\small
\centerline{
\begin{tabular}{|c||c|c|c|c||||c||c|c|c|c|}
\hline
\textbf{Benchmark} & \textbf{$n$} & \textbf{$k$} & \textbf{\# variables} & \textbf{\# locks} & \textbf{Benchmark} & \textbf{$n$} & \textbf{$k$} & \textbf{\# variables} & \textbf{\# locks}\\
\hline
\hline
\texttt{array} & 44 & 2 & 30 & 2 & \texttt{moldyn} & 164K & 2 & 1.0K & 2\\
\hline
\texttt{critical} & 49 & 3 & 30 & 0 & \texttt{derby} & 1.0M & 3 & 185K & 1.0K\\
\hline
\texttt{airtickets} & 116 & 2 & 46 & 0 & \texttt{jigsaw} & 3.0M & 13 & 103K & 280\\
\hline
\texttt{account} & 125 & 3 & 41 & 3 & \texttt{bufwriter} & 11M & 5 & 56 & 1\\
\hline
\texttt{pingpong} & 126 & 4 & 54 & 0 & \texttt{hsqldb} & 18M & 43 & 946K & 412\\
\hline
\texttt{bbuffer} & 322 & 2 & 73 & 2 & \texttt{cryptorsa} & 57M & 7 & 1.0M & 8.0K\\
\hline
\texttt{mergesort} & 3.0K & 4 & 621 & 3 & \texttt{eclipse} & 86M & 14 & 10M & 8.0K\\
\hline
\texttt{bubblesort} & 4.0K & 10 & 196 & 3 & \texttt{xalan} & 122M & 6 & 4.0M & 2.0K\\
\hline
\texttt{raytracer} & 16K & 2 & 3.0K & 8 & \texttt{lusearch} & 216M & 7 & 5.0M & 118\\
\hline
\texttt{ftpserver} & 48K & 10 & 5.0K & 304 & - & - & - & - & -\\
\hline
\end{tabular}
}
\caption{Statistics on our benchmark set.}
\label{tab:stats}
\end{table}

\smallskip\noindent{\bf Comparison with $\HB$, $\WCP$, $\DC$ and $\SHB$.}
We compare $\RaceFunction$ against the standard $\HB$, as well as $\WCP$~\cite{Kini17}, $\DC$~\cite{Roemer18}, and $\SHB$~\cite{Mathur18} which, to our knowledge, are the most recent advances in race prediction.
These are partial-order methods based on vector clocks.
All implementations are in Java: we rely on the tool Rapid~\cite{Mathur18} for running $\HB$, $\WCP$ and $\SHB$, and on our own implementation of $\DC$.
To obtain all race reports for an input trace $\Trace$, we use the following process.
%(similarly to~\cite{Kini17,Roemer18,Mathur18}).
We construct the corresponding partial order incrementally, by inserting new events in the order they appear in $\Trace$.
In addition, we use an extra vector clock $R_x$, $W_x$, for every location $x$,
which records the vector clock of the last read and write event, respectively, that accessed the respective location.
These vector clocks are used to determine whether the current event is racy.
After inserting an event $\Event_1$ in the partial order, we iterate over each conflicting event $\Event_2$ that precedes $\Event_1$ in $\Trace$,
and determine whether $\Event_2< \Event_1$. If not, we report a race $(\Event_1, \Event_2)$, and join the vector clock of $\Event_1$ with the vector clock of the corresponding location.

As $\HB$ and $\WCP$ are only sound on the first race, and $\DC$ is unsound, the above process creates, in general, false positives. 
In order to have a basis for comparison on sound reports, $(\Event_1, \Event_2)$ is regarded as a reported race by each of these methods only if $\RaceFunction$ reports it either as a race, or a possibly false negative (i.e., either $(\Event_1, \Event_2)\in \Races$ or $(\Event_1, \Event_2)\in \MayRaces$, the sets $\Races$ and $\MayRaces$ as defined in \cref{sec:function})
\footnote{$\SHB$ is sound on all race reports, and hence this filtering is not performed.}.

\smallskip\noindent{\bf Race reports.}
Recall that a race is defined as a pair of events $(\Event_1, \Event_2)$ of the input trace.
However, the interest of the programmer is on the actual code lines $(l_1, l_2)$ that these events correspond to.
Since long traces typically come from code that executes repeatedly in a loop, 
we expect to have many racy pairs of events  $\{(\Event^i_1, \Event^i_2)\}_i$ that correspond to the same racy pair of code lines $(l_1, l_2)$, and hence all such pairs $(\Event^i_1, \Event^i_2)$ require a single fix.
Hence, although the input trace might contain many different event pairs that correspond to the same line pair, these will result in a single race report. 
%Due to space limitations, we refer to \cref{sec:missing_experiments} for race reports on pairs of events instead of pairs of code lines.

\subsection{Experimental Results}\label{subsec:experimental_results}

Our evaluation is summarized in \cref{tab:experimental}.
The columns $\PR$ and $\Time$ show the number of reported races, and the time taken, respectively by each method.
The column $\FalseNegativesBound$ reports the size of the set $\MayRaces$, which gives an upper-bound on the number of false negatives of $\RaceFunction$ (see \cref{sec:function}).

\smallskip\noindent{\bf Race detection capability.}
We see that $\RaceFunction$ is very effective:~overall, it discovers hundreds of real races on all benchmarks, regardless of their size and number of processes.
%As a sanity check, for every race reported by our algorithm, we have produced a witness trace $\Trace^*$ and have verified that 
%(i)~$\Trace^*$ is a correct reordering of the input trace and
%(ii)~$\Trace^*$ exhibits the race.
In addition, $\RaceFunction$ is found complete on \emph{all} benchmarks
(i.e., our over-approximation of the false negatives in column $\FalseNegativesBound$ always reports at most $0$ false negatives).
Hence, $\RaceFunction$ manages to detect all races in our benchmark set.
To our knowledge, this is the first sound technique that reaches such a level of completeness.

On the other hand, the capability of $\HB$, $\WCP$, $\DC$ and $\SHB$ is more limited, as they all miss several races on several benchmarks.
We observe that $\WCP$ catches more races than $\HB$, and $\DC$ more races than $\WCP$.
This is predicted by theory, as $\HB$ races are $\WCP$ races~\cite{Kini17}, and $\WCP$ races are $\DC$ races~\cite{Roemer18}.
On the other hand, although $\SHB$ captures provably more races than $\HB$, $\SHB$ is incomparable with $\DC$.
In either case, $\RaceFunction$ captures more races
than $\DC$ on 10 benchmarks, and
than $\SHB$ on 6 benchmarks.
In addition, on 5 benchmarks (shown in bold), $\RaceFunction$ captures more races than any other algorithm.
In total, $\RaceFunction$ detects 71 more races than $\DC$ and 25 more races than $\SHB$.

We also remark that our algorithm provides more information than the baseline methods even on benchmarks where the number of reported races is the same.
This is because $\RaceFunction$ manages to detect that the reports in such benchmarks are complete
(i.e., no races are missed).
%Such information is not reported by the baseline methods.

\begin{table}
\small
\centerline{
\begin{tabular}{|c|||c|c||c|c||c|c||c|c||c|c|c|}
\hline
\textbf{Benchmark} & \multicolumn{2}{c||}{$\boldsymbol{\HB}$}& \multicolumn{2}{c||}{$\boldsymbol{\WCP}$}& \multicolumn{2}{c||}{$\boldsymbol{\DC}$}& \multicolumn{2}{c||}{$\boldsymbol{\SHB}$}& \multicolumn{3}{c|}{$\boldsymbol{\RaceFunction}$}\\
\hline
&$\boldsymbol{\PR}$ & $\boldsymbol{\Time}$ & $\boldsymbol{\PR}$ & $\boldsymbol{\Time}$ & $\boldsymbol{\PR}$ & $\boldsymbol{\Time}$ & $\boldsymbol{\PR}$ & $\boldsymbol{\Time}$ & $\boldsymbol{\PR}$ & $\boldsymbol{\FalseNegativesBound}$ & $\boldsymbol{\Time}$ \\
\hline
\hline
\texttt{array} & 0 & 0.30s & 0 & 0.28s & 0 & 2.12s & 0 & 0.29s & 0 & 0 & 0.12s\\
\hline
\texttt{critical} & 3 & 0.29s & 3 & 0.30s & 3 & 2.11s & 8 & 0.28s & 8 & 0 & 0.10s\\
\hline
\texttt{airtickets} & 3 & 0.32s & 3 & 0.31s & 3 & 2.10s & 4 & 0.31s & 4 & 0 & 0.12s\\
\hline
\texttt{account} & 1 & 0.31s & 1 & 0.32s & 1 & 2.12s & 1 & 0.30s & 1 & 0 & 0.12s\\
\hline
\texttt{pingpong} & 2 & 0.30s & 2 & 0.31s & 2 & 2.04s & 2 & 0.31s & 2 & 0 & 0.09s\\
\hline
\texttt{bbuffer} & 2 & 0.30s & 2 & 0.31s & 2 & 2.12s & 2 & 0.31s & 2 & 0 & 0.09s\\
\hline
\texttt{mergesort} & 1 & 0.36s & 1 & 0.41s & 1 & 2.16s & 1 & 0.37s & \textbf{2} & 0 & 0.18s\\
\hline
\texttt{bubblesort} & 4 & 0.46s & 4 & 0.54s & 5 & 2.28s & 6 & 0.62s & 6 & 0 & 0.71s\\
\hline
\texttt{raytracer} & 3 & 0.51s & 3 & 0.56s & 3 & 2.57s & 3 & 0.51s & 3 & 0 & 0.23s\\
\hline
\texttt{ftpserver} & 23 & 0.79s & 23 & 1.28s & 24 & 2.75s & 23 & 0.73s & \textbf{26} & 0 & 0.88s\\
\hline
\texttt{moldyn} & 2 & 1.50s & 2 & 1.81s & 2 & 3.88s & 2 & 1.52s & 2 & 0 & 1.08s\\
\hline
\texttt{derby} & 12 & 8.53s & 12 & 14.54s & 12 & 15.29s & 12 & 8.32s & 12 & 0 & 7.84s\\
\hline
\texttt{jigsaw} & 8 & 17.51s & 10 & 21.80s & 10 & 40.89s & 9 & 17.93s & \textbf{11} & 0 & 14.65s\\
\hline
\texttt{bufwriter} & 2 & 48.64s & 2 & 2m0s & 2 & 2m59s & 2 & 47.71s & 2 & 0 & 57.37s\\
\hline
\texttt{hsqldb} & 4 & 3m53s & 4 & 3m5s & 5 & 4m23s & 9 & 3m53s & 9 & 0 & 7m1s\\
\hline
\texttt{cryptorsa} & 5 & 3m42s & 5 & 3m0s & 7 & 6m58s & 5 & 3m29s & 7 & 0 & 6m6s\\
\hline
\texttt{eclipse} & 33 & 8m1s & 34 & 7m0s & 39 & 14m44s & 54 & 7m11s & \textbf{67} & 0 & 45m23s\\
\hline
\texttt{xalan} & 7 & 8m58s & 7 & 8m25s & 9 & 20m12s & 11 & 9m8s & \textbf{15} & 0 & 7m15s\\
\hline
\texttt{lusearch} & 30 & 16m4s & 30 & 9m59s & 30 & 2h49m6s & 52 & 15m28s & 52 & 0 & 8m9s\\
\hline
\hline
\textbf{Total} & \textbf{145} & \textbf{42m0s} & \textbf{148} & \textbf{34m14s} & \textbf{160} & \textbf{3h39m} & \textbf{206} & \textbf{40m31s} & \textbf{231} & \textbf{0} & \textbf{1h15m}\\
\hline
\end{tabular}
}
\caption{
Experimental comparison between $\HB$, $\WCP$, $\DC$, $\SHB$ and our algorithm $\RaceFunction$.
The column $\FalseNegativesBound$ shows an upper bound on the number of races missed by $\RaceFunction$.
}
\label{tab:experimental}
\end{table}
\begin{table}
\parbox{.45\linewidth}{
\small
\centerline{
\begin{tabular}{|c||c|c|c|}
\hline
\textbf{Benchmark} & \textbf{Mean Distance} & \textbf{Max Distance}\\
\hline
\hline
%\texttt{array} & 0 & 0\\
%\hline
%\texttt{critical} & 10 & 16\\
%\hline
%\texttt{airtickets} & 17 & 60\\
%\hline
%\texttt{account} & 10 & 11\\
%\hline
%\texttt{pingpong} & 46 & 97\\
%\hline
%\texttt{bbuffer} & 139 & 214\\
%\hline
%\texttt{mergesort} & 939 & 1.0K\\
%\hline
%\texttt{bubblesort} & 1.0K & 3.0K\\
%\hline
%\texttt{raytracer} & 382 & 880\\
%\hline
\texttt{ftpserver} & 939 & 12K\\
%\hline
%\texttt{moldyn} & 3.0K & 10K\\
%\hline
%\texttt{derby} & 434 & 3.0K\\
\hline
\texttt{jigsaw} & 1K & 4K\\
%\hline
%\texttt{bufwriter} & 42 & 130\\
\hline
\texttt{hsqldb} & 92K & 1M\\
\hline
\texttt{cryptorsa} & 1M & 8M\\
\hline
\texttt{eclipse} & 11M & 53M\\
\hline
\texttt{xalan} & 2.0K & 43K\\
\hline
\texttt{lusearch} & 44M & 125M\\
\hline
\end{tabular}
}
\caption{Mean and maximum distances (in number of intervening events) on races detected by $\RaceFunction$.}
\label{tab:distances}
}
\hfill
\parbox{.45\linewidth}{
\small
\centerline{
\begin{tabular}{|c||c|c|c|c|c|}
\hline
\textbf{Benchmark} &$\boldsymbol{\HB}$& $\boldsymbol{\WCP}$& $\boldsymbol{\DC}$& $\boldsymbol{\SHB}$& $\boldsymbol{\RaceFunction}$\\
\hline
\hline
%\texttt{array} & 0 & 0 & 0 & 0 & 0\\
%\hline
%\texttt{critical} & 0 & 0 & 0 & 0 & 0\\
%\hline
%\texttt{airtickets} & 0 & 0 & 0 & 0 & 0\\
%\hline
%\texttt{account} & 0 & 0 & 0 & 0 & 0\\
%\hline
%\texttt{pingpong} & 0 & 0 & 0 & 0 & 0\\
%\hline
%\texttt{bbuffer} & 0 & 0 & 0 & 0 & 0\\
%\hline
%\texttt{mergesort} & 0 & 0 & 0 & 0 & 0\\
%\hline
%\texttt{bubblesort} & 0 & 0 & 0 & 0 & 0\\
%\hline
%\texttt{raytracer} & 0 & 0 & 0 & 0 & 0\\
%\hline
\texttt{ftpserver} & 3 & 3 & 3 & 3 & 0\\
\hline
%\texttt{moldyn} & 0 & 0 & 0 & 0 & 0\\
%\hline
%\texttt{derby} & 0 & 0 & 0 & 0 & 0\\
%\hline
\texttt{jigsaw} & 2 & 0 & 0 & 2 & 0\\
\hline
%\texttt{bufwriter} & 0 & 0 & 0 & 0 & 0\\
%\hline
\texttt{cryptorsa} & 1 & 1 & 0 & 1 & 0\\
\hline
\texttt{eclipse} & 16 & 11 & 8 & 10 & 0\\
\hline
\texttt{xalan} & 3 & 2 & 2 & 2 & 0\\
\hline
\texttt{lusearch} & 11 & 11 & 11 & 0 & 0\\
\hline
\hline
\textbf{Total} & \textbf{36} & \textbf{28} & \textbf{24} & \textbf{18} & \textbf{0}\\
\hline
\end{tabular}
}
\caption{
Racy locations missed by each method.
For $\HB$, $\WCP$, $\DC$, $\SHB$ the numbers are lower-bounds.
For $\RaceFunction$, the numbers are upper-bounds.
}
\label{tab:missed_locations}
}
\end{table}

In terms of race distances, we have found that $\RaceFunction$ is able to detect races that are very far apart in the input trace.
We refer to \cref{tab:distances} for a few interesting examples,
where the distance of a race $(\Event_1, \Event_2)$ is counted as the number of intervening events between $\Event_1$ and $\Event_2$ in the input trace.
For instance, in \texttt{lusearch}, the maximum race distance detected by $\RaceFunction$ is 125M events.
Note that, in general, the same memory location can be reported as racy by many data-race pairs $(\Event_1, \Event_2)$.
To assess the significance of the new races detected by $\RaceFunction$, we have also computed the number of racy memory locations that are missed by each method.
We see that $\RaceFunction$ misses $0$ memory locations (i.e., it detects all racy memory locations).
For each of $\HB$, $\WCP$, $\DC$ and $\SHB$, this number has been computed by counting how many locations have been detected by $\RaceFunction$ and missed by the corresponding method.
We refer to \cref{tab:missed_locations} for the cases where at least one method missed some racy memory location.
In total, each of the baseline methods misses tens of racy memory locations.
Finally, we have also computed location-specific race distances, as follows.
For each location $x$, we computed the minimum distance $d_{x}$ between all races on location $x$.
Hence, $d_{x}$ holds the smallest distance of a race which reveals that the location $x$ is racy.
The mean and maximum location-specific race distances in \texttt{eclipse} are 3M and 38M events, respectively,
while in \texttt{lusearch}, they are 33M and 125M events, respectively.
These numbers indicate that windowing techniques, which are typically restricted to windows of a few hundreds/thousands of events, are likely to produce highly incomplete results, and even fail to detect that certain memory locations are racy.
Similar observations have also been made in recent works~\cite{Kini17,Roemer18}.
%As a side note, we remark that the majority of our benchmarks are taken from~\cite{Kini17}.
%In that work, it is shown that the SAT/SMT-based tool RV-predict is outperformed by $\WCP$.
%Since $\RaceFunction$ significantly outperforms $\WCP$ in the same benchmark set (and, in fact, $\RaceFunction$ detects all races in this set), a comparison between $\RaceFunction$ and RV-predict would not be informative, as it is already determined in favor of $\RaceFunction$.

\smallskip\noindent{\bf Scalability.}
We see that $\RaceFunction$ has comparable running time to the baseline methods, and is sometimes faster.
One clear exception is on \texttt{eclipse}, where $\RaceFunction$ requires about 45m, whereas the other methods spend between 7m and 14m.
However, this is the benchmark on which all other methods miss both the most races (at least 13, see \cref{tab:experimental}) and the most racy memory locations (at least 8, see \cref{tab:missed_locations}).
Hence, the completeness of our race reports comes at a relatively small increase in running time.

To better understand the efficiency of $\RaceFunction$, recall that its worst-case complexity is a product of two factors, $O(\alpha\cdot  \beta)$,
where $\alpha$ is the number of calls to $\RaceDecision$ for verifying race pairs, and $\beta$ is the time taken by $\Closure$ to compute the closure of the underlying partial order $P$. In the worst case, both $\alpha$ and $\beta$ are $\Theta(n^2)$ (ignoring log-factors).
In practice, we have observed that $\RaceFunction$ resorts on calling $\RaceDecision$ only a small number of times, hence $\alpha$ is small.
This illustrates the practical advantage of $\RaceFunction$ over the naive approach that just uses $\RaceDecision$ on all $\binom{n}{2}$ event pairs.
In addition, we can express $\beta$ as roughly $\beta=n+m \cdot \gamma$, where $m$ is the number of edges inserted in $P$ during closure, and $\gamma$ is the time spent for each such edge.
Using our data structure $\DS$ for representing $P$, we have $\gamma=O(\log n)$, and, although $m=\Theta(n^2)$ in the worst-case, we have observed that $m$ behaves as a constant in practice.

Finally, we note that the baseline methods may admit further engineering optimizations that reduce their running time.
Such optimizations exist for $\HB$, but we are unaware of any attempts to optimize $\WCP$, $\DC$ or $\SHB$ further.
In any case, although our tool is not faster,  the take-home message is well-supported:
$\RaceFunction$ makes sound and effectively complete race predictions, at running times comparable to the theoretically fastest, yet highly incomplete, state-of-the-art methods.

\section{Related Work}\label{SEC:RELATED_WORK}

In this section we briefly review related work on dynamic race detection.

\smallskip\noindent{\bf Predictive analyses.}
Predictive techniques aim at inferring program behavior simply by looking at given traces.
In the context of race detection, the $\CP$ partial order~\cite{Smaragdakis12} and $\WCP$ partial order~\cite{Kini17} are sound but incomplete predictive techniques based on partial orders.
A somewhat different approach was proposed recently in \cite{Roemer18}, based on the $\DC$ partial order.
$\DC$ imposes fewer orderings than $\WCP$, but is generally unsound. 
To create sound warnings, a $\DC$-race is followed by a vindication phase, which is sound but incomplete.
%It is not hard to see that if a race is soundly detected by $\WCP$ or $\DC$, then it is also detected by our algorithm $\RaceDecision$.
%Intuitively, this is true because the partial order underpinning $\RaceDecision$ imposes fewer orderings than either of the above methods.

Other works in this domain include~\cite{Said11,Huang14,Liu16,Wang09}, which typically approach the problem based on SAT/SMT encodings.
These works are sound and complete in theory, but require exponential time.
In practice, techniques such as windowing make these methods operate fast, at the cost of sacrificing completeness.
%Predictive techniques have also been used for race detection in cases where the input trace is missing events~\cite{Huang16}.
Predictive techniques have also been used for atomicity violations and synchronization errors~\cite{Sorrentino10,Chen08,Koushik05,Huang15},
as well as in lock-based communication~\cite{Kahlon05,Farzan09,Sorrentino10}.

\smallskip\noindent{\bf Happens-before techniques.}
A large pool of race detectors are based on Lamport's happens-before relation~\cite{Lamport78}, which yields the $\HB$ partial order.
$\HB$ can be computed in linear time~\cite{Mattern89} and has been the technical basis behind many approaches~\cite{Schonberg89,Christiaens01,Pozniansky03,Flanagan09,Bond10}.
The tradeoff between runtime and space usage in race-detection using the happens-before relation was studied in~\cite{Banerjee06}.
Recently, the $\SHB$ partial order was proposed as an extension to $\HB$ in order to effectively detect multiple races per trace~\cite{Mathur18}.
%All these techniques are typically fast and sound, but have several false negatives.

\smallskip\noindent{\bf Lockset-based techniques.}
A lockset of a variable is the set of locks that guard critical regions in which the variable is accessed.
Lockset-based techniques report races by comparing the locksets of the variables accessed by the corresponding events.
They were introduced in~\cite{Dinning91} and equipped by the tool of~\cite{Savage97}.
Lockset-based techniques tend to produce many false positives and this problem has been targeted by various enhancements such as
random testing~\cite{Sen08} and
static analysis~\cite{vonPraun01,Choi02}.

\smallskip\noindent{\bf Other approaches.}
To reduce unsound reports, lockset-based techniques have been combined with happens-before techniques~\cite{Elmas07,OCallahan03,Yu05}.
Other approaches include statistical techniques~\cite{Marino09,Bond10} and static race-detectors~\cite{Naik06,Voung07,Pratikakis11}.
Recently,~\cite{Genc19} applied a combination of static and dynamic techniques to allow for correct reorderings in which the observation of some read events is allowed to differ between the input and witness trace, as long as the read does not affect the control-flow of the respective thread.
Such static information can be directly incorporated in the techniques we have developed in this paper, and is left for interesting follow-up work.
Dynamic race detection has also been studied under structured parallelism~\cite{Raman12}
and relaxed memory models~\cite{Kim09,Lidbury17}.

\section{Conclusion}\label{sec:conclusion}
We have presented $\RaceFunction$: a new polynomial-time algorithm for the problem that has no false positives.
In addition, our algorithm is \emph{complete} for input traces that consist of two processes,
i.e., it provably detects \emph{all} races in the trace.
We have also developed criteria for detecting completeness dynamically, even in the case of more than two processes.
Our experimental validation found that $\RaceFunction$ is very effective in practice, as it soundly reported all races in the input benchmark set.
Although $\RaceFunction$ is not theoretically complete in the general case, we believe that its completeness guarantee on two processes provides some explanation of why it performs so well in practice.

\begin{acks}
I am grateful to Umang Mathur for his valuable assistance in the experimental part of the paper,
to Viktor Kun\v cak for his insightful comments in earlier drafts, 
and to anonymous reviewers for their constructive feedback.
This work is partly supported by the \grantsponsor{FWF}{Austrian Science Fund (FWF) Schr\"odinger}{} grant \grantnum{}{J-4220}.
\end{acks}

%%% Acknowledgments
%\begin{acks}                            %% acks environment is optional
%                                        %% contents suppressed with 'anonymous'
%  %% Commands \grantsponsor{<sponsorID>}{<name>}{<url>} and
%  %% \grantnum[<url>]{<sponsorID>}{<number>} should be used to
%  %% acknowledge financial support and will be used by metadata
%  %% extraction tools.
%  This material is based upon work supported by the
%  \grantsponsor{GS100000001}{National Science
%    Foundation}{http://dx.doi.org/10.13039/100000001} under Grant
%  No.~\grantnum{GS100000001}{nnnnnnn} and Grant
%  No.~\grantnum{GS100000001}{mmmmmmm}.  Any opinions, findings, and
%  conclusions or recommendations expressed in this material are those
%  of the author and do not necessarily reflect the views of the
%  National Science Foundation.
%\end{acks}

\appendix
\section{Incremental DAG Reachability}\label{sec:dag_reach}

In this section we target the problem of solving incremental reachability on Directed Acyclic Graphs (DAGs).
Informally, we are given a DAG and an online sequence of 
(i)~edge-insertion and
(ii)~reachability query operations.
The task is to answer each reachability query correctly, accounting for all preceding edge-insertion operations.
Here we develop a data structure $\DS$ for solving the problem efficiently on DAGs of small width.
In the main paper we use $\DS$ to compute the closure of partial orders efficiently.
We expect that $\DS$ might be of relevance also to other race-detection techniques that are graph-based.

\smallskip\noindent{\bf Directed acyclic graphs of small width.}
Let $G=(V,E)$ be a DAG and $E^*$ be the transitive closure of $E$.
Note that $E^*$ is a partial order, and we let $\Width(G)=\Width(E^*)$.
Our focus is on DAGs of small width, i.e., we take $\Width(G)=k=O(1)$.
For $u,v\in V$, we write $u\Path v$ if $v$ is reachable from $u$.
We represent $G$ as $k$ (totally ordered) chains with extra edges between them.
We let $V\subseteq [k]\times [n]$, so that a node of $G$ is represented as a pair $(i,j)$, meaning that it is the $j$-th node in the $i$-th chain.
For two nodes $\CNode{i, j_1}, \CNode{i, j_2}\in V$ with $j_2=j_1+1$, we have $((i,j_1), (i,j_2))\in E$.
Given two nodes $\CNode{i, j_1}, \CNode{i,j_2}$, we say that $\CNode{i, j_1}$ is \emph{higher than} $\CNode{i, j_2}$ if $j_1\leq j_2$.
In such a case, we say that $\CNode{i,j_2}$ is \emph{lower than} $\CNode{i, j_1}$.
The edge set is represented as a set of arrays $\Outgoing_{i_1}^{i_2}\to [n]\cup \{\infty\}$, where $i_1, i_2\in [k]$ and $i_1\neq i_2$.
We have that $\Outgoing_{i_1}^{i_2}[j_1]=j_2\in [k]$ iff $(\CNode{i_1, j_1}, \CNode{i_2, j_2})\in E$.
Note that since $k=O(1)$, such a representation requires $O(n)$ space even if $G$ is a dense graph.

\smallskip\noindent{\bf Incremental reachability on DAGs of small width.}
The \emph{incremental reachability} problem on a DAG $G=(V,E)$ is defined on an online sequence of operations of the following types.
\begin{compactenum}
\item An $\Insert(u,v)$ operation, such that $v\not \Path u$, inserts the edge $u,v$ in $G$.
\item A $\Query(u,v)$ operation returns $\True$ iff $u\Path v$.
\item A $\Successor(u,i)$ operation returns the highest successor of $u$ in the $i$-th chain.
\item A $\Predecessor(u,i)$ operation returns the lowest predecessor of $u$ in the $i$-th chain.
\end{compactenum}
The task is to answer $\Query$ operations correctly, taking into consideration all preceding $\Insert$ operations.
Note that the width of $G$ does not increase after any operation.
We will present a data structure that handles each such query in $O(\log n)$ time.
Our data structure is based on the dynamic suffix minima problem, presented below.

\smallskip\noindent{\bf Dynamic suffix minima and Fenwick-trees.}
The \emph{dynamic suffix minima} problem is defined given an integer array $A$ of length $n$, and an online sequence of operations of the following types.
\begin{compactenum}
\item An $\Update(i,x)$ operation, for $1\leq i \leq n$ and $x\in \Ints$, sets $A[i]=x$.
\item A $\Min(i)$ operation, for $1\leq i\leq n$, returns $\min_{i\leq j\leq n}A[j]$.
\item An $\ArgMin(i)$ operation, for $1\leq i\leq n$, returns $\max_{j\colon A[j]\leq i} j$.
\end{compactenum}
The task is to answer $\Min$ and $\ArgMin$ operations correctly, taking into consideration all preceding $\Update$ operations.
The Fenwick-tree data structure solves the dynamic suffix minima problem in $O(\log n)$ time per operation, after $O(n)$ preprocessing time~\cite{Fenwick94}.

\smallskip\noindent{\bf The data structure $\DS$ for solving the incremental reachability problem.}
We are now ready to describe our data structure $\DS$ for solving the incremental reachability problem
given a DAG $G=(V,E)$ of width $k$, and an online sequence $\Sigma$ of $\Insert$ and $\Query$ operations.
We consider that $G$ is given in a sparse representation form, where outgoing edges
are represented using the arrays $\Outgoing_{i_1}^{i_2}$, for each $i_1, i_2\in [k]$ and $i_1\neq i_2$.

In the initialization phase, $\DS$ performs the following steps. 
For each $i_1,i_2\in[k]$ such that $i_1\neq i_2$,
we initialize a Fenwick-tree data structure $\Fenwick_{i_1}^{i_2}$ with array $\Outgoing_{i_1}^{i_2}$.
This data structure stores forward reachability information from nodes of the $i_1$-th chain to nodes in the $i_2$ chain,
by maintaining the invariant that $\Fenwick_{i_1}^{i_2}.\Min(j_1)=j_2$ iff $\CNode{i_2, j_2}$ is the highest node of the $i_2$-th chain reachable from $\CNode{i_1, j_2}$.
A $\Successor(\CNode{i_1, j_1},i)$ (resp., $\Predecessor(\CNode{i_1, j_1},i)$) operation is handled by $\DS$ by returning $\Fenwick_{i_1}^{i}.\Min(j_1)$ (resp., $\Fenwick_{i_1}^{i}.\ArgMin(j_1)$).
Finally, the operations $\DS.\Insert$ and $\DS.\Query$ are handled by \cref{algo:insert} and \cref{algo:query}, respectively.
The following lemma establishes the correctness and complexity of the data structure $\DS$.

%\smallskip\noindent{\em Operations $\Query$ and $\Insert$.}
%A $\Query(\CNode{i_1, j_1},\CNode{i_2, j_2})$ operation is handled by the algorithm $\DS.\Query$ (\cref{algo:query}).
%In words, $\DS.\Query$ answers that $\CNode{i_2, j_2}$ is reachable from $\CNode{i_1, j_1}$ iff $\DS.\Successor(\CNode{i_1,j_1}, i_2)\leq  j_2$, 
%i.e., $\CNode{i_1, j_1}$ is higher than a node $u$ in the $i_1$-th chain such that $u$ points to a node $v$ such that $v$ is higher than $\CNode{i_2, j_2}$ in the $i_2$-th chain.
%An $\Insert(\CNode{i_1, j_1},\CNode{i_2, j_2})$ operation is handled by the algorithm $\DS.\Insert$ (\cref{algo:insert}).
%In words, $\DS.\Insert$ updates the data structure $\Fenwick_{i_1}^{i_2}$ with the new edge.
%Additionally, it performs two extra propagation steps.

\smallskip
\begin{restatable}{lemma}{datastructure}\label{lem:data_structure}
Let $\Sigma$ be an online sequence of incremental reachability operations.
The data structure $\DS$ correctly handles $\Sigma$ and spends
(i)~$O(n)$ preprocessing time and
(ii)~$O(\log n)$ time per operation.
\end{restatable}

\begin{minipage}[t]{.47\textwidth}
\vspace{-0.3cm}
\begin{algorithm}[H]
\small
\SetInd{0.4em}{0.4em}
\DontPrintSemicolon
%\setstretch{1.05}
\caption{$\DS.\Insert(\CNode{i_1, j_1},\CNode{i_2, j_2})$}\label{algo:insert}
\BlankLine
\ForEach{$i'_1\in [k]$}{
\ForEach{$i'_2\in [k]$}{
Let $j'_1\gets \Predecessor(\CNode{i_1, j_1}, i'_1)$\\
Let $j'_2\gets \Successor(\CNode{i_2, j_2}, i'_2)$\\
$\Fenwick_{i'_1}^{i'_2}.\Update(j'_1, j'_2)$\\
%$\BwdFenwick_{i'_2}^{i'_1}.\Update(j'_2, j'_1)$\\
}
}
\end{algorithm}
\end{minipage}
\qquad
\begin{minipage}[t]{.47\textwidth}
\vspace{-0.3cm}
\begin{algorithm}[H]
\small
\SetInd{0.4em}{0.4em}
\DontPrintSemicolon
%\setstretch{1.05}
\caption{$\DS.\Query(\CNode{i_1, j_1},\CNode{i_2, j_2})$}\label{algo:query}
\BlankLine
Let $j'_2\gets \DS.\Successor(\CNode{i_1,j_1}, i_2)$\label{line:successor}\\
\eIf{$j'_2\leq j_2$}{
\Return{$\True$}
}{
\Return{$\False$}
}
\end{algorithm}
\end{minipage}

\vspace{0.5cm}
\section{Definition of $\HB$, $\SHB$, $\WCP$ and $\DC$}\label{sec:hb_wcp_dc}
For the sake of completeness, here we give the definitions of the partial orders $\HB$, $\SHB$, $\WCP$ and $\DC$, based on~\cite{Kini17,Mathur18,Roemer18}.
In each case, we consider given a trace $\Trace$, and the respective partial order is over the set of events $\Events{\Trace}$ of $\Trace$.

\smallskip\noindent{\bf The $\HB$ partial order }
is the smallest partial order that satisfies the following conditions~.
\begin{compactenum}
\item\label{item:hb1} $\HB\Refines \TO$.
\item\label{item:hb2} For every lock-acquire and lock-release events $\Acquire$ and $\Release$ such that $\Release<_{\Trace}\Acquire$,
if $\Confl{\Acquire}{\Release}$ then $\Release<_{\HB} \Acquire$.
\end{compactenum}

\smallskip\noindent{\bf The $\SHB$ partial order }
is the smallest partial order that satisfies the following conditions.
\begin{compactenum}
\item\label{item:shb1} $\SHB\Refines \HB$.
\item\label{item:shb2} For every read event $\Read$ we have $\Observation_{\Trace}(\Read)<_{\SHB}\Read$.
Recall that $\Observation_{\Trace}(\Read)$ is the observation of $\Read$ in $\Trace$.
\end{compactenum}

\smallskip\noindent{\bf The $\WCP$ partial order }
is the smallest partial order that satisfies the following conditions.
\begin{compactenum}
\item\label{item:wcp1} $\WCP\Refines \TO$.
\item\label{item:wcp2} For every lock-release event $\Release$ and write/read event $\Event$ such that $\Release<_{\Trace} \Event$, if
(i)~$\Event$ is protected by a lock $\ell=\Location{\Release}$ and 
(ii)~$\Event$ conflicts with an event in the critical section of $\Release$,
then $\Release<_{\WCP}\Event$.
\item\label{item:wcp3} For every two lock-release events $\Release_1, \Release_2$ such that $\Release_1<_{\Trace}\Release_2$,
if the critical sections of $\Release_1$ and $\Release_2$ contain $\WCP$-ordered events, then $\Release_1<_{\WCP}\Release_2$.
\item $\WCP$ is closed under left and right composition with $\HB$.
\end{compactenum}

\smallskip\noindent{\bf The $\DC$ partial order }
is the smallest partial order that satisfies  conditions~\ref{item:wcp1}, \ref{item:wcp2} and \ref{item:wcp3} of $\WCP$.

\clearpage

%% Bibliography
\bibliography{bibliography}

%% Appendix
\clearpage
%\appendix

\section{Missing Proofs}\label{sec:missing_proofs}

\subsection{Proofs of Section~{\ref{sec:po}}}\label{SEC:MISSING_PROOFS_MINMAX}

\smallskip
\maxmin*
\begin{proof}
We first argue that the partial order $Q$ defined in \cref{line:Q} is indeed a partial order, and thus the linearization $\Trace^*$ is well-defined.
Assume towards contradiction otherwise, hence there exist two events $\Event_1, \Event_2\in X$ such that $\Event_i\in X_i$ for each $i\in [2]$ and the algorithm inserts an edge $\Event_1 \to \Event_2$ in $Q$.
Since all edges inserted in $Q$ go from $X_1$ to $X_2$, there exists an event $\Event'_1\in X_2$ such that $\Event_1<_{P} \Event'_1$ and $\Event'_1<_{P} \Event_2$.
But then $\Ordered{\Event_1}{P}{\Event_2}$, and the algorithm could not have inserted the edge $\Event_1\to \Event_2$ in $Q$, a contradiction.
Hence $Q$ is a partial order.

Note that $Q\Refines P$ and thus $\Trace^*$ is a linearization of $P$.
We show that 
(i)~the observation function of $\Trace^*$ agrees with the observation function of $\Trace$ and
(ii)~$\Trace^*$ respects the lock semantics.

\smallskip\noindent{\em Observations.}
Consider any read event $\Read\in \Reads{X}$, and let $\Write=\Observation_{\Trace}(\Read)$.
Since $P$ respects $\Trace$, we have that $\Write\in X$ and $\Write<_{P} \Read$, and since $\Trace^*$ is a linearization of $P$,
we have $\Write<_{\Trace^*}\Read$.
Let $\Write'\in \Writes{X}$ be any write event such that $\Confl{\Write}{\Read}$ and $\Write'\neq \Write$,
and we will argue that if $\Write'<_{\Trace^*}\Read$ then $\Write'<_{\Trace^*} \Write$.
We distinguish the following cases.
\begin{compactenum}
\item If $\Ordered{\Write'}{P}{\Read}$ or $\Ordered{\Write'}{P}{\Write}$, we have $\Write'<_{P} \Read$ and since $P$ is observation closed we have $\Write'<_{P} \Write$.
\item Otherwise, since $(P\Project X_i)$ is an M-trace for each $i\in [2]$, it follows that $\Read, \Write\in X_i$ and $\Write'\in X_{3-i}$, for some $i\in [2]$.
Since $\Write'<_{\Trace^*} \Read$, we have that $\Write'\in X_1$, and \cref{line:e2toe1} guarantees that $\Write'<_{Q} \Write$, and thus $\Write'<_{\Trace^*} \Write$.
\end{compactenum}

\smallskip\noindent{\em Locks.}
Consider two lock acquire events $\Acquire_1, \Acquire_2\in \Acquires{X}$ with  $\Location{\Acquire_1}=\Location{\Acquire_2}=\ell$.
Let $\Release_i=\Match[\Trace]{\Acquire_i}$ for each $i\in \{1,2\}$.
Assume wlog that $\Acquire_1<_{\Trace^{*}} \Acquire_2$, and observe that $\Release_1\in X$.
Indeed, since $P$ respects $\Trace$, if $\Release_1\not \in X$, we would have $\Release_2\in X$ and $\Release_2<_{P} \Acquire_1$,
and since $P$ is lock closed, we would also have $\Acquire_2<_{P} \Acquire_1$, a contradiction.
We will argue that $\Release_1<_{\Trace^{*}}\Acquire_2$.
We consider the following cases.
\begin{compactenum}
\item If $\Release_2\not \in X$ or $\Release_2\in X$ and $\Acquire_1 <_{P} \Release_2$, since $P$ is lock-closed,
we have $\Release_1<_{P} \Acquire_2$, and since $\Trace^{*}$ is a linearization of $P$,
we conclude that $\Release_1<_{\Trace^{*}}\Acquire_2$.
\item Otherwise, if $\Acquire_2 <_{P} \Release_1$, since $P$ is lock closed, we have $\Release_2 <_{P} \Acquire_1$ and thus $\Acquire_2 <_{P} \Acquire_1$, a contradiction.
\item Finally, we have $\Unordered{\Acquire_i}{P}{\Release_{3-i}}$ for each $i\in \{1,2\}$.
Since for each $i\in[2]$ we have that each $Q\Project X_i$ is an M-trace, we have that $\Acquire_i$ and $\Release_{3-i}$ do not belong to the same set $X_j$.
Since $\Acquire_1<_{\Trace^*} \Acquire_2$, we have that $\Acquire_1\in X_1$ and thus $\Release_1\in X_1$ and $\Acquire_2\in X_2$.
Hence, \cref{line:e2toe1} ensures that $\Release_1<_{Q} \Acquire_2$, as desired.
\end{compactenum}

%The desired result follows.
\end{proof}

\smallskip
\closureunique*
\begin{proof}
Assume towards contradiction otherwise, and consider two partial orders $Q_1, Q_2$  with properties (i) and (ii),
and such that $Q_i\not \Refines Q_{3-i}$ for each $i\in[2]$.
Let $Q=Q_1\cap Q_2$, and hence $Q_1,Q_2 \Refines Q$. We argue that $Q$ is closed.
First, observe that each $Q_i$ respects $\Trace$, and hence $Q$ respects $\Trace$.

We now argue that $Q$ is observation-closed.
For every read event $\Read\in \Reads{X}$ and write event $\Write\in \Writes{X}$ such that
$\Confl{\Write}{\Read}$ and $\Write\neq \Observation_{\Trace}(\Read)$,
if $\Write<_{Q} \Read$ then $\Write<_{Q_i} \Read$ for each $i\in[2]$.
Since each $Q_i$ is closed, we have $\Write <_{Q_i} \Observation_{\Trace}(\Read)$ for each $i$, and thus
$\Write <_{Q} \Observation_{\Trace}(\Read)$.
Similarly if $\Observation_{\Trace}(\Read) <_{Q} \Write$, we conclude that $\Read <_{Q} \Write$.

Finally, we argue that $Q$ is lock-closed.
Consider any pair of lock-release events $\Release_1, \Release_2\in \Releases{X}$,
and let $\Acquire_i=\Match[\Trace]{\Release_i}$ for each $i\in[2]$.
If $\Acquire_1<_{Q} \Release_2$,
then $\Acquire_1<_{Q_i} \Release_2$ for each $i\in[2]$.
Hence $\Release_1<_{Q_i} \Acquire_2$ for each $i\in[2]$, and thus
$\Release_1<_{Q} \Acquire_2$.
%The desired result follows.
\end{proof}

\smallskip
\smallskip
\begin{restatable}{lemma}{closurecorrectness}\label{lem:closure_correctness}
The algorithm $\Closure$ (\cref{algo:closure}) correctly computes the closure of $P$.
\end{restatable}
\begin{proof}
We argue that the partial order $Q$ stored in the data structure $\DS$ returned in \cref{line:return} represents the closure of $P$.
Because of \cref{line:cycle_formed}, it is easy to see that $\DS$ represents indeed a partial order $Q$, and due to \cref{line:po_edges}, 
we have that $Q\Refines P$.
We will argue that $Q$ is closed.

We start by showing that $Q$ is observation-closed.
Consider any read event $\Read\in \Reads{X}$, and let $\Location{\Read}=x$ and $\Write=\Observation_{\Trace}(\Read)$.
Consider any write event $\Write'\in \Writes{X}$ such that $\Confl{\Write'}{\Read}$ and $\Write'\neq \Write$.
\begin{compactenum}
\item Assume that $\Write' <_{Q} \Read$, and we will show that $\Write'<_{Q} \Write$.
We prove the claim for $\Write'$ being any \emph{last} such event, i.e.,
for every other $\Write''$ with $\Confl{\Write''}{\Read}$, $\Write''\neq \Observation_{\Trace}(\Read)$ and $\Write''<_{Q}\Read$, we have that
$\Write'\not< _{Q}\Write''$.
%$\Unordered{\Write'}{Q}{\Write''}$.
Clearly, this establishes the claim for all $\Write'$.
Note that there exist two events $\Event_1,\Event_2$ such that 
\begin{compactenum}
\item $\Write'<_{\TO}\Event_1$ and $\Event_2<_{\TO} \Read$ (possibly $\Event_1=\Write'$ and $\Event_2=\Read$), and
\item either $\Event_1<_{P}\Event_2$ or the algorithm performs a $\DS.\Insert(\Event_1, \Event_2)$ in \cref{line:ds_insert_edge}.
\end{compactenum}
By the choice of $\Write'$, we have that (i)~$\Write'=\To_x^{\SysWrites}(\Event_1)$ and (ii)~$\Observation_{\Trace}(\Read)=\Observation_{\Trace}(\From_x^{\SysReads}(\Event_2))=\Write$.
To see(i), note that if $\Write\neq \To_x^{\SysWrites}(\Event_1)$, this violates our choice of $\Write'$ being a last conflicting write.
To see (ii), let $\Observation_{\Trace}(\To_x^{\SysReads}(\Event_2))=\Write''$ and observe that $\Write''<_{Q} \Read$.
Because of \cref{line:obs_before} in $\ObsClosure$, we have $\Write'<_{Q} \Write''$, and thus if $\Write\neq \Write''$, this violates our choice of $\Write'$ being a last conflicting write.
After \cref{line:obs_closure_pre} of $\ObsClosure$ is executed, we have $\Write<_{Q}\Write$, as desired.

\item Assume that $\Write <_{Q} \Write'$, and we will show that $\Read<_{Q} \Write'$.
We prove the claim for $\Write'$ being any \emph{first} such event, i.e.,
for every other $\Write''$ with $\Confl{\Write''}{\Read}$, $\Write''\neq \Observation_{\Trace}(\Read)$ and $\Write<_{Q}\Write''$, we have that $\Write''\not <_{Q} \Write'$.
 %$\Unordered{\Write'}{Q}{\Write''}$.
Clearly, this establishes the claim for all $\Write'$.
Note that there exist two events $\Event_1,\Event_2$ such that 
\begin{compactenum}
\item $\Write<_{\TO}\Event_1$ and $\Event_2<_{\TO} \Write'$ (possibly $\Event_1=\Write$ and $\Event_2=\Write'$), and
\item either $\Event_1<_{P}\Event_2$ or the algorithm performs a $\DS.\Insert(\Event_1, \Event_2)$ in \cref{line:ds_insert_edge}.
\end{compactenum}
By the choice of $\Write'$, we have that (i)~$\Write=\To_{x}^{\SysWrites}(\Event_1)$ and (ii)~$\Write'=\From_{x}^{\SysWrites}(\Event_2)$.
To see (i), note that if $\Write\neq \To_{x}^{\SysWrites}(\Event_1)$, this violates our choice of $\Write'$ being a first conflicting write.
Similarly, to see (ii), note that if $\Write'\neq \From_{x}^{\SysWrites}(\Event_2)$, this also violates our choice of $\Write'$ being a first conflicting write.
After \cref{line:obs_before} of $\ObsClosure$ is executed, we have $\Read' <_{Q} \Write'$, where $\Read'=\LastFlow_{\Trace}^i(\Write)$ for $i$ such that $\Process_i=\Proc{\Read}$. By construction, we have $\Read<_{\TO}\Read'$, and since $Q\Refines \TO\Project X$, we have $\Read<_{Q}\Write'$, as desired.
\end{compactenum}

We now show that $Q$ is lock-closed.
Consider any pair of lock-release events $\Release_1, \Release_2\in \Releases{X}$, let $\Acquire_i=\Match[\Trace]{\Release_i}$, for $i\in [2]$,
and assume that $\Confl{\Release_2}{\Acquire_1}$ and $\Acquire_1<_{Q} \Release_2$
We will show that $\Release_1<_{Q}\Acquire_2$.
Observe that in this case, there exist two events $\Event_1, \Event_2$ with $\Acquire_1<_{\TO} \Event_1$ and $\Event_2<_{\TO} \Release_2$,
and either $\Event_1<_{P} \Event_2$ or the algorithm performs a $\DS.\Insert(\Event_1, \Event_2)$ in \cref{line:ds_insert_edge}.
In either case, the algorithm calls $\LockClosure(\Event_1, \Event_2)$ (in \cref{line:lock_closure_pre} for the former case, and in \cref{line:lock_closure_comp} for the latter).
The well-nestedness of locks in $\Trace$ guarantees that $\Release_1=\From_{l}^{\SysReleases}(\Event_1)$ and $\Release_2=\From_{l}^{\SysReleases}(\Event_2)$.
If $\Event_2<_{\TO} \Acquire_2$, then by transitivity we have $\Release_1<_{Q}\Acquire_2$, and we are done.
Otherwise, \cref{line:lock} of $\LockClosure$ will be executed, and thus $\Release_1<_{Q}\Acquire_2$.

Hence, we have shown that the partial order $Q$ represented by the data structure $\DS$ at the end of $\Closure$ is closed.
%The desired result follows.
\end{proof}

We now turn our attention to complexity.

\smallskip
\smallskip
\begin{restatable}{lemma}{closurecomplexity}\label{lem:closure_complexity}
$\Closure$ (\cref{algo:closure}) requires $O(n^2\cdot \log n)$ time.
\end{restatable}
\begin{proof}
First, by \cref{lem:data_structure}, the initialization of $\DS$ requires $O(n)$ time.
Observe that every invocation to $\ObsClosure$ and $\LockClosure$ requires $O(k\cdot |\Globals|)=O(1)$ time.
Hence the initialization of $\Closure$ in \cref{line:po_edges} requires $O(n\cdot \log n)$ time.
We now turn our attention to the main computation in \cref{line:while}, 
and consider an edge $(\ov{\Event}_1, \ov{\Event}_2)$ extracted in \cref{line:q_pop}.
By \cref{lem:data_structure}, every $\DS.\Query$ requires $O(\log n)$ time,
and the loop in \cref{line:inner_loop} will iterate over $O(1)$ edges $(\Event_1, \Event_2)$.
Since every invocation to $\ObsClosure$ and $\LockClosure$ requires $O(1)$ time, we conclude that the cost of every edge $(\ov{\Event}_1, \ov{\Event}_2)$ is $O(\log n)$.
Finally, observe that for every edge $z=(\ov{\Event}_1, \ov{\Event}_2)$ every held in $\Queue$ there exists an edge $z'=(\ov{\Event}'_1, \ov{\Event}'_2)$ which was inserted for the first time in $\DS$. Since there are are $O(n^2)$ such edges $z'$, we have that $\Queue$ will hold $O(n^2)$ elements in total.
Hence the total running time of $\Closure$ is $O(n^2\cdot \log n)$.
%The desired result follows.
\end{proof}

\smallskip
\closuredynamiccomplexity*
\begin{proof}[Proof (Sketch)]
Similarly to \cref{lem:closure_complexity}, the time required for handling $\Sigma$ is proportional to the size of $\Sigma$ times $O(\log n)$ for querying whether each edge of $\Sigma$ is already present in $\DS$,
plus $O(\log)$ for every new edge inserted in $\DS$.
Since there can be $O(n^2)$ new edges inserted, the time bound is $O(n^2\cdot \log n + |\Sigma|\cdot\log n)$ for the whole sequence $\Sigma$.
%The desired result follows.
\end{proof}

\subsection{Proofs of Section~{\ref{sec:decision}}}\label{SEC:MISSING_PROOFS_DECISION}

\smallskip
\begin{restatable}{lemma}{racedecisionsoundness}\label{lem:racedecision_soundness}
If $\RaceDecision$ returns $\True$ then $(\Event_1, \Event_2)$ is a predictable race of $\Trace$.
\end{restatable}
\begin{proof}
Observe that if $\RaceDecision$ returns $\True$ then at that point $X$ is a feasible set and by \cref{them:closure}, $Q$ is a closed partial order.
Additionally, due to the loop in \cref{line:mwidth1}, $X$ can be naturally partitioned into two sets $X_1$ and $X_2$ such that $\Width(Q\Project X_1)=1$ and $Q\Project X_2$ is an M-trace.
In particular, we have $X_1$ be the set of events of $\Process_i$ and $X_2=X\setminus X$.
By \cref{them:maxmin}, the sequence $\Trace^*$ returned by $\MinMaxAlgo$ on $Q$ is a correct reordering of $\Trace$.
Finally, by the definition of relative causal cones,
the events $\Event_1, \Event_2$ are enabled in their respective processes when $\Trace^*$ is executed.
%The desired result follows.
\end{proof}

\smallskip
\begin{restatable}{lemma}{racedecisioncompleteness}\label{lem:racedecision_completeness}
Let $\Trace$ be a trace of a program with $k=2$ processes, and let $(\Event_1, \Event_2)$ be a predictable race of $\Trace$.
Then $\RaceDecision$ returns $\True$ for the pair $(\Event_1, \Event_2)$.
\end{restatable}
\begin{proof}
First, note that since $(\Event_1, \Event_2)$ is a predictable race of $\Trace$, there exists a correct reordering $\Trace^*$ of $\Trace$
such that after $\Trace^*$ is executed, $\Event_1$ and $\Event_2$ are the enabled events in their respective processes.
It is easy to see that $\Events{\Trace^*}=\RPast_{\Trace}(\Event_1, \Proc{\Event_2})\cup \RPast_{\Trace}(\Event_2, \Proc{\Event_1})$, and thus
$\OpenAcquires_{\Trace}(\Events{\Trace^*})=\OpenAcquires_{\Trace}(X)$, for the set $X$ constructed in \cref{line:construct_x} of $\RaceDecision$.
Hence $X$ is feasible, and $\{\Event_1, \Event_2\}\cap X=\emptyset$.
Viewed as a partial order, $\Trace^*$ must respect $\Trace$, and as $\Events{\Trace^*}=X$, we have that $\Trace^*$ is a linearization of $P$ which is constructed in \cref{line:respect_po}.
Additionally, $\Trace^*$ must be closed, hence $P$ is feasible and $Q$ is a valid partial order in \cref{line:close}.
Since $k=2$, for every pair of events $\ov{\Event}_1, \ov{\Event}_2$ in the loop in \cref{line:mwidth1} of $\RaceDecision$ we have $\Proc{\Event_1}=\Proc{\Event_2}$ and thus $\Ordered{\Event_1}{Q}{\Event_2}$ and the loop inserts no new edges in $Q$.
Thus the algorithm returns $\True$.
%The desired result follows.
\end{proof}

\decision*
\begin{proof}
\cref{lem:racedecision_soundness} shows that $\RaceDecision$ is sound and \cref{lem:racedecision_completeness} shows that $\RaceDecision$ is complete for $k=2$.
Now we turn our attention to complexity.
It is easy to see that computing $\RPast_{\Trace}(\Event_i, \Proc{\Event_{-i}})$ requires $O(n)$ time, and by \cref{lem:data_structure}, constructing $P$ in \cref{line:respect_po} using our data structure $\DS$ requires $O(n\cdot \log n)$ time.
By \cref{them:closure}, computing the closure $Q$ of $P$ in \cref{line:close} requires $O(n^2\cdot \log n)$ time.
The loop in \cref{line:mwidth1} can also be executed in $O(n^2\cdot \log n)$ time, since by \cref{lem:dynamic_closure} all $\InsertAndClose$ operations are handled in $O(n^2\cdot \log n)$ time in total.
If $\RaceDecision$ returns $\True$, then $\MinMaxAlgo$ produces a witness trace $\Trace^*$ that linearizes $Q$ in $O(n\cdot \log n)$ time.
%The desired result follows.
\end{proof}

\subsection{Proofs of Section~{\ref{sec:function}}}\label{sec:missing_proofs_function}

\functioncorrectness*
\begin{proof}
Indeed, observe that $\Trace\Project X$ is a trace where there are no open lock-acquire events and $\Event_1, \Event_2$ are enabled in their respective processes.
Hence, $\Trace^*=\Trace\Project \left(X\cup \{\Event_1, \Event_2\}\right)$ is a correct reordering of $\Trace$ that exhibits the race $(\Event_1, \Event_2)$.
\end{proof}

\subsection{Proofs of Section~{\ref{sec:dag_reach}}}\label{sec:missing_proofs_dag_reach}

\datastructure*
\begin{proof}
We treat the correctness and complexity separately.

\noindent{\em Correctness.}
It is straightforward to establish that the data structure maintains the following invariant.
At the end of each $\Insert(\CNode{i, j}, \CNode{i', j'})$ operation of $\Sigma$, for every $i_1, i_2\in [k]$ and $j_1\in [n]$, 
\begin{compactenum}
\item\label{item:cor1} if $j_2=\Fenwick_{i_1}^{i_2}.\Query(j_1)$, then $\CNode{i_2, j_2}$ is the highest node of the $i_2$-th chain that can be reached from $\CNode{i_1, j_1}$, and
\item\label{item:cor2} if $j_2=\Fenwick_{i_1}^{i_2}.\ArgQuery(j_1)$, then $\CNode{i_2, j_2}$ is the lowest node of the $i_2$-th chain that can be reached from $\CNode{i_1, j_1}$.
\end{compactenum}

\noindent{\em Complexity.}
Initializing every Fenwick tree requires $O(n)$ time~\cite{Fenwick94}, and since we have $O(1)$ such Fenwick trees in total,
the initialization of $\DS$ requires $O(n)$ time.
A $\DS.\Query$ operation requires $O(\log n)$ time, which is determined by the $\DS.\Successor$ operation in \cref{line:successor}, 
which is implemented by a query operation in the respective Fenwick tree and thus requires $O(\log n)$ time~\cite{Fenwick94}.
We now turn our attention to the $\DS.\Insert$ operation.
Notice that this step performs $O(k^2)=O(1)$ update operations to Fenwick trees. Since each update operation requires $O(\log n)$ time~\cite{Fenwick94},
the total time spent in this operation is $O(\log n)$.
%The desired result follows.
\end{proof}
\section{Incompleteness of $\RaceDecision$ for $k\geq 3$ processes}\label{SEC:APPENDIX_INCOMPLETENESS}

\begin{figure*}[!h]
\begin{subfigure}[b]{0.25\textwidth}
\centering
\footnotesize
\def\rownumber{}
\begin{tabular}[b]{@{\makebox[1.2em][r]{\rownumber\space}} | l | l | l|}
\normalsize{$\mathbf{\SeqTrace_1}$} & \normalsize{$\mathbf{\SeqTrace_2}$} & \normalsize{$\mathbf{\SeqTrace_3}$}
\gdef\rownumber{\stepcounter{magicrownumbers}\arabic{magicrownumbers}} \\
\hline
& & $\Write(x_1)$\\
$\Acquire(\ell)$ & & \\
$\Write(x_2)$ & & \\
$\Read(x_1)$ & & \\
$\mathbf{\Write(y)}$ & & \\
$\Release(\ell)$ & & \\
& & $\Write(x_2)$ \\
& & $\Write(x_3)$ \\
& & $\Write(x_4)$ \\
& $\Acquire(\ell)$ & \\
& $\Write(x_3)$ & \\
& $\Release(\ell)$ & \\
& $\Write(x_1)$ & \\
& $\Read(x_2)$ & \\
& $\Read(x_4)$ & \\
& $\mathbf{\Write(y)}$ & \\
\hline
\end{tabular}
\caption{Is $(\Event_4, \Event_{15})$ a race?}
\label{subfig:algo_fail_trace}
\end{subfigure}
\quad
\begin{subfigure}[b]{0.4\textwidth}
\centering
\small
\begin{tikzpicture}[thick,
pre/.style={<-,shorten >= 2pt, shorten <=2pt, very thick},
post/.style={->,shorten >= 2pt, shorten <=2pt,  very thick},
seqtrace/.style={->, line width=2},
und/.style={very thick, draw=gray},
event/.style={rectangle, minimum height=3mm, draw=black, fill=white, minimum width=4.5mm,   line width=1pt, inner sep=0, font={\small}},
virt/.style={circle,draw=black!50,fill=black!20, opacity=0}]

\newcommand{\xdisposition}{0}
\newcommand{\ydisposition}{0}
\newcommand{\xstep}{1.7}
\newcommand{\ystep}{0.8}
\newcommand{\ybias}{0.4}
\newcommand{\xbias}{0.8}
\newcommand{\xbiassmall}{0.6}

\node	[]		(t1a)	at	(\xdisposition + 0*\xstep, \ydisposition + 0*\ystep)	{\normalsize$\SeqTrace_1$};
\node	[]		(t1b)	at	(\xdisposition + 0*\xstep, \ydisposition + -7*\ystep)	{};
\node	[]		(t2a)	at	(\xdisposition + 1*\xstep, \ydisposition + 0*\ystep)	{\normalsize$\SeqTrace_2$};
\node	[]		(t2b)	at	(\xdisposition + 1*\xstep, \ydisposition + -7*\ystep)	{};
\node	[]		(t3a)	at	(\xdisposition + 2*\xstep, \ydisposition + 0*\ystep)	{\normalsize$\SeqTrace_3$};
\node	[]		(t3b)	at	(\xdisposition + 2*\xstep, \ydisposition + -7*\ystep)	{};

\draw[seqtrace] (t1a) to (t1b);
\draw[seqtrace] (t2a) to (t2b);
\draw[seqtrace] (t3a) to (t3b);

\node[event] (e11) at (\xdisposition + 0*\xstep, \ydisposition + -1*\ystep) {$\Event_{2}$};
\node[] (et11) at (\xdisposition + 0*\xstep - \xbias, \ydisposition + -1*\ystep) {$\Acquire(\ell)$};
\node[event] (e12) at (\xdisposition + 0*\xstep, \ydisposition + -2*\ystep) {$\Event_{3}$};
\node[] (et12) at (\xdisposition + 0*\xstep - \xbiassmall, \ydisposition + -2*\ystep) {$\ov{\Write}(x_2)$};
\node[event] (e13) at (\xdisposition + 0*\xstep, \ydisposition + -3*\ystep) {$\Event_{4}$};
\node[] (et13) at (\xdisposition + 0*\xstep - \xbias, \ydisposition + -3*\ystep) {$\Read(x_1)$};

\node[event] (e21) at (\xdisposition + 1*\xstep, \ydisposition + -1*\ystep) {$\Event_{10}$};
\node[] (et21) at (\xdisposition + 1*\xstep + \xbias, \ydisposition + -1*\ystep) {$\Acquire(\ell)$};
\node[event] (e22) at (\xdisposition + 1*\xstep, \ydisposition + -2*\ystep) {$\Event_{11}$};
\node[] (et22) at (\xdisposition + 1*\xstep + \xbias, \ydisposition + -2*\ystep) {$\Write(x_3)$};
\node[event] (e23) at (\xdisposition + 1*\xstep, \ydisposition + -3*\ystep) {$\Event_{12}$};
\node[] (et23) at (\xdisposition + 1*\xstep + \xbias, \ydisposition + -3*\ystep) {$\Release(\ell)$};
\node[event] (e24) at (\xdisposition + 1*\xstep, \ydisposition + -4*\ystep) {$\Event_{13}$};
\node[] (et24) at (\xdisposition + 1*\xstep + \xbias, \ydisposition + -4*\ystep) {$\ov{\Write}(x_1)$};
\node[event] (e25) at (\xdisposition + 1*\xstep, \ydisposition + -5*\ystep) {$\Event_{14}$};
\node[] (et25) at (\xdisposition + 1*\xstep + \xbias, \ydisposition + -5*\ystep) {$\Read(x_2)$};
\node[event] (e26) at (\xdisposition + 1*\xstep, \ydisposition + -6*\ystep) {$\Event_{15}$};
\node[] (et26) at (\xdisposition + 1*\xstep + \xbias, \ydisposition + -6*\ystep) {$\Read(x_4)$};

\node[event] (e31) at (\xdisposition + 2*\xstep, \ydisposition + -1*\ystep) {$\Event_{1}$};
\node[] (et31) at (\xdisposition + 2*\xstep + \xbias, \ydisposition + -1*\ystep) {$\Write(x_1)$};
\node[event] (e32) at (\xdisposition + 2*\xstep, \ydisposition + -2*\ystep) {$\Event_{7}$};
\node[] (et32) at (\xdisposition + 2*\xstep + \xbias, \ydisposition + -2*\ystep) {$\Write(x_2)$};
\node[event] (e33) at (\xdisposition + 2*\xstep, \ydisposition + -3*\ystep) {$\Event_{8}$};
\node[] (et33) at (\xdisposition + 2*\xstep + \xbias, \ydisposition + -3*\ystep) {$\Write(x_3)$};
\node[event] (e34) at (\xdisposition + 2*\xstep, \ydisposition + -4*\ystep) {$\Event_{9}$};
\node[] (et34) at (\xdisposition + 2*\xstep + \xbias, \ydisposition + -4*\ystep) {$\Write(x_4)$};

\draw[post] (e23) to (e11);
\draw[post, bend left=22.5] (e34) to (e26);
\draw[post, bend left=0, dotted, draw=\darkgreen] (e33) to (e22);
\draw[post, bend left=0, dashed, draw=\darkred] (e25) to (e12);
\draw[post, bend right=10, dashed, draw=\darkred] (e24) to (e31);

\end{tikzpicture}
\caption{The partial order $P$ and its closure $Q$.}
\label{subfig:algo_fail_po}
\end{subfigure}
\quad
\begin{subfigure}[b]{0.25\textwidth}
\centering
\footnotesize
\def\rownumber{}
\begin{tabular}[b]{@{\makebox[1.2em][r]{\rownumber\space}} | l | l | l|}
\normalsize{$\mathbf{\SeqTrace_1}$} & \normalsize{$\mathbf{\SeqTrace_2}$} & \normalsize{$\mathbf{\SeqTrace_3}$}
\gdef\rownumber{\stepcounter{magicrownumbers}\arabic{magicrownumbers}} \\
\hline
& $\Acquire(\ell)$ & \\
& $\Write(x_3)$ & \\
& $\Release(\ell)$ & \\
& $\Write(x_1)$ & \\
& & $\Write(x_1)$ \\
& & $\Write(x_2)$ \\
& & $\Write(x_3)$ \\
& & $\Write(x_3)$ \\
& & $\Read(x_2)$ \\
& & $\Read(x_4)$ \\
$\Acquire(\ell)$ & & \\
$\Write(x_2)$ & & \\
$\Read(x_1)$ & & \\
$\mathbf{\Write(y)}$ & & \\
& $\mathbf{\Write(y)}$ & \\
\hline
\end{tabular}
\caption{The witness trace.}
\label{subfig:algo_fail_witness}
\end{subfigure}
\caption{
(\protect\subref{subfig:algo_fail_trace}) The input trace.
(\protect\subref{subfig:algo_fail_po}) The partial order $P$ (solid edges) and its closure $Q$ (solid and dashed edges).
(\protect\subref{subfig:algo_fail_witness}) A witness trace for the race $(\Event_4, \Event_{15})$.
}
\label{fig:algo_fail}
\end{figure*}

In this section we provide a small example of an input trace for $k\geq 3$ processes which has a predictable race that is not detected by $\RaceDecision$.
Consider the input trace $\Trace$ given in \cref{subfig:algo_fail_trace}, where the task is to decide whether $(\Event_4, \Event_{16})$ is a predictable race of $\Trace$.
To make the notation somewhat simple, given a variable $x$, if $x$ is not read, we denote every write event to $x$ by $\Write(x)$.
If $x$ is read, we denote by $\Read(x)$ the unique read event to $x$ by $\Write(x)$ the observation $\Observation_{\Trace}(\Read(x))$, and by $\ov{\Write}(x)$ any other write event to $x$.

We now outline the steps of $\RaceDecision$ on input the potential race $(\Event_4, \Event_{16})$.
Observe that the set $X$ constructed in \cref{line:construct_x} of the algorithm contains all events of $\Trace$, since $\Write(x_4)$ is read by $\Read(x_4)$ which belongs to $\SeqTrace_2$ and thus $\Write(x_4)$ is in the causal past cone of $\Event_{16}$.
Initially, the algorithm constructs a partial order shown in \cref{subfig:algo_fail_po} in solid edges.
Observe that this partial order is closed, hence the algorithm proceeds to make a nondeterministic choice for $i\in[2]$ in \cref{line:nondetermistic_choice}.
We argue that for $i=2$, the algorithm reports that $(\Event_4, \Event_{16})$ is not a predictable race of $\Trace$.
Indeed, in this case the algorithm will execute $\InsertAndClose(\Event_8 \to \Event_{11})$ in \cref{line:insert_and_close}, since $\Event_8<_{\Trace}\Event_{11}$.
This inserts the dotted edge in the partial order of \cref{subfig:algo_fail_po}.
Observe that this edge imposes the ordering $\Write(x_2)\to \ov{\Write}(x_2)$, hence by the rules of observation closure, the algorithm inserts the edge $\Read(x_2)\to \ov{\Write}(x_2)$, shown in dashed in \cref{subfig:algo_fail_po}.
However, this edge imposes the ordering $\ov{\Write}(x_1)\to \Read(x_1)$, hence by the rules of the observation closure, the algorithm inserts the edge 
$\ov{\Write}(x_2)\to \Write(x_1)$, shown in dashed in \cref{subfig:algo_fail_po}.
Observe that this edge creates a cycle in the partial order, hence for $i=2$, the algorithm reports that $(\Event_4, \Event_{16})$ is not a predictable race of $\Trace$.

On the other hand, \cref{subfig:algo_fail_witness} shows a correct reordering of $\Trace$ that exposes the race.
As a final remark, we note that the nondeterministic choice for $i\in [2]$ in \cref{line:nondetermistic_choice} of $\RaceDecision$ will also try $i=1$.
In this case, $\RaceDecision$ will detect the race and produce the witness trace shown in \cref{subfig:algo_fail_witness}.
It is not hard to extend this example so that the algorithm misses the race also for $i=1$.

\end{document}